\documentclass[prx,twocolumn,superscriptaddress,showpacs,notitlepage,floatfix,nofootinbib]{revtex4-2}
\usepackage{ulem} 
\usepackage{graphicx}
\usepackage[export]{adjustbox} 
\usepackage{braket}
\usepackage{mathrsfs}
\usepackage{extarrows} 
\usepackage{amsmath, amssymb, amsthm}
\usepackage{physics}
\usepackage{subcaption}
\usepackage{caption}
\usepackage{ragged2e}
\usepackage{xcolor}
\usepackage{enumerate}
\usepackage{bbm}
\usepackage{comment}
\usepackage[pdftex,colorlinks=true,linkcolor=violet,citecolor=blue,urlcolor=cyan]{hyperref}
\usepackage[capitalize]{cleveref}
\usepackage{booktabs}
\usepackage{multirow}
\usepackage{algorithm}
\usepackage{algpseudocode} % algorithmicx
\newcommand{\Bin}{\mathrm{Bin}}

\newtheorem{lemma}{Lemma}
\newtheorem{defn}{Definition}
\newtheorem{proposition}{Proposition}

\newtheorem{theorem}{Theorem}
\newtheorem{corollary}{Corollary}

\DeclareMathOperator{\ord}{ord}
\DeclareMathOperator{\supp}{supp}
\DeclareMathOperator{\decomp}{decomp}

\DeclareCaptionJustification{justified}{\justifying}
\usepackage[justification=justified,singlelinecheck=false]{caption}

\begin{document}

\title{Clifford-deformed zero-rate LDPC codes with 50\% biased noise thresholds}

\author{Jagannath Das}
\thanks{These authors contributed equally.}
\affiliation{Department of Physics, Virginia Tech, Blacksburg, Virginia 24061, USA}

\author{Sayandip Dhara}
\thanks{These authors contributed equally.}
\affiliation{Department of Physics, Virginia Tech, Blacksburg, Virginia 24061, USA}

\author{Pedro Medina}
\affiliation{Department of Physics, Virginia Tech, Blacksburg, Virginia 24061, USA}

\author{Arthur Pesah}
\affiliation{Dahlem Center for Complex Quantum Systems, Freie Universität Berlin, 14195 Berlin, Germany}

\author{Arpit Dua}
\email{adua@vt.edu}
\altaffiliation[Current affiliation: ]{QuEra Computing Inc., 1380 Soldiers Field Road, Boston, MA 02135}
\affiliation{Department of Physics, Virginia Tech, Blacksburg, Virginia 24061, USA}

\begin{abstract}
Applying single-qubit Clifford unitaries to a Pauli stabilizer code produces a Clifford-deformed variant whose stabilizers remain Pauli operators, but with locally rotated Pauli axes.
Such deformations provide a simple way to tailor a fixed code to anisotropic noise, and have enabled unusually high thresholds under strongly biased dephasing.
In this work, we discuss zero-rate quantum low-density parity-check (LDPC) codes, for which there exist Clifford-deformed variants where the number of biased logical operators scales slower than the distance, or there exists a basis of logical operators whose overlap satisfies certain scaling conditions; in this case, the code-capacity threshold for the Clifford-deformed variant under i.i.d.\ pure dephasing noise approaches $50\%$. This property provably explains previously known code examples with 50\% biased noise thresholds, such as XY surface code, XZZX surface code, color code, as well as some 3D Clifford-deformed codes. As a concrete new example, we study Clifford deformations of the tile codes of Ref.~\cite{steffan2025tile}. Similar to the phase diagram of 50\% thresholds for random Clifford deformations of the surface code in Ref.~\cite{dua_CDSC}, we find a similar phase diagram for the tile codes. 
We also construct several translationally invariant deformations of the tile code with 50\% thresholds, and present numerical evidence for improved performance at finite bias and under circuit-level noise.
In the circuit-level setting, performance is governed by the residual bias after a full syndrome-extraction cycle, linking our simulations to phenomenological models commonly used to study Clifford-deformed codes.
We estimate this residual bias for different qubit platforms by modeling microscopic implementations of tile-code syndrome extraction.
\end{abstract}

\maketitle
\tableofcontents
\section{Introduction}

Noise in realistic qubits is often strongly anisotropic, with dephasing processes dominating other error channels~\cite{Aliferis_2009, Nigg_2014, QDnoise2012, Wu_2022}.
This bias can be viewed as an opportunity: if the dominant error mechanism is known, codes and decoders can be tailored to exploit it and thereby increase the fault-tolerance threshold~\cite{shor1997faulttolerantquantumcomputation,shorPhysRevA.52.R2493,steane, Kitaev2003} and reduce overhead~\cite{Panos_Preskill, Poulinbiased, Brooksbiased, Ultrahigh2018, XZZX2021, Tuckett2020,srivastava2022xyzhexagonal,san2023cellular, Roffe2023biastailoredquantum,tiurev2024domain,setiawan2025tailoring,bias_preserving_CNOT_cat_qubits_2020, Bias_Cluster, Erasure_bias, Claes2023, Tailored_XZZX,Fazio2025logicalnoisebiasin,Fellous-Asiani2025,kusano2026spincatqubitbiasednoise,PhysRevLett.128.140503}.
In the extreme ``infinite-bias'' code-capacity limit, where only dephasing errors (Pauli $Z$ errors) occur independently on data qubits, $50\%$ is the fundamental information-theoretic threshold: above $p_Z=1/2$, reliable decoding is impossible, while approaching $1/2$ corresponds to near-capacity performance in this setting~\cite{shannon}.

% \AD{let's check if we missed any recent biased noise references among citations here}
% \SD{Added a few recent citations here}

A particularly simple and versatile bias-tailoring mechanism is the Clifford deformation.
Applying a tensor product of single-qubit Clifford unitaries to a Pauli stabilizer code maps its stabilizer group to another Pauli stabilizer group with identical parameters $[[n,k,d]]$, but with locally permuted Pauli axes.
Although this transformation is trivial in the absence of noise, it can substantially enhance the decoding capabilities under biased noise. 
Recent work has demonstrated that suitable deformations can yield dramatically improved thresholds for several topological codes, including thresholds approaching $50\%$ at infinite bias~\cite{Tailoring2019,Huang_2023,dua_CDSC,XZZX2021}.
Not only topological codes, but also in recent times, schemes have been proposed to construct bias-tailored qLDPC codes, such as the bias-tailored lifted product codes, \cite{Roffe2023biastailoredquantum}, and the Clifford deformed bivariate bicycle codes \cite{romanescocodes}. In both cases, the logical error rate has been observed to improve in the large-bias regime.  

However, to date, most $50\%$ results under pure dephasing have been obtained for low-dimensional topological codes.
At the same time, the more general quantum low-density parity-check (LDPC) codes, with constant-weight checks and distances growing with block length, have seen rapid progress and are strong candidates for reducing the asymptotic overhead of fault-tolerant computation~\cite{dinur2022goodquantumldpccodes, Panteleev2021degeneratequantum,panteleev2022asymptoticallygoodquantumlocally,leverrier2022quantumtannercodes}.
This raises a natural question: is near-capacity performance under biased dephasing an inherently topological phenomenon, or can it be engineered more broadly for quantum LDPC codes via local Clifford deformations?

% \paragraph*{Contributions.}
% We show that near-$50\%$ thresholds are not restricted to topological families. \AP{Do we show that? I think our families are still topological}

% \AP{"low-weight" in QEC usually means O(1). Maybe "minimal-weight"? Or just "pure Z" since BLO is really what matters in the end, not MWBLO?},

In this work, we derive general sufficient conditions under which a family of zero-rate quantum LDPC codes admits a Clifford-deformed variant with a $50\%$ infinite-bias code-capacity threshold under i.i.d.\ $Z$ noise. Our first result is a union-bound criterion showing that if every non-trivial pure-$Z$ logical operator has weight growing as $n^\alpha$ and the total number of such logical operators grows subexponentially in $n^\alpha$, then the logical failure probability decays as $\exp(-\Omega(n^\alpha))$ for all $p_Z<1/2$. This recovers the $50\%$ threshold mechanism in a form that is not restricted to previously studied topological examples.
We then introduce a basis of logical operators (BLO), namely an independent generating set for the pure-$Z$ commuting operators in which every basis element is itself a non-trivial pure-$Z$ logical operator. This allows us to sharpen the general counting argument. In particular, if there exists a non-overlapping BLO whose elements all have weight at least $K n^\alpha$, then a union bound over BLO elements alone already implies a $50\%$ threshold. More generally, if the BLO splits into a non-overlapping sector and an overlapping sector whose generated logical operators remain subexponential in number, then the same conclusion follows.

These results suggest a broader picture for Clifford-deformed LDPC codes under biased noise. In some families, the threshold follows directly from controlling the total number of pure-$Z$ logical operators. In others, one can exploit additional structure of the pure-$Z$ logical space, such as a non-overlapping BLO or, more generally, iterative weight-reduction mechanisms induced by macroscopic pure-$Z$ stabilizers, to reduce the decoding problem to one already covered by our threshold criteria. This provides a common framework encompassing some of the previously known $50\%$ biased-noise threshold examples, including XZZX, XY, and random Clifford-deformed surface codes~\cite{XZZX2021,Ultrahigh2018,dua_CDSC}.

As a concrete case study, we investigate the tile codes introduced in Ref.~\cite{steffan2025tile} (see also Refs.~\cite{rmy6-9n89,qv65-vmzr,breuckmann2025logicaloperatorsderivedautomorphisms}). For these codes, we consider both random local deformations and translationally invariant deformations, and numerically obtain a $50\%$ infinite-bias threshold. For one translationally invariant deformation with Hadamard gates on half of the qubits, we applied the weight-reduction technique ~\cite{Huang_2023} to analytically prove an infinite-bias threshold of $50\%$. For other deformations, including the random ones, we numerically demonstrate a threshold of 50\% using a belief-propagation+ordered statistics (BP-OSD) decoder. For all of these codes, we check the scaling of the number of elements in BLO and the scaling of the distance, and we find that it obeys the conditions required by our proof.

In order to assess the practical importance of these results, we use the BP-OSD decoder to numerically study the resulting Clifford-deformed tile codes under finite-bias code-capacity noise models, as well as circuit-level noise models, and obtain enhanced performance for the Clifford-deformed variants. We understand the enhancement in the circuit-level scenario by considering the residual bias at the end of the circuit, which maps the problem to the phenomenological setting. We calculate the residual biases for microscopic implementation of 
syndrome-extraction circuits in various hardware platforms, assessing how native gate sets and compilation strategies
affect bias preservation. This provides concrete evidence at the hardware level for considering LDPC codes tailored to biased noise.

\section{50\% biased noise threshold for zero-rate LDPC codes \label{sec: Clifford-deformed zero-rate LDPC codes}}

As mentioned in the introduction, several previous works~\cite{dua_CDSC,Huang_2023,san2023cellular} have demonstrated $50\%$ thresholds for Clifford-deformed topological codes under biased noise. In this section, we isolate a simple and general sufficient condition for such a threshold in the infinite-bias limit. The key observation is that if the number of biased-noise logical operators grows sufficiently slowly compared with the scaling of their weight, then maximum-likelihood decoding has a $50\%$ threshold under i.i.d.\ $Z$ noise. As mentioned, this formulation provides a general proof for several previously known examples, including the XZZX code, the XY surface code, the Clifford-deformed color code, and the random Clifford-deformed surface code studied in earlier works~\cite{Tailoring2019, dua_CDSC,Huang_2023,san2023cellular}.

\subsection{Setting}

We consider the code-capacity noise model with i.i.d.\ Pauli noise with dephasing bias. Each data qubit independently suffers an $X$, $Y$, or $Z$ error with probabilities $(p_X,p_Y,p_Z)$, where $p_Z$ may be much larger than $p_X$ and $p_Y$. The infinite-bias limit corresponds to $p_X=p_Y=0$ and $p_Z>0$. Unless stated otherwise, our analytical results are for the code-capacity model, while our numerics include both finite-bias code-capacity and circuit-level noise.

Let $\mathsf{P}_n$ be the $n$-qubit Pauli group and let $\mathsf{S}\subset\mathsf{P}_n$ be the stabilizer group of an $[[n,k]]$ Pauli stabilizer code. A Clifford deformation is specified by a choice of single-qubit Clifford unitaries $\{U_i\}_{i=1}^n$, with $U=\bigotimes_{i=1}^n U_i$. The deformed stabilizer group is
\begin{equation}
\mathsf{S}' \;=\; U\,\mathsf{S}\,U^\dagger .
\end{equation}
Because conjugation by single-qubit Cliffords maps Pauli operators to Pauli operators, $\mathsf{S}'$ defines another Pauli stabilizer code with the same code parameters. The benefit of deformation arises under biased noise: by permuting the Pauli axes on each qubit, the deformation alters how the dominant physical error aligns with stabilizers and logical operators, which can significantly affect decoding performance.

\subsection{Union bound over pure-$Z$ logical operators}

Consider a stabilizer code $\mathsf{C}$ on $n$ qubits with stabilizer group $\mathsf{S}$. Let $\mathsf{S}_Z$ denote the subgroup of pure-$Z$ stabilizers, and let
\begin{align}
    \mathsf{S}_Z &:= \mathsf{P}_n^Z \cap \mathsf{S}, \\
    \mathcal{L}_Z
    &:=
    \left( \mathsf{P}_n^Z \cap \mathcal{C}(\mathsf{S}) \right) \setminus \mathsf{S}_Z,
\end{align}
be the set of non-trivial pure-$Z$ logical operators, where $\mathsf{P}_n^Z$ is the group of $n$-qubit Pauli-$Z$ operators and $\mathcal{C}(\mathsf{S})$ denotes the centralizer of $\mathsf{S}$. Given an operator $A \in \mathsf{P}_n^Z$, we write $\supp(A)$ for its support. Given two operators $A,B \in \mathsf{P}_n^Z$, we write $A \cap B \in \mathsf{P}_n^Z$ for the $Z$ operator supported on $\supp(A) \cap \supp(B)$, and $|A|:=|\supp(A)|$ for the Hamming weight of $A$.

In the infinite-bias model, a Pauli-$Z$ error $E$ is sampled i.i.d.\ with rate $p_Z$ and decoded by a most-likely error decoder. A decoding failure occurs whenever there exists a non-trivial pure-$Z$ logical operator $L \in \mathcal{L}_Z$ such that $LE$ is at least as likely as $E$. Since the noise is i.i.d.\ and purely $Z$ type, the likelihood is determined by Hamming weight. For $p_Z<1/2$, the condition $P(LE)\ge P(E)$ is equivalent to
\begin{align}
    |LE| \le |E|.
\end{align}

For a fixed $L$, write the error as $E=E_L E_{\bar L}$ where $E_L=E \cap L$. Since multiplication by $L$ flips errors only on $\supp(L)$, the condition $|LE|\le |E|$ implies that at least half of the qubits in $\supp(L)$ must be in error,
\begin{align}
    |E_L| \ge \frac{|L|}{2}.
\end{align}
Equivalently,
\begin{align}
    |E\cap L| \ge \frac{|L|}{2}.
\end{align}

Applying a union bound over all pure-$Z$ logical operators gives
\begin{align}
    P_{\mathrm{fail}}
    \le
    \sum_{L \in \mathcal{L}_Z}
    \Pr\!\left(\,|E\cap L| \ge \frac{|L|}{2}\right).
    \label{eq:simple_union_bound}
\end{align}

\paragraph{Basis of logical operators}
% Consider a stabilizer code $\mathsf{C}$ on $n$ qubits with stabilizer group $\mathsf{S}$.
% Let $\mathsf{S}_Z$ denote the group of pure $Z$ stabilizers, and $\mathcal{L}_Z$ the set of all non-trivial pure-$Z$ logical operators of $\mathsf{C}$, namely
% \begin{align}
%     \mathsf{S}_Z &:= \mathsf{P}_n^Z \cap \mathsf{S} \\
%     \mathcal{L}_Z
%     &:=
%     \left( \mathsf{P}_n^Z \cap \mathcal{C}(\mathsf{S}) \right) \setminus \mathsf{S}_Z,
% \end{align}
% where $\mathsf{P}_n^Z$ is the group of $n$-qubit Pauli-$Z$ operators and $\mathcal{C}(\mathsf{S})$ denotes the centralizer of $\mathsf{S}$. 
We define the notion of \textit{basis of logical operators} (BLO) as an independent generating set $\mathcal{B}_Z$ of $\mathsf{P}_n^Z \cap \mathcal{C}(\mathsf{S})$, such that every element is non-trivial, i.e. belongs to $\mathcal{L}_Z$. If $\{ S_1^Z,\ldots,S_{n_z}^Z \}$ is basis of $\mathsf{S}_Z$ and $\{L^Z_1,\ldots,L^Z_k\}$ a set of representatives for $k$ independent non-trivial cosets of $\left( \mathsf{P}_n^Z \cap \mathcal{C}(\mathsf{S}) \right) / \mathsf{S}_Z$, we can construct such a BLO as 
\begin{align}
\mathcal{B}_Z= \{L^Z_{1},\ldots,L^Z_{k},
L^Z_1 S^Z_1,\ldots, L^Z_1 S^Z_{n_z} \}.
\label{eq:BLO}
\end{align}
% For any error $E \in \mathsf{P}_n^Z \cap \mathcal{C}(\mathsf{S})$, we write $\decomp(E)$ for the set of basis elements present in the the basis decomposition of $E$ every element of For 
A natural generating set for $\mathsf{P}_n^Z \cap \mathcal{C}(\mathsf{S})$ would be $\{L^Z_1,\ldots,L^Z_k,S^Z_1,\ldots,S^Z_{n_z}\}$, but the stabilizer elements $S^Z_j$ are trivial logical operators. Since we require every BLO element to be a non-trivial pure-$Z$ logical operator, we instead replace each $S^Z_j$ with the product $L^Z_1 S^Z_j \in \mathcal{L}_Z$, yielding the basis in Eq.~\eqref{eq:BLO}. This set still generates the same group, since $S^Z_j = (L^Z_1)^{-1}(L^Z_1 S^Z_j)$, but every element now belongs to $\mathcal{L}_Z$.
Any error $E\in \mathcal{L}_Z$ can be expressed as a non-empty product of elements of $\mathcal{B}_Z$. 
% It is convenient to introduce a basis of logical operators (BLO), because it determines the total number of pure-$Z$ logical operators and also because we will use it for our general theorem. Let $\mathcal{B}_Z$ be a basis for the quotient vector space
% \begin{align}
%     (\mathsf{P}_n^Z \cap \mathcal{C}(\mathsf{S})).
% \end{align}
% Each basis element can be represented by a non-trivial pure-$Z$ logical operator, 

Hence,
\begin{align}
    |\mathcal{L}_Z| = 2^{|\mathcal{B}_Z|}-2^{n_z}.
\end{align}
where $2^{n_z}$ is the number of stabilizers consisting of $Z$s and $I$s alone since $n_z$ is the number of pure $Z$ stabilizer generators. 

% \AP{Currently, $\mathcal{B}_Z$ is a basis of the quotient, meaning that $\mathcal{B}_Z=k$, while $\mathcal{L}_Z$ is the space of all representatives. So that equality doesn't work.}

\paragraph{Counting pure-$Z$ logical operators}

Let $H=(H_X\mid H_Z)$ be a binary check matrix for the Clifford-deformed stabilizer group in symplectic form. A $Z$-type Pauli operator with binary vector representation $z\in\mathbb{F}_2^n$ commutes with all stabilizer generators if and only if
\begin{align}
    H_X z = 0 \pmod 2.
\end{align}
Thus the vector space of pure-$Z$ operators commuting with the stabilizers is $\ker(H_X)$, with
\begin{align}
    \dim(\ker(H_X)) = n-\rank(H_X).
    \label{eq:dimBLO}
\end{align}
This determines the dimension of the pure-$Z$ operators, which is also the number of BLO elements i.e. $|\mathcal{B}_Z|$. 
% \AP{Which is just $k$ as defined, so the rest of the paragraph doesn't tell us much.}.

\subsection{50\% threshold bound}

We now prove the threshold bound. We first prove it using a union bound involving all biased noise logical operators and then later prove it using a union bound involving a non-overlapping basis of logical operators. The first bound is inspired by the union-bound argument as Lemma 1 in Ref.~\cite{san2023cellular}, used there for a biased-noise-tailored color code which has a single biased noise logical. Here, we note this generalizes to any family of Clifford-deformed stabilizer codes for which the number of non-trivial pure-$Z$ logical operators remains subexponential, that is, grows slower than the scaling of their weight. 
In this way, the result captures in a common framework the previously known $50\%$ threshold examples, including XZZX surface code, XY surface code, and certain three-dimensional Clifford-deformed constructions with $50\%$ biased-noise thresholds~\cite{dua_CDSC,Huang_2023,san2023cellular}. We also generalize this result by considering overlapping BLO but with constraints on the scaling of the assignments of overlaps to the participating logicals.

% \AP{Should this be a Theorem rather than a Lemma? That's how it is referred later, and it justifies the use of the term "Corollary" later. Or we relabel the Corollary as a Theorem.} 
% \AD{yeah, let's use theorem and make sure it is referred to as theorem later}\JD{It is changed accordingly}

\begin{theorem}[Threshold bound from the growth rate of pure-$Z$ logicals]
\label{thm:subexp_logicals_threshold}

Consider a family of Clifford-deformed stabilizer codes in the infinite-bias i.i.d.\ $Z$-noise model. Suppose that

\begin{itemize}
\item every non-trivial pure-$Z$ logical operator $L$ satisfies
\begin{align}
|L| \ge K n^\alpha
\end{align}
for constants $K>0$ and $\alpha>0$,

\item the total number of pure-$Z$ logical operators satisfies
\begin{align}
|\mathcal{L}_Z| \le \exp(o(n^\alpha)).
\end{align}
\end{itemize}

Then for any $p_Z<\tfrac12$, the logical failure probability of maximum-likelihood decoding satisfies
\begin{align}
P_{\mathrm{fail}}(n) \le \exp(-\Omega(n^\alpha)),
\end{align}
and therefore the code family has a $50\%$ infinite-bias threshold.

\end{theorem}

\begin{proof}

Fix $L\in\mathcal{L}_Z$. Under i.i.d.\ $Z$ noise, the number of errors on $\supp(L)$ is
\begin{align}
X_L \sim \mathrm{Bin}(|L|,p_Z).
\end{align}

A necessary condition for $L$ to cause decoding failure is
\begin{align}
X_L \ge \frac{|L|}{2}.
\end{align}

Since $p_Z<\tfrac12$, the threshold $|L|/2$ lies above the mean $p_Z|L|$. Applying a Chernoff–Hoeffding bound gives
\begin{align}
\Pr\!\left(X_L \ge \frac{|L|}{2}\right)
\le
\exp\!\left[-2|L|\left(\tfrac12-p_Z\right)^2\right].
\end{align}

Using $|L|\ge K n^\alpha$, we obtain
\begin{align}
\Pr\!\left(X_L \ge \frac{|L|}{2}\right)
\le
\exp(-\Omega(n^\alpha)).
\end{align}

Summing over all pure-$Z$ logical operators using Eq.~\eqref{eq:simple_union_bound} gives
\begin{align}
P_{\mathrm{fail}}
&\le
|\mathcal{L}_Z|\,\exp(-\Omega(n^\alpha)).
\end{align}

Since $|\mathcal{L}_Z| \le \exp(o(n^\alpha))$, the prefactor cannot overcome the exponential suppression. Therefore
\begin{align}
P_{\mathrm{fail}}(n)
\le
\exp(-\Omega(n^\alpha)),
\end{align}
which vanishes as $n\to\infty$ for every $p_Z<\tfrac12$.
\end{proof}

A particularly simple case occurs when the number of pure-$Z$ logical operators grows polynomially with system size. If
\begin{align}
|\mathcal{L}_Z| \le n^\beta
\end{align}
for some constant $\beta$, then the prefactor in the union bound is polynomial and therefore negligible compared with $\exp(\Omega(n^\alpha))$. In this case, the above theorem immediately implies a $50\%$ infinite-bias threshold.

\subsection{A bound from a non-overlapping basis of logical operators}

The previous lemma sums over all pure-$Z$ logical operators. When there exists a BLO whose supports are pairwise disjoint, one can replace this by a union bound over BLO elements only.

\begin{theorem}[Reduction to a non-overlapping BLO]
\label{lem:nonoverlap_BLO_bound}
Let $\mathcal{B}_Z$ be a BLO such that
\begin{align}
\supp(L\cap L')=\varnothing
\qquad
\text{for all distinct }L,L'\in\mathcal{B}_Z.
\end{align}
Then
\begin{align}
P_{\mathrm{fail}}
\le
\sum_{L\in\mathcal{B}_Z}
\Pr\!\left(\,|E\cap L|\ge \frac{|L|}{2}\right).
\label{eq:nonoverlap_BLO_union}
\end{align}
\end{theorem}

\begin{proof}
Let $E \in \mathsf{P}_n^Z$ be an error that causes decoding failure. Then there exists a non-trivial pure-$Z$ logical operator
\[
M=\prod_{L\in S} L
\]
for some nonempty subset $S\subseteq \mathcal{B}_Z$, such that
\[
|E\cap M| \ge \frac{|M|}{2}.
\]
Because the supports of the BLO elements are pairwise disjoint, we have
\[
\supp(M)=\bigsqcup_{L\in S}\supp(L)
\qquad\text{and}\qquad
|M|=\sum_{L\in S}|L|.
\]
Hence
\[
|E\cap M|
=
\sum_{L\in S}|E\cap L|.
\]
If for every $L\in S$ one had
\[
|E\cap L| < \frac{|L|}{2},
\]
then summing over $L\in S$ would give
\[
|E\cap M|
<
\frac{1}{2}\sum_{L\in S}|L|
=
\frac{|M|}{2},
\]
a contradiction. Therefore, there exists at least one $L\in S$ such that
\[
|E\cap L| \ge \frac{|L|}{2}.
\]
This proves that every failure event is contained in the union of the events on the right-hand side of Eq.~\eqref{eq:nonoverlap_BLO_union}. Taking probabilities and applying the union bound gives the claim of the lemma. 
\end{proof}

\begin{corollary}[Threshold bound from a non-overlapping BLO]
\label{cor:nonoverlap_BLO_threshold}
Consider a family of Clifford-deformed stabilizer codes in the infinite-bias i.i.d.\ $Z$-noise model. Suppose that for each blocklength $n$ there exists a non-overlapping BLO $\mathcal{B}_Z(n)$ such that every $L\in \mathcal{B}_Z(n)$ satisfies        $|L| \ge K n^\alpha$ for constants $K>0$ and $\alpha>0$, 
then for any $p_Z<\tfrac12$, the logical failure probability satisfies
\begin{align}
    P_{\mathrm{fail}}(n)\le \exp(-\Omega(n^\alpha)),
\end{align}
and hence the code family has a $50\%$ infinite-bias threshold.
\end{corollary}

\begin{proof}
By ~\cref{lem:nonoverlap_BLO_bound},
\begin{align}
P_{\mathrm{fail}}
\le
\sum_{L\in \mathcal{B}_Z(n)}
\Pr\!\left(|E\cap L|\ge \frac{|L|}{2}\right).
\end{align}
For each $L\in \mathcal{B}_Z(n)$, the random variable
\begin{align}
X_L := |E\cap L|
\end{align}
is distributed as $\mathrm{Bin}(|L|,p_Z)$. Since $p_Z<\tfrac12$, a Chernoff--Hoeffding bound gives
\begin{align}
\Pr\!\left(X_L\ge \frac{|L|}{2}\right)
\le
\exp\!\left[-2|L|\left(\tfrac12-p_Z\right)^2\right].
\end{align}
Using $|L|\ge K n^\alpha$, we obtain
\begin{align}
\Pr\!\left(X_L\ge \frac{|L|}{2}\right)
\le
\exp(-\Omega(n^\alpha)).
\end{align}
Therefore
\begin{align}
P_{\mathrm{fail}}
&\le
|\mathcal{B}_Z(n)|\,\exp(-\Omega(n^\alpha)).
\end{align}

Since the supports of the BLO elements are pairwise disjoint and each has weight at least $K n^\alpha$, we have
\begin{align}
    |\mathcal{B}_Z(n)| \le \frac{n}{K n^\alpha} = O(n^{1-\alpha}),
\end{align}
so the prefactor is at most polynomial and therefore negligible compared with the stretched-exponential suppression $\exp(-\Omega(n^\alpha))$ and hence,
\begin{align}
P_{\mathrm{fail}}(n)\le \exp(-\Omega(n^\alpha)).
\end{align}
This proves the claim.
\end{proof}

Thus, if a family admits a non-overlapping BLO whose elements all have weight at least $K n^\alpha$, then the same Chernoff--Hoeffding argument yields a $50\%$ infinite-bias threshold. This is precisely the mechanism realized in examples such as XZZX on a torus.

\paragraph*{Hybrid extension.}
The same reasoning extends to a mixed situation in which only part of the BLO is non-overlapping. Suppose the BLO decomposes as
\[
\mathcal{B}_Z(n)=\mathcal{B}_Z^{\mathrm{dis}}(n)\sqcup \mathcal{B}_Z^{\mathrm{ov}}(n),
\]
where the elements of $\mathcal{B}_Z^{\mathrm{dis}}(n)$ are pairwise disjoint, while overlaps are allowed within $\mathcal{B}_Z^{\mathrm{ov}}(n)$. Then one may bound decoding failure by a sum of two contributions: a union bound over basis elements in the disjoint sector, and a union bound over all pure-$Z$ logical operators generated by the overlapping sector. If every basis element has weight at least $K n^\alpha$, if the number of logical operators generated by $\mathcal{B}_Z^{\mathrm{ov}}(n)$ satisfies
\[
|\mathcal{L}_Z^{\mathrm{ov}}(n)|\le \exp(o(n^\alpha)),
\]
then considering $|\mathcal{B}_Z^{\mathrm{dis}}(n)|=O(n)$, the contributions from both parts of $|\mathcal{B}_Z(n)$ are bounded by $\exp(-\Omega(n^\alpha))$, and hence
\[
P_{\mathrm{fail}}(n)\le \exp(-\Omega(n^\alpha)).
\]
Thus, a $50\%$ infinite-bias threshold still follows provided the overlapping sector generates only a subexponential number of logical operators.

A more general extension, in which overlapping BLO elements are handled by assigning overlap penalties across basis elements, is deferred to ~\cref{app:overlap_covering}.

\paragraph*{Remark on the non-zero-rate case}

The above sufficient conditions are not expected to hold for non-zero-rate families. Indeed, if the number of encoded qubits scales extensively, $k=\Theta(n)$, and hence contributes $k=\Theta(n)$ independent pure-$Z$ logical operators, then the total number of non-trivial pure-$Z$ logical operators grows exponentially,
\begin{align}
    |\mathcal{L}_Z|=\exp(\Theta(n)),
\end{align}
so the subexponential-growth assumption in ~\cref{thm:subexp_logicals_threshold} fails.

Moreover, the non-overlapping BLO condition is also incompatible with a non-zero rate. If $\mathcal{B}_Z$ is a non-overlapping BLO, then the supports of its elements are pairwise disjoint, so
\begin{align}
    \sum_{L\in\mathcal{B}_Z}|L| \le n.
\end{align}
If, in addition, every BLO element has weight at least $K n^\alpha$, then
\begin{align}
    |\mathcal{B}_Z| \le \frac{n}{K n^\alpha}=\mathcal{O}(n^{1-\alpha}).
\end{align}
Thus, to accommodate an extensive number of independent pure-$Z$ logical directions, one would need $\alpha=0$, i.e.\ BLO elements of bounded weight, which is incompatible with the distance growth required for a $50\%$ threshold. Equivalently, a non-zero-rate family cannot simultaneously have an extensive pure-$Z$ logical space and a non-overlapping BLO whose elements all have growing weight.

Therefore, both threshold mechanisms discussed above are naturally zero-rate mechanisms: non-zero-rate families are not expected to satisfy either the subexponential bound on the total number of pure-$Z$ logical operators or the existence of a non-overlapping BLO with sufficiently sparse support growth.
This is consistent with the capacity of the binary symmetric channel: at rate $R>0$, the Shannon limit implies that reliable communication is impossible once the crossover probability exceeds $h^{-1}(1-R)<\tfrac{1}{2}$, where $h$ is the binary entropy function~\cite{shannon}. In particular, a family with $k=\Theta(n)$ cannot achieve a $50\%$ threshold under i.i.d.\ bit-flip (or $Z$) noise.

\section{Tile codes and Clifford deformation}
\label{sec: Clifford Deformed Tile code}

In ~\cref{sec: Clifford-deformed zero-rate LDPC codes}, we derived sufficient conditions for a $50\%$ threshold for Clifford-deformed LDPC codes in the infinite-bias limit. In this section, we study concrete families of tile codes under biased Pauli noise and evaluate their decoding performance through numerical simulations and an analysis of their pure-$Z$ basis logical operators (BLOs).

Our primary focus is on Clifford-deformed variants of the planar, geometrically local tile codes introduced in Refs.~\cite{steffan2025tile,qv65-vmzr}. These codes have bounded-weight local stabilizers and exhibit favorable efficiency metrics $kd^2/n$ compared to the surface code. Their planar layout and local checks also make them naturally compatible with planar hardware platforms such as superconducting architectures. We estimate thresholds at infinite bias and examine whether the corresponding BLOs satisfy the structural hypotheses of ~\cref{thm:subexp_logicals_threshold}, providing evidence for $50\%$ thresholds in the infinite-bias setting for suitable deformations.

The remainder of this section is organized as follows. In ~\cref{Tile Code description}, we define the tile-code families considered here, including both open and periodic boundary conditions. ~\cref{subsec: phase_diagram} presents the phase diagram of randomly Clifford-deformed periodic tile codes. Translationally invariant (TI) deformations and their decoding performance are summarized in ~\cref{TI codes}, where we also describe a weight-reduction procedure for a refined periodic subfamily of linearly deformed tile codes. Finally, ~\cref{sec:mwblo} studies BLO structure in support of ~\cref{thm:subexp_logicals_threshold}.

\subsection{Tile code families}
\label{Tile Code description}

To evaluate thresholds, we use a family of tile codes with increasing distance and fixed logical dimension. The tile codes in Refs.~\cite{steffan2025tile,qv65-vmzr} were introduced as fully planar open-boundary qLDPC codes with geometrically local weight-6 stabilizers and code families with logical dimensions $6 \le k \le 13$.

Here, we focus on a family with fixed logical dimension $k=8$ and increasing distance. An example is the $[[288,8,12]]$ open-boundary tile code shown in ~\cref{fig:tile-code-all}. In this representation, vertices correspond to stabilizer generators and edges correspond to physical qubits. Bulk stabilizers have weight six and are supported within a $3\times 3$ tile pattern; boundary stabilizers arise from truncations of this bulk pattern. Among the 16 equivalent bulk configurations identified in Ref.~\cite{qv65-vmzr}, each yields the same code parameters $[[288,8,12]]$. For the stabilizer set used here, the checks are independent, so the number of encoded qubits is
\begin{equation}
k = n - \mathrm{rank}(S)=8.
\end{equation}

Extending the bulk pattern to an $L\times L$ lattice yields the following open-boundary family:
\begin{align*}
    L=6:&\quad [[72,8,4]],\quad kd^2/n = 1.78,\\
    L=8:&\quad [[128,8,6]],\quad kd^2/n = 2.25,\\
    L=10:&\quad [[200,8,9]],\quad kd^2/n = 3.24,\\
    L=12:&\quad [[288,8,12]],\quad kd^2/n = 4.00,\\
    L=13:&\quad [[338,8,13]],\quad kd^2/n = 4.00,\\
    L=14:&\quad [[392,8,15]],\quad kd^2/n = 4.59.
\end{align*}
Each code in this family has
\begin{equation}
    n = 2L^2, \qquad k = 8,
\end{equation}
and is LDPC since stabilizer weights remain bounded independently of $n$.

In~\cref{subsec: phase_diagram}, to mitigate finite-size boundary effects in threshold estimation, we also consider periodic boundary conditions with the same bulk stabilizers. In this periodic setting, lattice dimensions are restricted to $L_x, L_y\in 7\mathbb{Z}$, and the code encodes six logical qubits (as opposed to eight in the open-boundary case). We use these periodic codes to study random single-qubit Clifford deformations and to map out the corresponding phase diagram. For the TI deformations in ~\cref{TI codes}, we return to the open-boundary family above. For the weight-reduction study, we use a refined periodic subfamily specified in ~\cref{app:weight-reduction}.

\begin{figure*}[t]
    \centering
    \begin{subfigure}[b]{0.48\textwidth}
        \centering
        \includegraphics[width=\linewidth,height=0.6\linewidth,keepaspectratio]{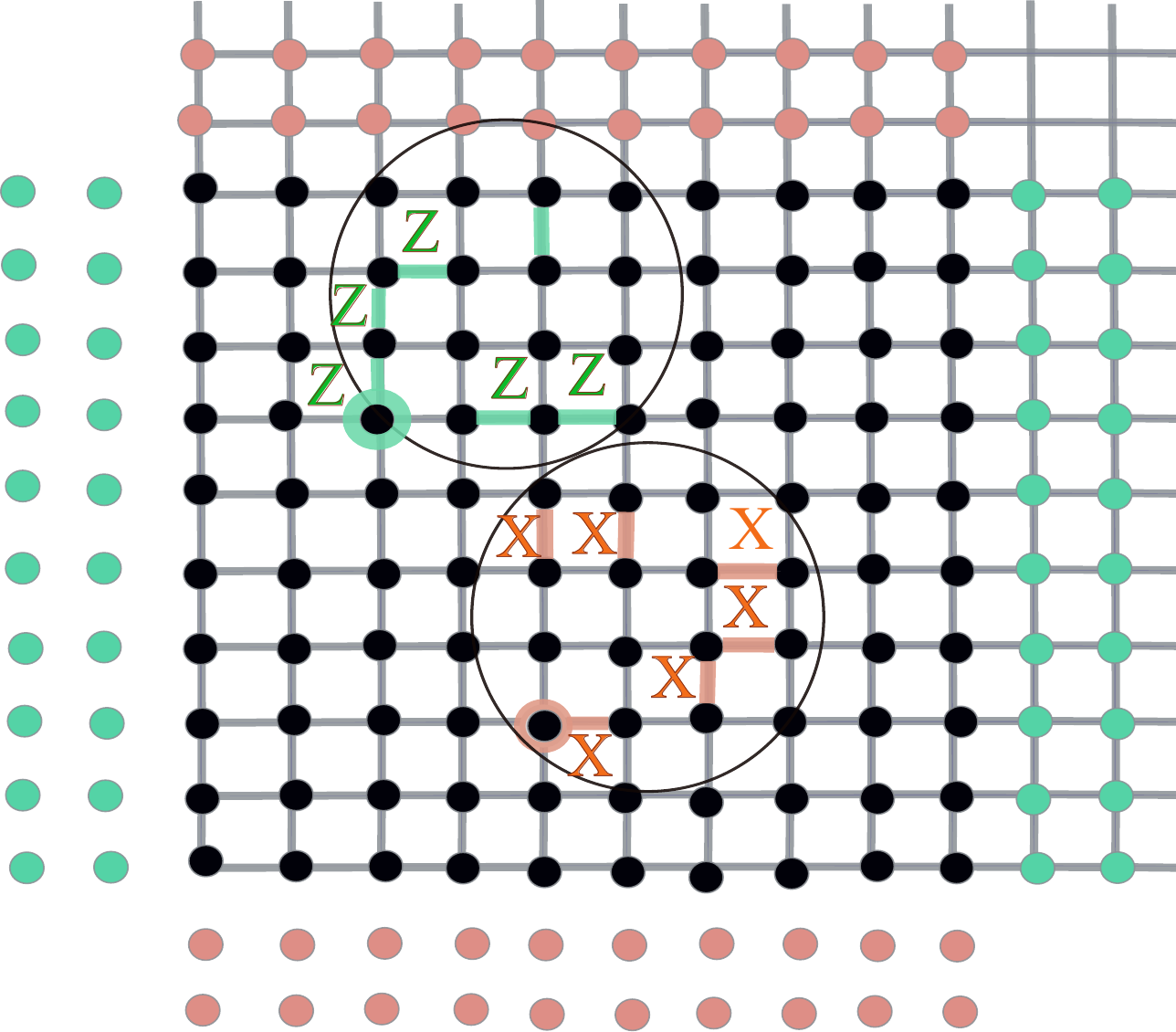}
        \captionsetup{justification=raggedright,singlelinecheck=false}
        \caption{Bulk stabilizer generators ($X$ and $Z$ type) of the $[[288,8,12]]$ tile code.}
        \label{fig:bulk}
    \end{subfigure}\hfill
    \begin{subfigure}[b]{0.48\textwidth}
        % \centering
        \includegraphics[width=\linewidth,height=0.6\linewidth,keepaspectratio]{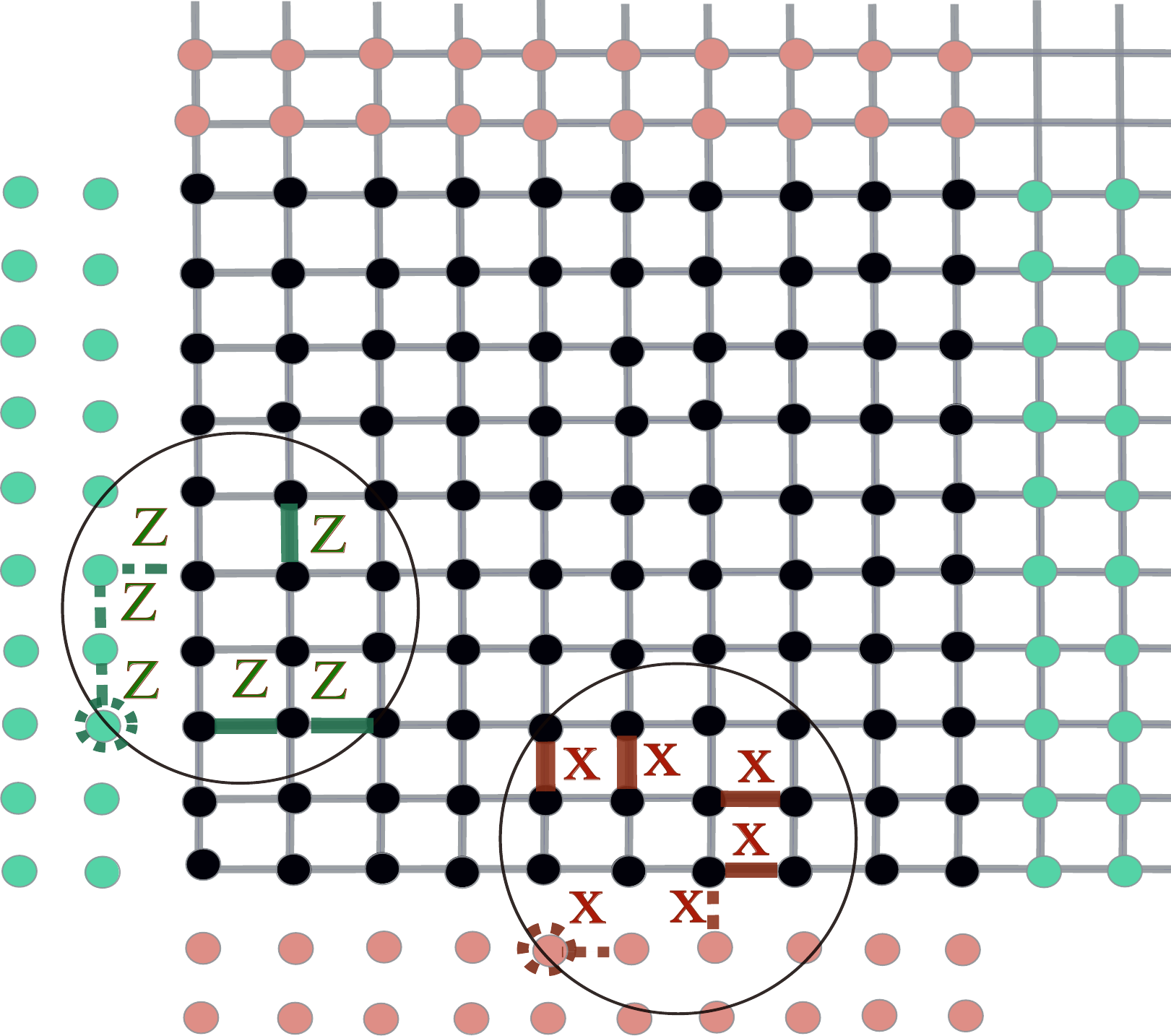}
        \captionsetup{justification=raggedright,singlelinecheck=false}
        \caption{Boundary stabilizer generators ($X$ and $Z$ type) of the $[[288,8,12]]$ tile code.}
        \label{fig:boundary}
    \end{subfigure}

    \vspace{0.3cm}

    \begin{subfigure}[b]{0.48\textwidth}
        \centering
       \includegraphics[width=\linewidth,height=0.6\linewidth,keepaspectratio]{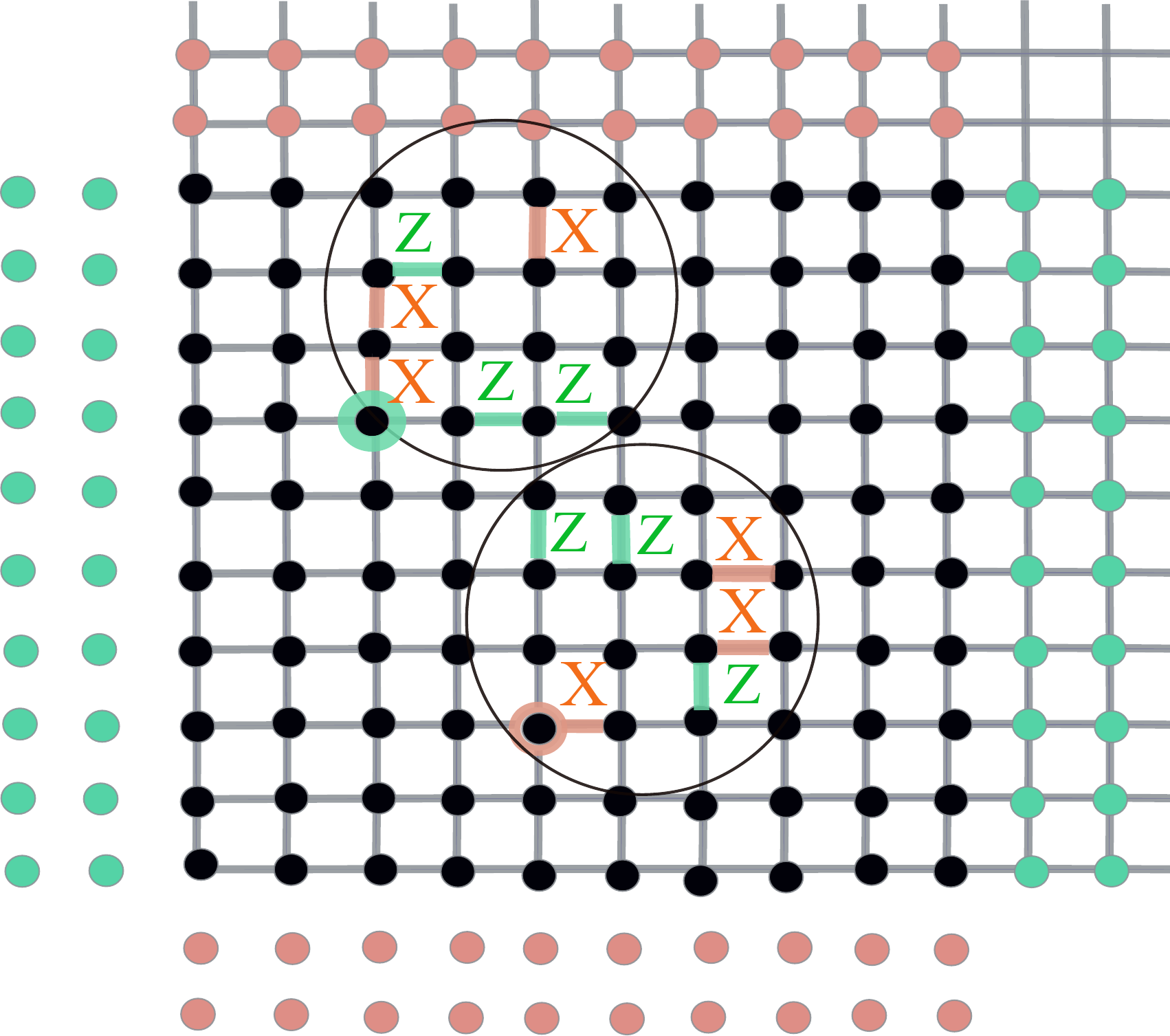}
       \captionsetup{justification=raggedright,singlelinecheck=false}
        \caption{Linear-deformed open tile code, obtained by applying $H$ on all vertical bonds.}
        \label{fig:linear-deformation}
    \end{subfigure}\hfill
    \begin{subfigure}[b]{0.48\textwidth}
        \centering
        \includegraphics[width=\linewidth,height=0.6\linewidth,keepaspectratio]{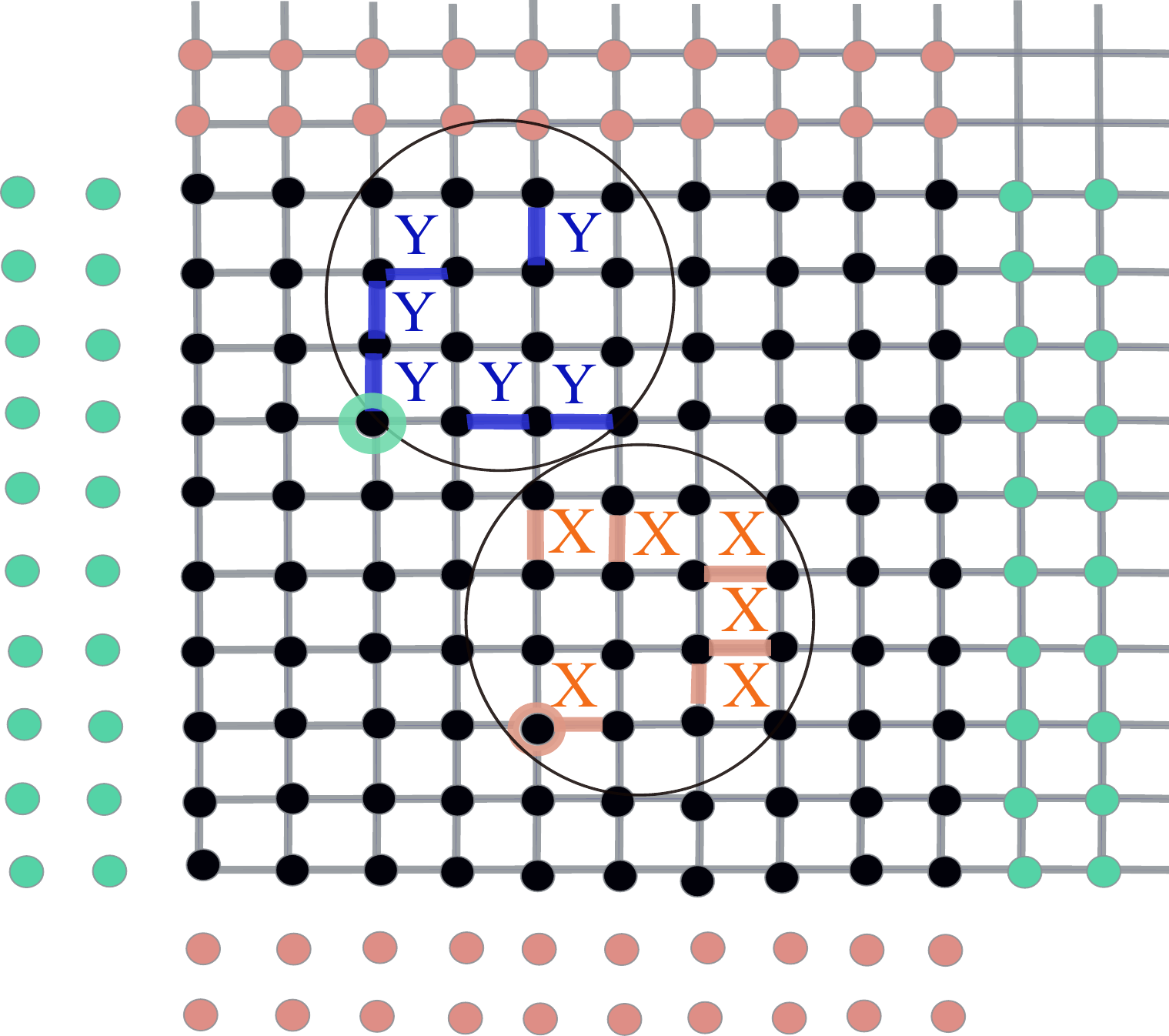}
        \captionsetup{justification=raggedright,singlelinecheck=false}
        \caption{$XY$-deformed open tile code, obtained by applying $SH$ on all qubits.}
        \label{fig:xy-deformation}
    \end{subfigure}

    \captionsetup{justification=justified,singlelinecheck=false}
    \caption{Stabilizers of the $[[288,8,12]]$ tile code and its Clifford-deformed variants. In the undeformed tile code, qubits reside on the edges of the lattice, and dots denote stabilizer generators. Boundary dots correspond to single-type stabilizers: mint-green dots denote $Z$-type stabilizers and orange dots denote $X$-type stabilizers. Bulk black dots indicate locations where both $X$- and $Z$-type stabilizers are present. The support of each stabilizer on nearby qubits is illustrated pictorially. In the bulk, stabilizer generators have weight six, acting on three vertical and three horizontal qubits, as in (a). At the boundaries, truncated geometry produces stabilizers of reduced weight (2, 3, or 4); examples of a weight-3 $Z$-type stabilizer and a weight-4 $X$-type stabilizer are shown in (b). Panels (c) and (d) show bulk stabilizers of Clifford-deformed tile codes with $n=288$ and open boundary: (c) a linear deformation obtained by applying $H$ on all vertical bonds, and (d) an $XY$ deformation obtained by applying $SH$ on all qubits.}
    \label{fig:tile-code-all}
\end{figure*}

\subsection{Random Clifford-deformed tile codes}
\label{subsec: phase_diagram}
In Ref.~\cite{dua_CDSC}, random Clifford-deformed surface codes were studied, and a phase diagram of thresholds for typical realizations was obtained. A central observation is that, under biased noise, code performance can be optimized by tuning the rates of Pauli-permutation deformations, specifically the $X\leftrightarrow Z$ and $Y\leftrightarrow Z$ exchanges. The mechanism is transparent at the level of the decoding graph: a CSS code under pure $Z$ noise has a decoding graph with two disconnected components, and $Y$-type deformations couple these components, converting what would otherwise be independent decoding problems into a single connected one. Crucially, this mechanism is a property of the Clifford-deformation structure itself under biased noise, not just of the surface-code geometry, and should therefore extend to any CSS code.
 
In this work, we test this idea on LDPC codes precisely. We construct random Clifford-deformed tile codes (CDTCs) by starting from a CSS tile code and independently applying, to each qubit, one of the single-qubit Clifford operators
\[
H,\qquad HSH,\qquad I
\]
with probabilities $\Pi_{XZ}$, $\Pi_{YZ}$, and $1-\Pi_{XZ}-\Pi_{YZ}$, respectively, where $S=\sqrt{Z}$.
Conjugation by $H$ swaps $X\leftrightarrow Z$, while conjugation by $HSH$ swaps $Y\leftrightarrow Z$ (up to phases). Thus $(\Pi_{XZ},\Pi_{YZ})$ specifies the probabilities of local Pauli-axis permutations.
We explore the full two-parameter phase space
\begin{equation}
\Pi_{XZ}\ge 0,\qquad \Pi_{YZ}\ge 0,\qquad \Pi_{XZ}+\Pi_{YZ}\le 1.
\end{equation}

\begin{figure}[t]
\centering
\includegraphics[width=1\columnwidth]{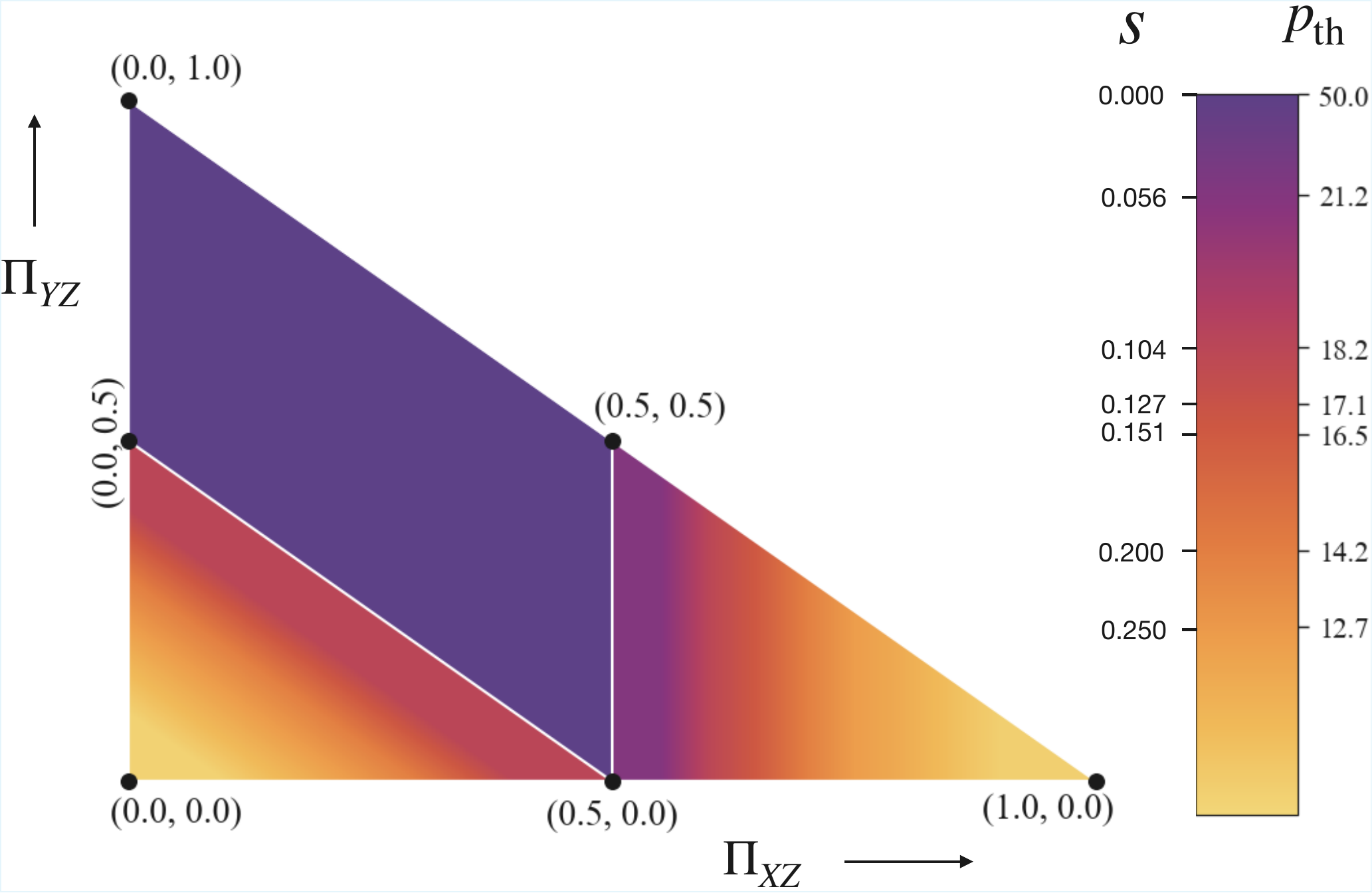}
\captionsetup{justification=justified,singlelinecheck=false}
\caption{Phase diagram of periodic tile codes is shown under random Clifford deformation at infinite bias ($\eta\to\infty$). The axes $(\Pi_{XZ},\Pi_{YZ})$ specify the probabilities of applying local Clifford permutations that swap $X\leftrightarrow Z$ and $Y\leftrightarrow Z$, respectively. Colors indicate regions with different estimated thresholds. The rhombic region bounded by the phase points $(0.5,0)$, $(0.5,0.5)$, $(0,1)$, and $(0,0.5)$ exhibits a $50\%$ threshold. Along the white boundary, we do not observe a sharp crossing in the threshold plot for most of the phase points. Moving away from this regime, the threshold decreases monotonically. Iso-threshold contours appear on the two sides of the rhombic regime, forming vertical lines on the right side and diagonal lines on the left side. The slope $s$ in the scaling relation $|\mathcal{B}_Z(n)| = sn + c$ characterizes the growth of basis logical operators with system size: larger $s$ corresponds to a lower threshold, while the $50\%$ regime corresponds to $s=0$.} 
\label{fig:phase_diagram}
\end{figure}

For this phase-diagram study, we use periodic boundary conditions and restrict to the infinite-bias limit ($\eta\to\infty$). In the $(\Pi_{XZ},\Pi_{YZ})$ plane, the point $(0,0)$ corresponds to the undeformed CSS code. The line $\Pi_{YZ}=0$ corresponds to ensembles with only $X\leftrightarrow Z$ deformations applied with probability $\Pi_{XZ}$. The point $(0,1)$ corresponds to a fixed deformation where every physical $Z$ operator is mapped to $Y$ (up to phases).

We estimate thresholds by averaging logical error rates over both deformation realizations and sampled error configurations for each point $(\Pi_{XZ},\Pi_{YZ})$. Decoding is performed using belief propagation with ordered statistics decoding (BP-OSD)~\cite{Poulin_BP,roffe_decoding_2020} in the infinite-bias setting. We find a connected region of the phase diagram with an estimated $50\%$ threshold. Empirically, this region is bounded by
\begin{equation}
0.5 < \Pi_{XZ}+\Pi_{YZ} < 1,
\qquad
\Pi_{XZ}\le 0.5,
\end{equation}
as shown in ~\cref{fig:phase_diagram}. This qualitative structure closely parallels the phase diagram obtained for random Clifford-deformed surface codes in Ref.~\cite{dua_CDSC}.

For phase points near the boundary of the $50\%$ region, such as $(0.5,0)$ and $(0,1)$, threshold extraction from BP-OSD is less conclusive for the specific periodic tile-code sizes considered here, and we do not observe a sharp crossing. Representative threshold estimates for selected boundary points are provided in ~\cref{app:phase-diagram}.

Moving away from the $50\%$ region, the estimated threshold decreases monotonically. We also observe distinct iso-threshold regimes on either side of the $50\%$ region. On the left side, the iso-threshold contours are approximately described by
\begin{equation}
\Pi_{XZ}+\Pi_{YZ}=\text{constant},
\end{equation}
while on the right side, they are approximately described by
\begin{equation}
\Pi_{XZ}=\text{constant}.
\end{equation}
We verify this behavior by comparing representative phase points along the diagonal and vertical white guide lines in ~\cref{fig:phase_diagram}. The origin of these iso-threshold contours, and their relation to the structure of pure-$Z$ logical operators, is analyzed in ~\cref{sec:mwblo}.

\subsection{Translational-invariant CDTCs}
\label{TI codes}

Random Clifford deformations characterize typical performance through ensemble averaging, but individual realizations can be atypical and can induce markedly different logical-operator structure. Since the geometry and overlap structure of pure-$Z$ logical operators depend on the deformation pattern, it is useful to study structured realizations where symmetry is enforced by construction. We therefore analyze translation-invariant (TI) Clifford-deformed tile codes with open boundaries.

We consider three TI deformations: (i) linear, (ii) $XY$, and (iii) a TI realization associated with the phase point $(0.25,0.5)$.

\paragraph{Linear deformation.}
In the linear deformation, we apply a Hadamard gate $H$ to every vertical qubit (the vertical bonds in Fig.~\ref{fig:linear-deformation}). This deformation preserves translational invariance. After deformation, each bulk stabilizer generator contains three $X$ and three $Z$ Paulis.
\begin{figure}[t]
\centering
\includegraphics[width=0.9\columnwidth]{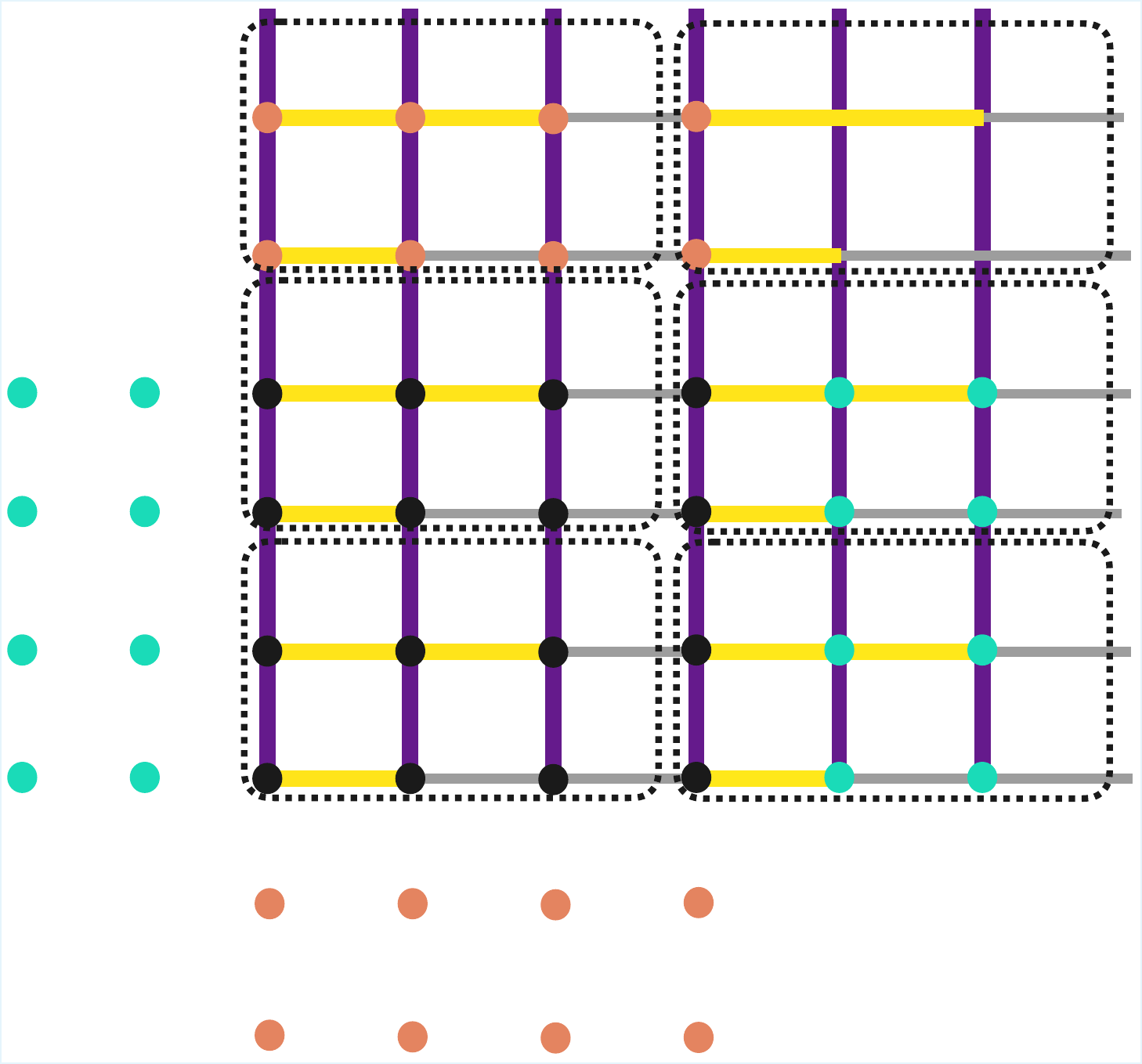}
\captionsetup{justification=justified,singlelinecheck=false}
\caption{Translation-invariant realization associated with the phase point $(0.25,0.5)$ for the open tile code $[[72,8,4]]$. The dotted rectangle denotes a unit cell. The color of any edge reflects the deformation applied to the qubit on that edge. On dark purple (yellow) qubits, we apply $YZ$ ($XZ$) deformations, while gray qubits remain undeformed.}
\label{fig:TI_middle_unit_cell}
\end{figure}

\paragraph{$XY$ deformation.}
In the $XY$ tile code, we apply the single-qubit unitary $HSH$ to every qubit (Fig.~\ref{fig:xy-deformation}). This maps $Z$-type stabilizers to $Y$-type stabilizers (up to phases) while preserving the locality pattern.

\paragraph{TI realization associated with $(0.25,0.5)$.}
We also study a TI deformation whose unit cell is shown in Fig.~\ref{fig:TI_middle_unit_cell}. The unit cell highlighted by the dotted rectangle is chosen so that, within it, the local deformations match the probabilities $(\Pi_{XZ},\Pi_{YZ})=(0.25,0.5)$. On qubits colored dark purple (yellow), we apply $Y\leftrightarrow Z$ ($X\leftrightarrow Z$) deformations, while gray qubits remain undeformed.

All three TI deformations show threshold behavior consistent with a $50\%$ threshold at infinite bias for the code sizes studied here (see Fig.~\ref{fig:ti} in the appendix), in contrast to the $\approx 10\%$ threshold of the undeformed open CSS tile-code family. This is consistent with prior work on Clifford-deformed surface codes, where suitable deformations can induce additional structure that improves decoding performance under biased noise~\cite{dua_CDSC}.

For the linear deformation, we can also prove analytically the existence of a family of periodic lattices for which the code has a $50\%$ threshold at infinite bias. The full proof is given in Appendix~\ref{app:weight-reduction}. It is based on the weight-reduction method~\cite{Huang_2023,tiurev2024domain}, which constructs a recursive decoding strategy in terms of repetition-code subproblems with decreasing check weight at each step. We show that this method applies to a subset of $L\times L$ periodic lattices satisfying explicit number-theoretic constraints stated in the appendix.

\subsection{Basis of logical operators for CDTCs} \label{sec:mwblo}
The proof in \cref{sec: Clifford-deformed zero-rate LDPC codes} shows that the scaling, distance, and cardinality of a generating set of pure-$Z$ logical operators control whether a $50\%$ threshold is possible at infinite bias.
% This motivates a direct examination of the basis of pure-$Z$ logical operators (MWBLO) for CDTCs.
In what follows, we study BLOs in both random and translation-invariant deformations, focusing mainly on two quantities: the number of logical operators in the basis and the $Z$ distance (minimum weight of a nontrivial pure-$Z$ logical).
This analysis supports the observed threshold behavior and the iso-threshold contours in Fig.~\ref{fig:phase_diagram}, and it provides a concrete diagnostic for when the conditions of Theorem~\ref{thm:subexp_logicals_threshold} are realized.

Following this perspective, it is also important to understand how the size of the BLO affects the overlap structure among its generators. In particular, we observe that for the deformed tile codes, the basis size scales linearly with the system size,
\[
|\mathcal{B}_Z(n)| = s n + c .
\]
where c is a constant. Note that in general, the basis size can scale sub-linearly; for example, in the XZZX surface code, it scales as $~n^{1/2}$. 
For ~\cref{thm:subexp_logicals_threshold} to ensure a 50\% threshold, it is necessary that the total number of pure-$Z$ logical operators grows at most sub-exponentially with system size. In the case where the slope $s=0$, both the number of logical generators and the total number of pure-$Z$ logical operators remain constant as the system size increases. Under these conditions, ~\cref{thm:subexp_logicals_threshold} guarantees a 50\% threshold.

In contrast, when the slope $s>0$, the number of independent logical generators increases with system size. This leads to an exponential growth in the total number of pure-$Z$ logical operators. As a consequence, the condition of sub-exponential growth is violated, and the 50\% threshold is no longer achieved.
Also, as additional generators are introduced, their supports are expected to intersect more frequently, leading to larger overlaps among the basis logical operators. Physically, this implies that a single physical error can simultaneously contribute to multiple logical explanations. Consequently, the decoder faces an enlarged set of competing logical hypotheses capable of explaining the observed syndrome. 

This proliferation of competing logical explanations increases the number of error configurations that can produce decoding ambiguity. In effect, logical failure can occur with fewer errors on average, thereby reducing the amount of noise the decoder can tolerate. As a result, the conditions required for ~\cref{thm:subexp_logicals_threshold} are no longer satisfied, and the proof guaranteeing a $50\%$ threshold does not apply in this regime. Consistent with this expectation, we observe numerically that the threshold decreases as the slope $s$ in $|\mathcal{B}_Z(n)| = sn + c$ increases (See Fig. \ref{fig:slope_vs_threshold}).

\subsubsection{Translation-invariant deformations}

For every translational-invariant deformation of the tile code, we have found that the number of elements in the BLO is constant at $8$, matching the $8$ logical qubits of the underlying open CSS tile-code family. That implies slope $s\approx 0.$  In contrast, the $Z$-distance (i.e., the weight of the minimum-weight pure-$Z$ logical) increases with system size.  A summary of MWBLO data for the TI-deformed tile codes is provided in  \cref{app:mwbo} of the Appendix.

Since for all translation-invariant deformations studied here, $|\mathcal{B}_Z(n)|$ is constant in system size, and the observed distance scaling is consistent with $0<\alpha\le 1;$ a 50\% threshold is expected from ~\cref{thm:subexp_logicals_threshold}. 

\subsubsection{Random deformations}
We now summarize the BLO results for random Clifford-deformed tile codes with periodic boundary conditions across the phase space in ~\cref{fig:phase_diagram}.

\paragraph{Inside the $50\%$ threshold regime.}
We consider representative phase points
\begin{multline*}
(\Pi_{XZ},\Pi_{YZ})\in\{(0.25,0.5),\ (0.1875,0.635),\\ (0.3125,0.375)\},
\end{multline*}
which lie in the region with an estimated $50\%$ threshold.
For these points, the BLO cardinality is essentially constant across system sizes (i.e slope $s\approx 0$) and equals the number of logical qubits, which is $6$ for the corresponding periodic CSS tile codes.
Moreover, the weight of  the minimum-weight pure-$Z$ logical grows approximately linearly with system size, consistent with  $0<\alpha\le 1.$
Near the boundary of the $50\%$ region, the BLO size begins to increase very slowly with system size, although the threshold remains close to $50\%$ for the sizes accessible in our simulations.
For most of the phase points on the boundary where BP-OSD does not yield a sharp threshold crossing, the extracted BLO size and $Z$ distance do not show systematic finite-size behavior.

\paragraph{Outside the $50\%$ threshold regime.}
Moving away from the $50\%$ region along the $\Pi_{XZ}$ axis, the threshold decreases on both sides, with the endpoints $(0,0)$ (undeformed CSS) and $(1,0)$ attaining the lowest thresholds.
The BLO results provide a consistent explanation: as the number of elements in the BLO starts to grow significantly with system size, the ~\cref{thm:subexp_logicals_threshold} no longer holds.
Equivalently, the code admits an increasingly dense collection of distinct near-minimum-weight pure-$Z$ logical representatives, which produces many competing logical error pathways and weakens the union-bound mechanism of ~\cref{sec: Clifford-deformed zero-rate LDPC codes}.

In this case of increasing $|\mathcal{B}_Z(n)|,$, we study how the threshold gets affected by the slope of the linear scaling of $|\mathcal{B}_Z(n)|.$  We illustrate this by examining four phase points
% \begin{equation*}
% (\Pi_{XZ},\Pi_{YZ})\in
% \{(0.25,0.5),\ (0.55,0.125),\ (0.125,0.25),\ (0.125,0.125)\},
% \end{equation*}
\begin{multline*}
(\Pi_{XZ},\Pi_{YZ}) \in
\{(0.25,0.5),\ (0.55,0.125), \\
 (0.125,0.25),\ (0.125,0.125)\},
\end{multline*}
with progressively decreasing thresholds (See ~\cref{fig:phase_diagram} for details). Our numerical analysis shows that the threshold drops sharply as the BLO size increases significantly with $n$  (See ~\cref{fig:support_for_threshold_decrease_final}a in Appendix), consistent with a violation of the ~\cref{thm:subexp_logicals_threshold}.  We explicitly show how the threshold decreases with the increase of the slope $s$ in \cref{fig:slope_vs_threshold}. 

\paragraph{Scaling across iso-threshold contours.}
~\cref{fig:phase_diagram} also exhibits iso-threshold contours on either side of the $50\%$ region.
On the left, the contours are approximately given by $\Pi_{XZ}+\Pi_{YZ}=\mathrm{const}$, while on the right they are approximately given by $\Pi_{XZ}=\mathrm{const}$.
Although we do not yet have an analytical proof, the BLO data corroborate this organizing principle. In particular, the phase points within the same contour have essentially identical BLO scaling with system size i.e., the slope $s$ (See ~\cref{fig:support_for_threshold_decrease_final}b in Appendix).
Different contours correspond to different $s$, and $s$ increases as the threshold decreases.
\begin{figure}[t!]
\centering
\includegraphics[width=1\columnwidth]{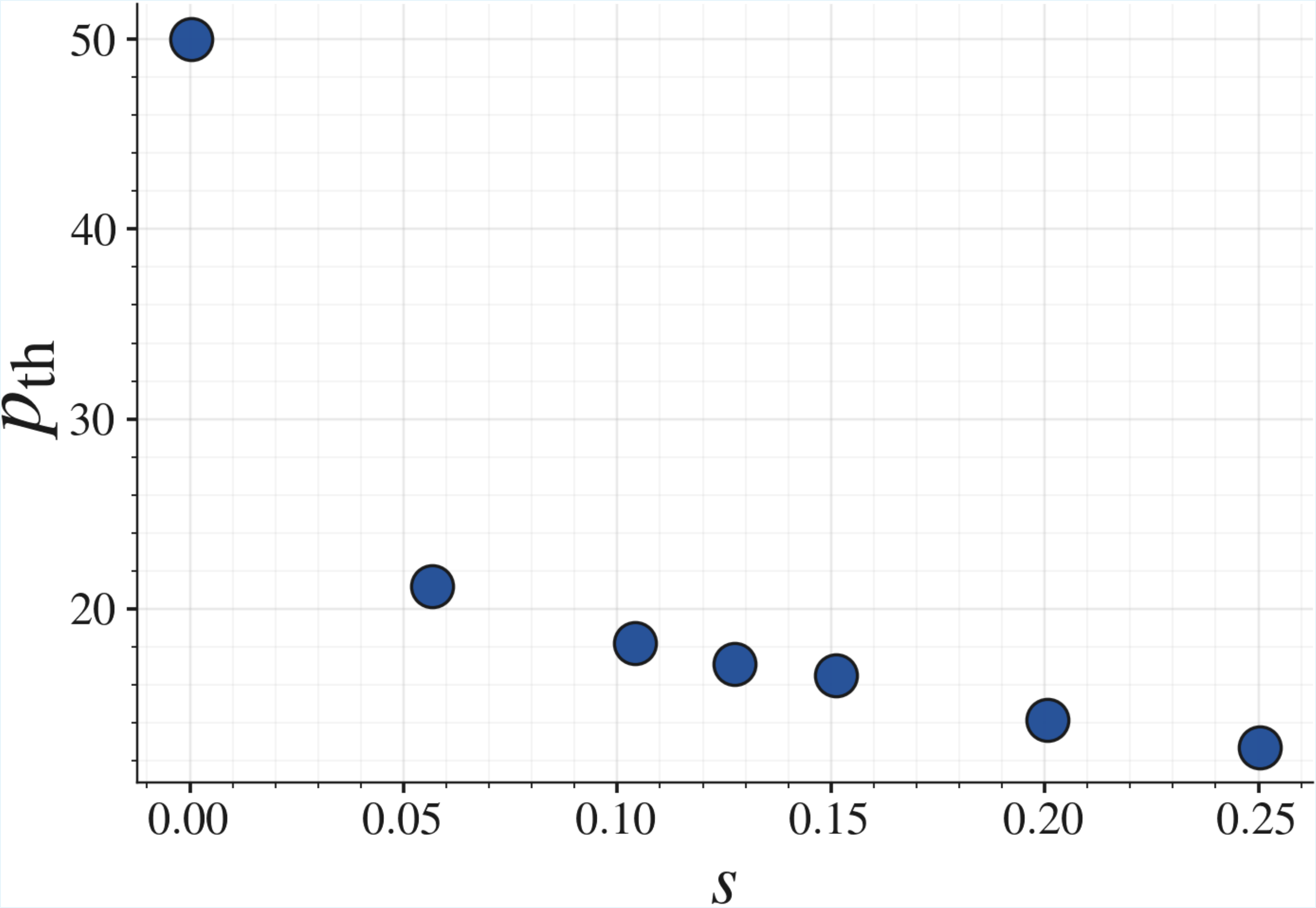}
\captionsetup{justification=justified,singlelinecheck=false}
\caption{Threshold ($p_{\text{th}}$) as a function of the slope $s$ in the linear scaling $|\mathcal{B}_Z(n)| = sn + c$ for the phase points shown in ~\cref{fig:support_for_threshold_decrease_final} in Appendix. The threshold decreases with increasing $s$.\label{fig:slope_vs_threshold}}
\end{figure}

\section{Performance of CDTCs under finite-bias code-capacity model}

While the infinite-bias code-capacity model isolates the mechanism behind the $50\%$ threshold, it is an idealized limit in which only a single Pauli error channel is present. In practice, noise typically has a large but finite bias, so subdominant Pauli components can disrupt the logical-operator structure that underlies optimal infinite-bias performance. We therefore study the finite-bias code-capacity model to quantify how thresholds and subthreshold scaling evolve as the bias is reduced.

\subsection{Noise model and decoding setup}

We compare four open-boundary tile-code variants: the undeformed CSS tile code, the linear deformation, the $XY$ deformation, and the translation-invariant (TI) deformation associated with the phase point $(0.25,0.5)$. Each physical qubit is subjected to i.i.d.\ biased Pauli noise with total error rate
\begin{equation}
p = p_X+p_Y+p_Z,
\end{equation}
and bias parameter $\eta$ defined by the ratio
\begin{equation}
\frac{p_Z}{p_X}=\frac{p_Z}{p_Y}=\eta,
\quad\text{equivalently}\quad
p_X:p_Y:p_Z = 1:1:\eta.
\end{equation}
Thus,
\begin{equation}
p_X=p_Y=\frac{p}{2+\eta},
\quad
p_Z=\frac{\eta p}{2+\eta}.
\end{equation}
Because $X$ and $Y$ errors enter symmetrically in this model, it suffices to consider single-qubit Clifford deformations whose conjugation action on Paulis is either trivial, an $X\leftrightarrow Z$ permutation, or a $Y\leftrightarrow Z$ permutation.\footnote{An $X\leftrightarrow Z$ permutation is implemented by a Hadamard gate $H$. A $Y\leftrightarrow Z$ permutation can be implemented by $HSH$, where $S=\sqrt{Z}$.}

Once the chosen Clifford deformation is applied to the parity-check matrix, we sample full Pauli error patterns from the finite-bias channel and decode using BP-OSD~\cite{Poulin_BP,Roffe2023biastailoredquantum,Roffe_LDPC_Python_tools_2022}. All reported logical error rates are estimated from $100{,}000$ Monte Carlo trials, using the full deformed parity-check matrix and the full sampled Pauli error vector.

\subsection{Threshold and subthreshold behavior versus bias}

Across all deformations studied, the estimated threshold increases monotonically with the bias $\eta$. This is expected: increasing $\eta$ suppresses the $X/Y$ components and drives the channel closer to the infinite-bias regime for which the deformations are designed to be most effective. The corresponding thresholds, together with the hashing bound as functions of $\eta$, are presented in ~\cref{fig:Threshold_vs_bias_logical_finite_bias_final}a. At large bias, the thresholds of the Clifford-deformed codes exceed the corresponding hashing bound. At lower bias, the performance of BP+OSD is expected to degrade in the presence of $Y\leftrightarrow Z$ type deformations; nevertheless, by analogy with the linear ($X \leftrightarrow Z$ only) deformation, where optimal decoding achieves thresholds above the hashing bound, it is plausible that suitably optimized decoding could also surpass the hashing bound in this regime.

~\cref{fig:Threshold_vs_bias_logical_finite_bias_final}b shows the logical error rate as a function of $\eta$ at fixed physical error rate $p \approx 0.14$ for the open-boundary code with $n=200$ and distance $d=9$. Among the four variants, the TI $(0.25,0.5)$ deformation exhibits the strongest suppression of logical error as the bias increases.

To probe subthreshold scaling at fixed bias, ~\cref{fig:Threshold_vs_bias_logical_finite_bias_final}c shows the logical error rate versus distance at $\eta=100$ and the same physical error rate $p \approx 0.14$. The logical error rate decays fastest for TI $(0.25,0.5)$, indicating the most favorable finite-bias scaling among the variants considered. Overall, these results show that the Clifford-deformed variants retain a clear finite-bias performance hierarchy, with TI $(0.25,0.5)$ performing best in both the threshold and subthreshold regimes.

\begin{figure}[t!]
\centering
\includegraphics[width=1\columnwidth]{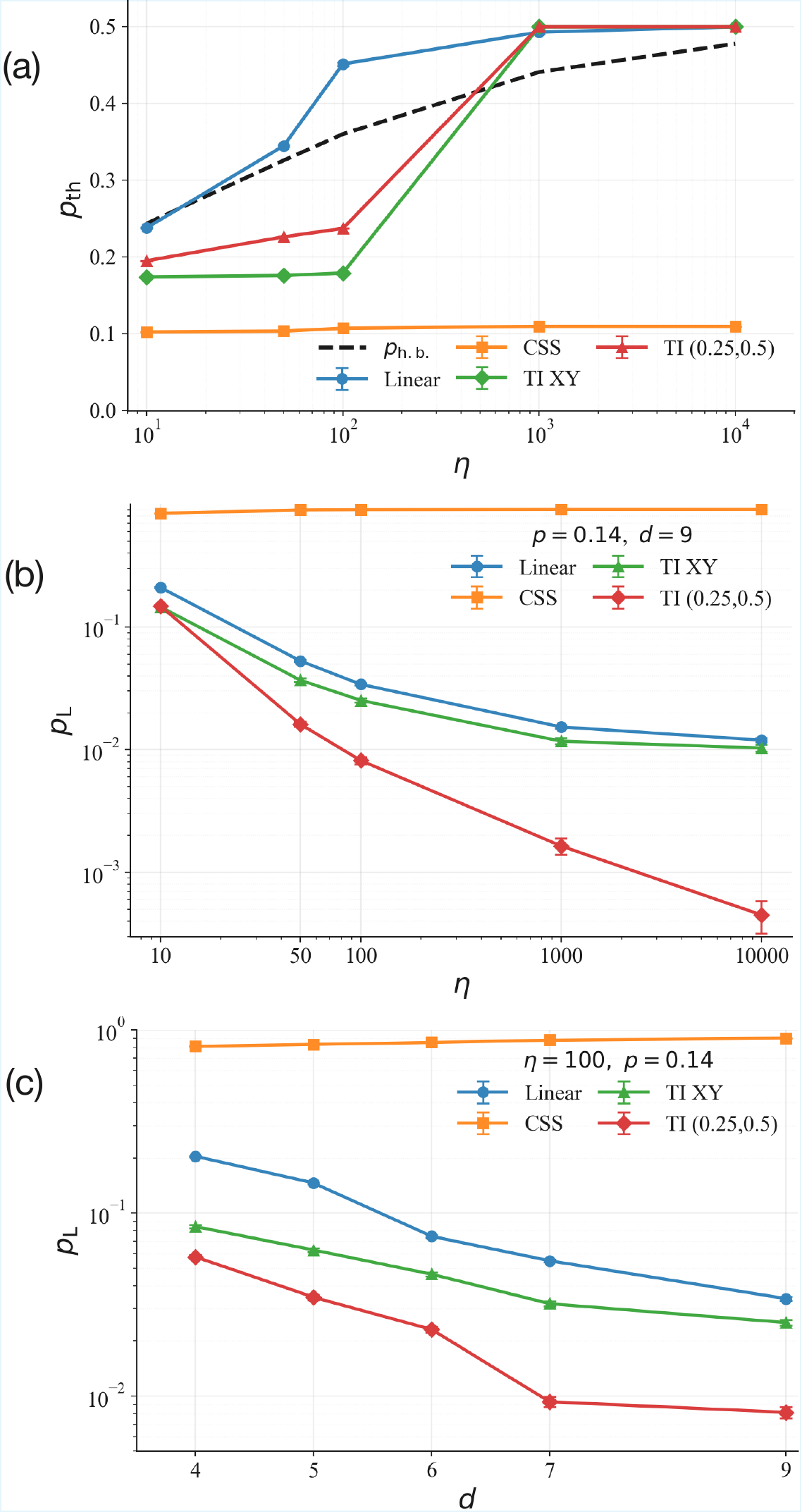}
\captionsetup{justification=justified,singlelinecheck=false}
\caption{\textbf{Performance under finite-bias code-capacity noise.} (a) The thresholds ($p_{\mathrm{th}}$) for four open-boundary variants (CSS, linear, $XY$, and TI $(0.25,0.5)$) and hashing bound ($p_{\mathrm{h.b.}}$) as functions of $\eta.$ (b) Logical error rate versus bias $\eta$ for the same variants at $p\approx0.14$ and $d=9$. (c) Subthreshold scaling: logical error rate versus distance at $\eta=100$ and $p\approx 0.14$. Each point is estimated using $10^5$ Monte Carlo trials with the BP-OSD decoder.}
\label{fig:Threshold_vs_bias_logical_finite_bias_final}
\end{figure}

\section{Performance of CDTCs under circuit-level noise model}

In a circuit-level noise model, faults arise from imperfect gates, measurements, resets, and idle operations. Even when the underlying physical noise is strongly biased, syndrome-extraction circuits generally mix Pauli components, producing correlated noise with reduced effective bias at the logical level. Circuit-level benchmarking is therefore essential for assessing practical fault-tolerant performance.
We study the same four open-boundary variants as above: the undeformed CSS tile code, the linear deformation, the $XY$ deformation, and the TI deformation associated with the phase point $(0.25,0.5)$. Here we employ a biased single-qubit Pauli noise model applied after Clifford operations (including both single- and two-qubit gates), before measurement, and after reset, with additional noise applied to all data qubits at the beginning of each round, while neglecting explicit idle errors.  We use Stim~\cite{gidney2021stim} to construct and sample full stabilizer-measurement circuits. Concretely, we adapt the surface-code-style syndrome-extraction templates provided in the StimBPOSD framework~\cite{higgott2024stimbposd} to the tile-code geometry and then apply the appropriate Clifford pattern to obtain the circuits for each deformed variant.

For decoding, we combine Stim’s circuit-level sampling with BP-OSD via the StimBPOSD interface~\cite{higgott2024stimbposd}. This integration enables efficient decoding of correlated detector-error patterns arising from circuit-level noise.

\subsection{Syndrome-extraction circuit structure}

~\cref{fig:circuit} shows an example syndrome-extraction circuit for the linear-deformed tile code (bulk checks shown; boundary checks are analogous but have lower weight). All data qubits are initialized in the $\ket{+}$ basis. Check qubits are initialized in the $\ket{+}$ basis at the start of each measurement round and are reset to $\ket{+}$ after measurement.

Let $\mathcal{V}$ denote the set of vertical data qubits and $\mathcal{H}$ the set of horizontal data qubits. In the linear deformation, the vertical data qubits in $\mathcal{V}$ are acted on by Hadamard gates $H$ before entangling with check qubits. Next, for each $X$-type check of the original CSS code, the corresponding check qubit is coupled to three horizontal data qubits via $CX$ gates and to three vertical data qubits via $CZ$ gates. Conversely, for each $Z$-type check, the check qubit is coupled to three horizontal data qubits via $CZ$ and to three vertical data qubits via $CX$. In all two-qubit gates, the check qubit is the control, and the data qubit is the target.

After each round, all check qubits are measured and then reset in the $\ket{+}$ basis. At the end of the circuit, all data qubits are measured in the $\ket{+}$ basis. Immediately before this final measurement, the vertical data qubits are acted on by $H$ so that the measured observables are deterministic in the chosen memory basis.

The circuits for the other deformed variants are constructed analogously by applying the corresponding Clifford pattern and updating the entangling gates accordingly. In the first measurement cycle, we place detectors only on checks in the chosen memory basis, since these checks are initialized in the $+1$ eigenspace and any deviation can be identified directly without comparing to earlier measurements. In subsequent cycles, detectors are placed on all checks to compare outcomes between consecutive rounds.

\begin{figure}[t]
\centering
\includegraphics[width=1\columnwidth]{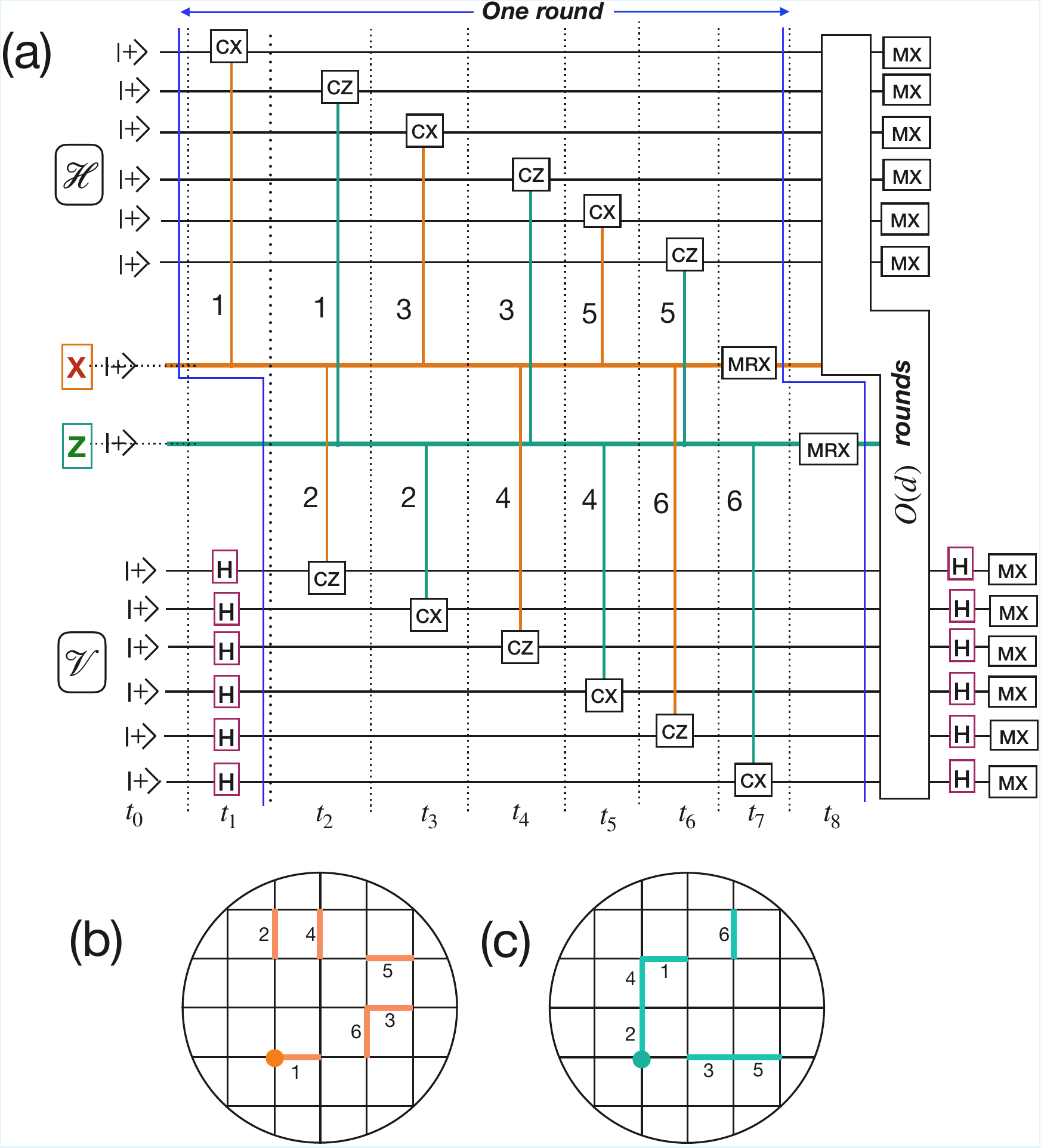}
\captionsetup{justification=justified,singlelinecheck=false}
\caption{\textbf{Circuit diagram for the tile code under linear deformation.} Example syndrome-extraction circuit showing $CX$, $CZ$, and Hadamard ($H$) gates. The circuit in (a) illustrates how a single bulk $X$ check and a single bulk $Z$ check (defined in the CSS code) couple to three vertical and three horizontal data qubits after deformation. Here $\mathcal{H}$ and $\mathcal{V}$ represent the horizontal and vertical qubits, respectively, in the tile code. All data qubits and checks are initialized in the $\vert+\rangle$ state. MRX denotes measurement and reset in the $\vert +\rangle$ basis, whereas MX denotes only measurement in the same basis. The qubit ordering for $X$ and $Z$ type stabilisers is shown in (b) and (c), respectively.}
\label{fig:circuit}
\end{figure}

% .
As explained here, when implementing each stabilizer-measurement circuit, it is important to avoid spacetime error mechanisms that can reduce the effective circuit distance. Such mechanisms can arise from error propagation involving both data qubits and ancilla qubits during stabilizer extraction, and they depend on the ordering of the two-qubit gates. With an unfavorable ordering, a single fault may propagate to produce a data-qubit error whose support is contained within a logical operator, thereby reducing the circuit distance below the code distance.
We provide the details of the circuit ordering optimization procedure in the  
\cref{sec:appendix_optimization}

To check that the chosen gate ordering preserves circuit distance, we analyze the circuit’s detector error model (DEM) using the codeDistance Python package \cite{webster2026}. As suggested in this work, we use both the heuristic and the exact algorithms mentioned to check the circuit distance for the CSS and one corresponding Clifford-deformed (linear deformation) circuit. Using the QDistEvol algorithm \cite{10874169,webster2026}, which has a significantly higher accuracy for the quantum LDPC code dataset than other heuristic algorithms, we found the circuit distance to be the same as the code distance for circuits up to d=9.  
Next, we perform the exact calculation using a Mixed Integer Linear Programming (MILP) based method with the help of the Gurobi solver implemented through the codeDistance package\cite{webster2026}. All the CSS and the linear-deformed open tile code circuits up to d=9 mentioned before give the same distance as the heuristic method before. Therefore, we can say that we have the optimal ordering for these open tile code circuit implementations, and they don't propagate hook errors, which reduces circuit distance.
% Stim’s routine \texttt{search\_for\_undetectable\_logical\_errors}~\cite{gidney2021stim}. This routine performs a heuristic search for collections of DEM error mechanisms whose combined action flips a logical observable while producing no detection events. The DEM includes both graph-like (degree-2) mechanisms and higher-degree hyperedges, unless truncated by user-defined limits. In our implementation, we allow error mechanisms to generate up to six detection events. The routine returns candidate sets of mechanisms that implement an undetectable logical fault; the smallest number of mechanisms found provides an upper bound on the circuit distance, defined as the minimum number of circuit-level fault mechanisms capable of inducing an undetected logical failure~\cite{kang2025quits}.

% For the gate ordering used in our simulations, the estimated circuit distance is never smaller than the corresponding code distance. Within the explored truncation limits of the heuristic search, this indicates that hook-error-induced distance reduction is not observed. While we do not claim this ordering is unique or optimal, it satisfies the necessary distance-preservation check used throughout this work.

\subsection{Circuit-level performance under biased noise}

Our results show that Clifford-deformed tile codes retain substantial performance advantages under biased noise even in the full circuit-level setting.

~\cref{fig:Threshold_vs_bias_logical_bias_final}a shows circuit-level thresholds as a function of bias $\eta$ for the four variants. At $\eta=10{,}000$, the estimated thresholds are approximately $0.634\%$ (CSS), $1.5\%$ (linear), $0.814\%$ ($XY$), and $1.04\%$ (TI $(0.25,0.5)$). Additional details are provided in ~\cref{app:stim}. All circuit-level simulations use StimBPOSD with $100{,}000$ Monte Carlo trials and $8$ measurement rounds.

Compared to code-capacity thresholds, circuit-level thresholds are substantially smaller than expected. This can be understood using renormalization of the physical error rate and the bias: even if each elementary operation is subjected to strongly biased noise, gate evolution, and measurements produce a phenomenological (including errors on the ancillas) effective noise channel at the end of the full measurement cycle with significantly reduced bias and enhanced effective physical error rates at the end of the circuit. This leads to a reduction in threshold. We explain this in more detail in Sec .~\cref {sec:qubit-platforms} below for circuit implementation in realistic hardware platforms. 

\cref{fig:Threshold_vs_bias_logical_bias_final}b shows the logical error rate versus bias at $p=0.005$ for distance $d=9$. ~\cref{fig:Threshold_vs_bias_logical_bias_final}c shows subthreshold scaling versus distance at $p=0.005$ and $\eta=100$. Across these diagnostics, TI $(0.25,0.5)$ is the best-performing variant.

\begin{figure}[t!]
\centering
\includegraphics[width=1\columnwidth]{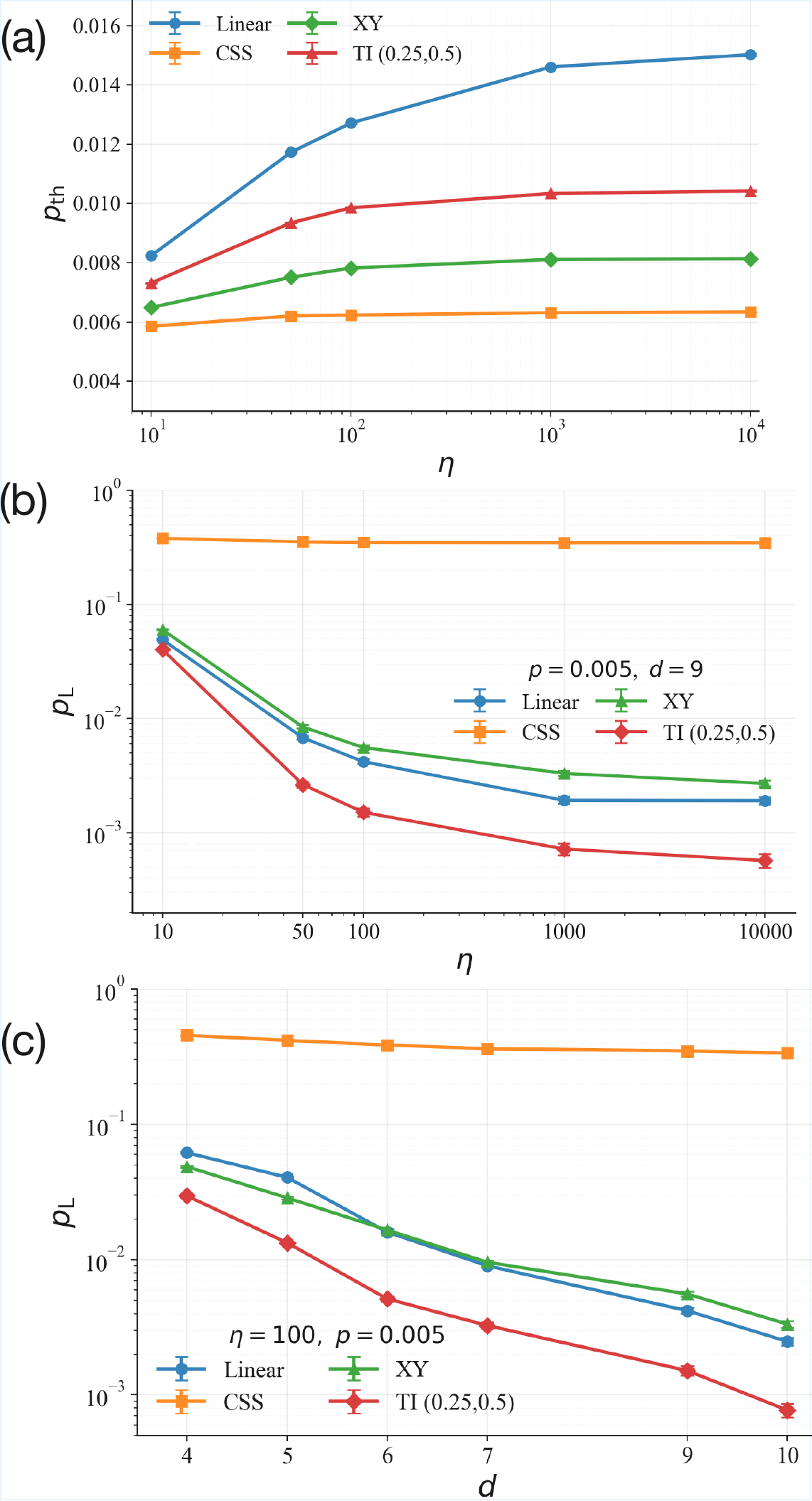}
\captionsetup{justification=justified,singlelinecheck=false}
\caption{\textbf{Performance under circuit-level noise.} (a) Circuit-level threshold versus bias $\eta$ for open tile codes under four variants: CSS, linear, $XY$, and TI $(0.25,0.5)$. (b) Logical error rate versus bias at $p=0.005$ and $d=9$ (code with $n=200$). (c) Subthreshold scaling: logical error rate versus distance at $\eta=100$ and $p=0.005$. Each data point is obtained from circuit-level simulations using the StimBPOSD decoder with $10^5$ Monte Carlo trials and $8$ rounds of syndrome extraction.}
\label{fig:Threshold_vs_bias_logical_bias_final}
\end{figure}

\section{Residual bias and effective phenomenological mapping for tile codes on different hardware platforms}
\label{sec:qubit-platforms}

In the previous section, we analyzed performance using standard circuit-level noise models. In such models, Pauli channels are assigned to operations such as gates, idle periods, initialization, and measurement, but are not tied to specific hardware implementations. Here, we instead consider hardware-informed circuit-level noise models, where Pauli channels are derived from microscopic descriptions of native gate dynamics. This additional structure can modify the noise bias during circuit execution, since realistic gate implementations need not preserve dephasing bias.

To relate these microscopic circuit-level models to phenomenological threshold results, we construct a mapping from one round of the time-translation-invariant (TI) syndrome-extraction (SE) circuit to an effective repeated-round Pauli noise model, i.e., a stationary Pauli noise model applied identically at each round. This mapping provides a simplified description for circuit-level performance, since it replaces a full threshold analysis for each microscopic variation of the noise model with the extraction of effective parameters that can be directly compared against known phenomenological threshold curves.

Concretely, for a TI-SE circuit, we sample Pauli error configurations from the location-dependent circuit-level noise model and propagate them through the Clifford circuit to the end of the syndrome-extraction round. This produces a distribution over end-of-round Pauli strings acting on both data and ancilla qubits. To enable comparison with phenomenological models, we approximate this distribution by its single-qubit marginal statistics, defining effective Pauli error probabilities \((\bar{p}_x,\bar{p}_y,\bar{p}_z)\). From these we obtain the total effective error rate \(p_{\mathrm{eff}}\) and effective bias \(\eta_{\mathrm{eff}}\). More details of this procedure can be found in~\cref{app:effective_bias}.

This approximation replaces the full correlated distribution of propagated Pauli strings with the corresponding single-qubit marginals, so that each multi-qubit Pauli string contributes to the effective probabilities according to the single-qubit Pauli operators appearing in its support. Its validity depends on the structure of the underlying noise. When errors are predominantly local and independent at injection, multi-qubit Pauli strings arise mainly from propagation through the Clifford circuit, which redistributes error support without changing the probability scaling of error events. In this regime, correlated error probabilities scale comparably to products of single-qubit error rates, and the marginal approximation is well justified. 

In contrast, hardware-informed noise models include two-qubit Pauli channels that introduce correlated errors directly at the level of the noise model. These correlations need not obey the same scaling and therefore are not fully captured by single-qubit marginals. In this work, we nevertheless employ the marginal approximation as a reduced description, using it to define an effective phenomenological model. This reduced description captures only single-qubit statistics and does not fully represent correlated error structure. When needed, one can instead retain and sample the full distribution of propagated Pauli strings to construct a correlated effective noise model.

Given \((p_{\mathrm{eff}},\eta_{\mathrm{eff}})\), we use the phenomenological threshold curve as a reference to assess threshold proximity for the effective model. Being below threshold indicates that logical error rates can be suppressed by increasing code distance, but does not guarantee low logical error rates at a fixed distance. Cases that lie well below threshold are expected to be more robust to the neglected correlated structure, whereas cases closer to threshold are inherently more sensitive to modeling details and should be interpreted with greater caution.

We focus on finite-dimensional qubit platforms. In such systems, CZ gates tend to preserve dephasing bias because they do not propagate $Z$ errors into $X/Y$ errors, whereas a perfectly bias-preserving CX gate cannot be implemented~\cite{no-go-theorem-2019}. As a result, platforms whose native entangling gate is CX generally reduce bias during circuit execution. In both native implementations, the complementary entangling gate is implemented via Hadamard conjugation. Since syndrome-extraction circuits generally require both types of entangling operations, these Hadamard operations introduce additional mechanisms that can reduce the effective bias regardless of the native implementation.

To incorporate hardware-specific effects, we assign Pauli channels to the native gates of each platform, obtained from microscopic simulations of the gate dynamics (see ~\cref{app:native_gates,app:noise_channels}). These channels are constructed to match a target physical error rate \(p\), corresponding to a specified gate fidelity. Throughout this section, we consider a system bias \(\eta_{\mathrm{sys}}=100\) and a physical error rate \(p=3\times 10^{-4}\), corresponding to a gate fidelity of approximately \(99.97\%\). While this value is somewhat beyond current experimental capabilities for some platforms, it represents a target regime where fault-tolerant operation is expected to become feasible.

We consider three compilation strategies: (i) a no-native compilation, where both CX and CZ gates are available directly and no hardware-specific gate noise is applied, (ii) native compilations without physical gate noise, where the circuit structure reflects CX- or CZ-native constraints but uses only single-qubit Pauli channels, and (iii) hardware-informed native compilations, where platform-specific two-qubit Pauli channels are included. This hierarchy allows us to distinguish the effects of circuit structure from those of hardware-induced correlations.

~\cref{tab:effective-bias} reports the total effective error rate \(p_{\mathrm{eff}}\), the resulting effective bias \(\eta_{\mathrm{eff}}\), and logical error rates obtained from standard phenomenological simulations using the corresponding effective Pauli noise model. All extracted effective models lie below the corresponding phenomenological thresholds, as determined from ~\cref{fig:Threshold_vs_bias_pheno}. The observed logical error rates generally track the distance of each effective model \((p_{\mathrm{eff}},\eta_{\mathrm{eff}})\) from the threshold curve, with configurations exhibiting larger \(p_{\mathrm{eff}}\) and smaller \(\eta_{\mathrm{eff}}\) lying closer to threshold and hence showing higher logical error rates. Deviations from this trend can occur when differences in the underlying Pauli error composition or neglected correlated structure become significant, as observed for the XY deformation.

We first compare the no-native circuits with the native compilations in the absence of hardware-specific gate noise. For CSS, Linear, and TI-middle codes, introducing native compilation (CX or CZ) increases the total effective error rate \(p_{\mathrm{eff}}\) and reduces the effective bias \(\eta_{\mathrm{eff}}\), moving the effective noise closer to the threshold curve and resulting in higher logical error rates. The difference between CX- and CZ-native implementations is relatively small for CSS and Linear codes, and modest for the TI-middle code, with the CX-native implementation yielding a slightly more favorable \((p_{\mathrm{eff}},\eta_{\mathrm{eff}})\). In contrast, the XY deformation shows a pronounced difference between native compilation strategies. In the no-native setting, the XY circuit is implemented using CX and CY gates without requiring CZ gates. A CZ-native compilation must therefore construct all entangling operations via \(H\)-conjugation, introducing additional Hadamard layers. These layers redistribute errors between Pauli components, reducing the effective bias \(\eta_{\mathrm{eff}}\) and increasing the total error rate \(p_{\mathrm{eff}}\), pushing the effective model significantly closer to the threshold curve and leading to higher logical error rates. More generally, cases with similar \(\eta_{\mathrm{eff}}\) but different \(p_{\mathrm{eff}}\) exhibit different logical error rates, indicating that performance depends on both parameters rather than on bias alone.

Finally, we compare the no-physical-gate and hardware-informed implementations. The inclusion of platform-specific two-qubit Pauli channels modifies the effective error distribution, typically increasing \(p_{\mathrm{eff}}\) and reducing \(\eta_{\mathrm{eff}}\). However, for CSS, Linear, and TI-derived codes, the relative performance between CX- and CZ-native compilations remains largely unchanged within each code. This indicates that most of the bias is lost due to restricting the code compilation to one entangling gate rather than from the additional noise introduced by the hardware. The XY deformation remains the main exception, where compilation overhead and the resulting change in error composition lead to significantly different performance across implementations.

\begin{figure}[t!]
\centering
\includegraphics[width=1\columnwidth]{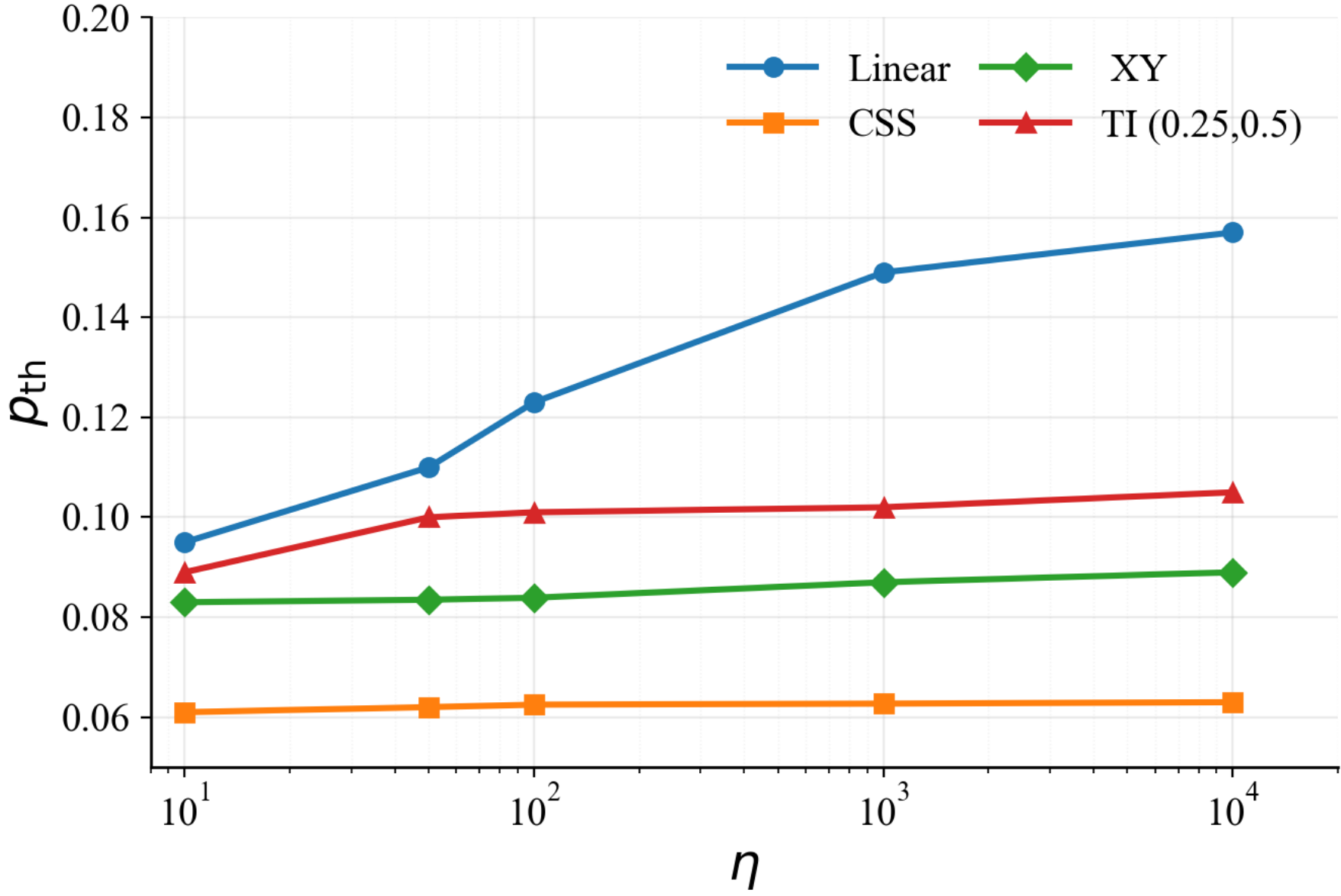}
\captionsetup{justification=justified,singlelinecheck=false}
\caption{\textbf{Phenomenological threshold as a function of bias.}
Threshold error rate as a function of noise bias \(\eta\) for the four code variants under a fault-tolerant phenomenological Pauli noise model. Each point is obtained from Monte Carlo simulations using the StimBPOSD decoder with \(10^5\) samples and \(8\) rounds of syndrome extraction. This curve defines the threshold boundary in \((p,\eta)\) space used to assess the effective parameters \((p_{\mathrm{eff}},\eta_{\mathrm{eff}})\) extracted from circuit-level noise models.}
\label{fig:Threshold_vs_bias_pheno}
\end{figure}

\begin{table}[t]
\centering
\scriptsize
\setlength{\tabcolsep}{4pt}
\renewcommand{\arraystretch}{1.05}

\begin{tabular}{l l l r r r}
\toprule
Code & Comp. & Platform &
\multicolumn{1}{c}{$p_{\mathrm{eff}}$} &
\multicolumn{1}{c}{$\eta_{\mathrm{eff}}$} &
\multicolumn{1}{c}{$p_l$} \\
 & strat. & &
\multicolumn{1}{c}{$\times 10^{-2}$} & & \multicolumn{1}{c}{$\times 10^{-3}$} \\
\midrule

\multirow{7}{*}{CSS}
& No & \multicolumn{1}{c}{--} & 0.296 & 73.9 & 1.59 \\
\cmidrule(lr){2-6}
& \multirow{3}{*}{CX}
  & No phys.-gate & 0.708 & 10.2 & 12.2 \\
& & Trapped ions & 0.726 & 5.91 & 15.2 \\
& & Superconduct. & 0.713 & 7.15 & 14.9 \\
\cmidrule(lr){2-6}
& \multirow{3}{*}{CZ}
  & No phys.-gate & 0.741 & 10.4 & 13.8 \\
& & Trapped ions & 0.764 & 10.1 & 16.8 \\
& & Neutral atoms & 0.746 & 10.5 & 15.7 \\
\midrule

\multirow{7}{*}{Linear}
& No & \multicolumn{1}{c}{--} & 0.303 & 24.7 & 0.24 \\
\cmidrule(lr){2-6}
& \multirow{3}{*}{CX}
  & No phys.-gate & 0.765 & 9.06 & 2.36 \\
& & Trapped ions & 0.816 & 6.13 & 6.79 \\
& & Superconduct. & 0.804 & 7.29 & 5.57 \\
\cmidrule(lr){2-6}
& \multirow{3}{*}{CZ}
  & No phys.-gate & 0.734 & 9.91 & 2.51 \\
& & Trapped ions & 0.744 & 8.40 & 4.03 \\
& & Neutral atoms & 0.725 & 8.72 & 3.45 \\
\midrule

\multirow{7}{*}{XY}
& No & \multicolumn{1}{c}{--} & 0.403 & 72.2 & 0.31 \\
\cmidrule(lr){2-6}
& \multirow{3}{*}{CX}
  & No phys.-gate & 1.018 & 56.5 & 2.55 \\
& & Trapped ions & 0.965 & 9.19 & 2.26 \\
& & Superconduct. & 0.992 & 17.6 & 2.43 \\
\cmidrule(lr){2-6}
& \multirow{3}{*}{CZ}
  & No phys.-gate & 2.239 & 13.2 & 22.6 \\
& & Trapped ions & 2.225 & 12.4 & 21.5 \\
& & Neutral atoms & 2.203 & 13.2 & 20.8 \\
\midrule

\multirow{7}{*}{TI Mid.}
& No & \multicolumn{1}{c}{--} & 0.443 & 41.2 & 0.09 \\
\cmidrule(lr){2-6}
& \multirow{3}{*}{CX}
  & No phys.-gate & 1.276 & 23.0 & 1.29 \\
& & Trapped ions & 1.220 & 9.42 & 2.62 \\
& & Superconduct. & 1.221 & 13.4 & 2.17 \\
\cmidrule(lr){2-6}
& \multirow{3}{*}{CZ}
  & No phys.-gate & 1.805 & 13.6 & 5.88 \\
& & Trapped ions & 1.775 & 12.8 & 5.61 \\
& & Neutral atoms & 1.749 & 13.5 & 5.01 \\
\bottomrule
\end{tabular}

\captionsetup{justification=justified}
\caption{\textbf{Effective error rate and bias across compilation strategies and platforms.}
For each configuration we report the total effective error rate \(p_{\mathrm{eff}} = \bar{p}_x + \bar{p}_y + \bar{p}_z\), the effective bias \(\eta_{\mathrm{eff}} = \bar{p}_z/(\bar{p}_x+\bar{p}_y)\), and the logical error rate \(p_l\) obtained from phenomenological simulations using the effective Pauli noise model. Compilation strategies and platforms progress from no-native and no hardware two-qubit noise (\textbf{No - }), to native structure with only single-qubit noise (\textbf{CX/CZ - No phys.-gate}), to native structure with platform-specific two-qubit Pauli channels (\textbf{CX/CZ - Trapped ions, Superconduct., Neutral atoms}).}
\label{tab:effective-bias}
\end{table}

%\begin{figure}[t]
%  \centering
%%  \begin{subfigure}[b]{0.9\linewidth}
%   \centering
%    \includegraphics[width=\linewidth]{cnot_ptm.pdf}
%    \caption{Noiseless CX PTM.}
%    \label{fig:ptm-cx-a}
%%  \end{subfigure}\hfill
%  \begin{subfigure}[b]{0.9\linewidth}
%    \centering
%    \includegraphics[width=\linewidth]{cz_ptm.pdf}
%    \caption{Noiseless CZ PTM.}
%    \label{fig:ptm-cz-b}
%  \end{subfigure}
%
%  \caption{Pauli transfer matrices (PTMs) for noiseless native gates. (a) CX gate obtained by simulating Eq.~\eqref{eq:H_MS} (plus single-qubit rotations) with $\Omega = 100~\mathrm{kHz}$, $\theta=\pi/2$, and $t=\theta/\Omega$. The PTM for the superconducting-qubit platform is identical in the noiseless case. (b) CZ gate obtained by simulating Eq.~\eqref{eq:H_ZZ_ions} (plus single-qubit rotations) with $\Omega = 100~\mathrm{kHz}$, $\theta=\pi$, and $t=\theta/\Omega$. The PTM for the neutral-atom platform is identical in the noiseless case.}
%  \label{fig:ptm-noiseless}
%\end{figure}

\section{Conclusion}

We show that Clifford deformations can drive zero-rate quantum LDPC codes toward the optimal $50\%$ threshold in the infinite-bias code-capacity setting, with the sufficient condition governed by the structure of pure-$Z$ logical operators. For tile codes, we identify both random and translation-invariant deformations consistent with this behavior and map out a phase diagram exhibiting an extended $50\%$ region.

At finite bias, the performance hierarchy among tile-code variants persists in both code-capacity and circuit-level settings. Translation-invariant deformations, and in particular TI $(0.25,0.5)$, yield the strongest thresholds and the most favorable subthreshold scaling among the variants studied, while the undeformed CSS code performs the worst. Circuit-level thresholds are substantially smaller than code-capacity thresholds, consistent with bias renormalization induced by syndrome-extraction circuits, yet the relative advantage of Clifford-deformed variants remains robust.

%Finally, we quantify the extent to which dephasing bias is preserved in realistic syndrome-extraction circuits for the tile code and its Clifford deformations across different qubit platforms. By deriving gate noise channels from Lindblad simulations and propagating Pauli errors through the resulting circuits, we compute the effective bias $\eta_{\mathrm{eff}}$ for both idealized and platform-specific gate sets. This provides a concrete estimate of how native gates and compilation strategies impact bias preservation and, consequently, the practical performance of biased-noise error-correcting codes.

% \PM{New conclusion about my section}
Finally, we relate hardware-informed circuit-level noise to an effective phenomenological description for tile codes and their Clifford deformations across different qubit platforms. By propagating Pauli errors through syndrome-extraction circuits, we extract the effective parameters \(p_{\mathrm{eff}}\) and \(\eta_{\mathrm{eff}}\) for each implementation. This provides a practical way to compare compilation strategies and hardware platforms through their position relative to the phenomenological threshold curve, highlighting how both circuit structure and noise properties influence performance.

Overall, our results show that Clifford deformations provide a systematic route to biased-noise optimization for LDPC codes, linking high thresholds and strong subthreshold behavior to concrete, checkable properties of the pure-$Z$ logical-operator structure.\\

\noindent\textit{Note:} The code used for the simulations in this work is publicly available at 
\url{https://github.com/JD0111phys/biased_noise_tile_code}.

\section{Acknowledgement}
The authors acknowledge Advanced Research Computing at Virginia Tech for providing computational resources and technical support that have contributed to the results reported within this paper (\url{https://arc.vt.edu}). AD was supported through NSF award 2515064 and through the start-up fund at VT. AP was supported by the Engineering and Physical Sciences Research Council (EP/S021582/1) and PASQUOPS (BMFTR).

\bibliographystyle{unsrtnat}
\bibliography{bib}

\clearpage
\onecolumngrid

%\appendix

\section{Appendix}
\subsection{Threshold bound for the case of overlapping BLO via assigned overlap loads}
\label{app:overlap_covering}

In ~\cref{sec: Clifford-deformed zero-rate LDPC codes}, we proved two sufficient conditions for a $50\%$ infinite-bias threshold: a union bound over all pure-$Z$ logical operators, and a sharper bound when there exists a non-overlapping BLO. The purpose of this appendix is to extend the latter argument to the general case in which BLO elements may overlap. The main idea is to assign overlap qubits to basis elements, define corresponding shifted bad events, and then show that every logical failure is covered by one of these bad events.

We keep the notation of the main text. In particular, $\mathcal{L}_Z$ denotes the set of non-trivial pure-$Z$ logical operators, $\mathcal{B}_Z=\{L_1,\ldots,L_{n_B}\}$ denotes a BLO, with $n_B=k+n_z$. For any $L\in\mathcal{L}_Z$, the notation $\decomp(L)\subseteq [n_B]$ denotes the set of BLO indices appearing in the basis decomposition of $L$.

Throughout this appendix, we work in the infinite-bias code-capacity setting, namely i.i.d.\ $Z$-only noise with per-qubit error rate $p_Z\in(0,1/2)$. For a sampled Pauli $Z$-error $E\in \mathsf{P}_n^Z$, we write
\begin{equation}
p(E)=p_Z^{|E|}(1-p_Z)^{n-|E|}.
\end{equation}
As discussed in the main text, for a non-trivial pure-$Z$ logical operator $L\in\mathcal{L}_Z$, define
\begin{equation}
U(L)
:=
\{\,E\in\mathsf{P}_n^Z:\ |LE|\le |E|\,\}.
\end{equation}
Under i.i.d.\ $Z$ noise with $p_Z<1/2$, this is exactly the condition that $LE$ is at least as likely as $E$. By the half-weight criterion proved in the main text,
\begin{equation}
E\in U(L)
\quad\Longleftrightarrow\quad
|E \cap L| \ge |L|/2.
\label{eq:app_half_weight_recall}
\end{equation}

\subsubsection{Assigned overlap loads and basis bad events}

When BLO elements overlap, products of basis logicals can cancel on overlap qubits and produce logical operators supported on symmetric differences. Consequently, an error pattern can be uncorrectable for a product logical even when no single basis element satisfies the naive half-weight condition on its full support. To control this while summing only over the BLO, we assign overlap qubits to basis elements\footnote{The role of overlap assignment is already visible in the simplest case of two overlapping BLO elements. Suppose $L_3=L_1L_2$ and let $c:=|L_1\cap L_2|$. Since overlap qubits cancel in the product, one has $|L_3|=|L_1|+|L_2|-2c$. Thus, any uncorrectable pattern for $L_3$, namely one with at least $|L_3|/2$ errors on $\supp(L_3)$, must be covered by shifted half-weight events on $L_1$ and $L_2$ as long as the total shift assigned to the two generators is at least the overlap $c$. For instance, if $\delta_1,\delta_2\ge 0$ satisfy $\delta_1+\delta_2\ge c$, then
\[
\{x_3\ge |L_3|/2\}
\subseteq
\{x_1\ge |L_1|/2-\delta_1\}
\cup
\{x_2\ge |L_2|/2-\delta_2\}.
\]
Likewise, one may assign the entire overlap asymmetrically to one generator: if $\delta\ge c$, then
\[
\{x_3\ge |L_3|/2\}
\subseteq
\{x_1\ge |L_1|/2\}
\cup
\{x_2\ge |L_2|/2-\delta\}.
\]
These are immediate special cases of the general covering lemma proved below.}.

Let
\[
\Omega:=\{\,q\in[n]: \big|\{i\in[n_B]:\ q\in \supp(L_i)\}\big|\ge 2\,\}
\]
be the set of overlap qubits, i.e. the qubits that belong to more than one basis logical. For each $q\in\Omega$, define
\[
\mathcal{B}(q):=\{i\in[n_B]:\ q\in \supp(L_i)\}.
\]
Choose an assignment map
\[
a:\Omega\to \mathcal{B}_Z
\qquad\text{such that}\qquad
a(q)\in \mathcal{B}(q)\quad\text{for every }q\in\Omega.
\]
Thus, each overlap qubit is charged to exactly one BLO element that contains it.

\begin{defn}[Assigned overlap load]
For $i\in[n_B]$, define its assigned overlap load by
\begin{equation}
o_i:=\bigl|\{q\in \Omega\cap \supp(L_i): a(q)=L_i\}\bigr|.
\label{eq:oL_def_app}
\end{equation}
\end{defn}

The role of \(o_i\) is to relax the half-weight threshold on \(L_i\) by the amount of overlap that has been assigned to it.

\begin{defn}[Shifted thresholds and bad events]
For $i\in[n_B]$, let $X_i\sim \Bin(|L_i|,p_Z)$ be the number of $Z$ errors on $\supp(L_i)$. Define the shifted thresholds
\begin{equation}
t_i:=\frac{|L_i|}{2}-o_i,
\label{eq:tL_def_app}
\end{equation}
and define the associated bad event
\begin{equation}
\mathrm{BAD}(i)
:=
\bigl\{ E\in\mathsf{P}_n^Z:\ | E \cap L_i | \ge t_i\bigr\}.
\label{eq:BAD_def_app}
\end{equation}
Equivalently,
\[
\Pr(\mathrm{BAD}(i))=\Pr(X_i\ge t_i).
\]
\end{defn}

\subsubsection{A covering lemma for products of overlapping basis logicals}

We now state the key combinatorial lemma. It shows that if a product logical has at least half its support in error, and if the shifted thresholds on its BLO factors sum to at most half the support of the product, then at least one basis bad event must occur.

\begin{lemma}[General covering lemma]
\label{lem:general_union_covering}
Let $L \in \mathcal{L}_Z$, $S=\decomp(L)$, and $E \in \mathsf{P}_n^Z$.
Suppose some real numbers $\{a_i\}_{i\in S}$ satisfy
\begin{equation}
\sum_{i\in S} a_i \le \frac{|L|}{2}.
\label{eq:general_sum_condition_app}
\end{equation}
Then
\begin{equation}
|E\cap L|\ge \frac{|L|}{2} \Longrightarrow
\exists i \in S,\, |E\cap L_i| \ge a_i.
\label{eq:general_union_cover_app}
\end{equation}
\end{lemma}

\begin{proof}
Assume $|E\cap L|\ge |L|/2$, but that for every $i\in S$, $|E\cap L_i|<a_i$. Then,
\[
\sum_{i\in S} |E\cap L_i| <\sum_{i\in S} a_i\le \frac{|L|}{2}.
\]
On the other hand, every error counted in $|E\cap L|$ lies on the support of at least one $L_i$, so
\[
|E\cap L| \le \sum_{i\in S} |E\cap L_i| <\frac{|L|}{2},
\]
contradicting $|E\cap L| \ge |L|/2$. This proves the claim.
\end{proof}

This motivates the following condition on the assigned loads.

\begin{defn}[Covering condition]
\label{def:covering_condition_app}
The assigned loads $\{o_i\}_{i\in[n_B]}$ satisfy the covering condition if for every non-trivial $L\in\mathcal{L}_Z$, writing $S=\decomp(L)$,
one has
\begin{equation}
\sum_{i\in S}\left(\frac{|L_i|}{2}-o_i\right)\le \frac{|L|}{2}.
\label{eq:covering_condition_app}
\end{equation}
\end{defn}

Under this condition, every logical failure is covered by a basis bad event.

\textit{Remark:} We should note that we don't know of any examples satisfying Definition 3 (covering condition) which requires that for every nontrivial logical, the shifted thresholds sum to at most $|L|/2$. Also, for highly overlapping BLOs, it may be impossible to find loads $o_i$ satisfying the covering condition.

\begin{proposition}[Reduction to basis bad events]
\label{prop:reduction_to_BAD_app}
Assume that the assigned loads satisfy Definition~\cref{def:covering_condition_app}. Then
\begin{equation}
\{\text{logical failure}\}
\subseteq
\bigcup_{i \in [n_B]}\mathrm{BAD}(i),
\label{eq:failure_subset_BAD_app}
\end{equation}
and therefore
\begin{equation}
P_{\mathrm{fail}}
\le
\sum_{i=1}^{n_B}\Pr(\mathrm{BAD}(i)).
\label{eq:pfail_BAD_union_app}
\end{equation}
\end{proposition}

\begin{proof}
Let $E \in \mathsf{P}_n^Z$ be an error that leads to a logical failure after decoding. Then there exists $L\in\mathcal{L}_Z$ such that $E\in U(L)$. Let $S=\decomp(L)$. By Eq.~\eqref{eq:app_half_weight_recall},
\[
|E \cap L | \ge |L|/2.
\]
By Definition~\cref{def:covering_condition_app}, the thresholds
\[
t_i=\frac{|L_i|}{2}-o_i
\qquad (i\in S)
\]
satisfy the hypothesis of Lemma~\cref{lem:general_union_covering}. Therefore, there exists $i \in S$ such that $|E \cap L_i| \geq t_i$. This means that
\[
E\in \bigcup_{i\in S}\mathrm{BAD}(i)
\subseteq
\bigcup_{i\in[n_B]}\mathrm{BAD}(i).
\]
Taking probabilities and applying the union bound yields \eqref{eq:pfail_BAD_union_app}.
\end{proof}

To get a tighter upper bound on the failure probability, the assignment map $a$ should be chosen so that the shifted thresholds
\[
t_i=\frac{|L_i|}{2}-o_i
\]
remain as large as possible for all basis elements. Equivalently, one wants to choose the loads $\{o_i\}$ so that the reduction from the naive half-weight threshold $|L|/2$ is as small as possible. This optimization can be formulated as an integer linear program. Introduce binary variables $y_{q,L}\in\{0,1\}$ for each overlap qubit $q\in\Omega$ and each $L\in\mathcal{B}(q)$, with $y_{q,L}=1$ if $q$ is assigned to $L$. The constraints are
\[
\sum_{L\in\mathcal{B}(q)} y_{q,L}=1
\qquad\text{for all } q\in\Omega,
\]
together with
\[
o_i=\sum_{q\in\Omega\cap\supp(L)} y_{q,{L_i}}
\qquad\text{for all } i\in [n_B].
\]
We could then optimize, for example, $\max_i o_i$, or any other linear objective reflecting the desired notion of minimal threshold reduction.
In practice, we could either optimize the loads subject only to the assignment constraints, or include the covering inequalities as additional linear constraints so that the resulting assignment is directly compatible with Definition~\cref{def:covering_condition_app}.

\subsubsection{Threshold bound with assigned overlap loads}

We can now derive the asymptotic failure bound from the bad-event reduction.

\begin{theorem}[Threshold bound with assigned overlap loads]
\label{thm:overlap_threshold_assigned}
Consider a family of zero-rate LDPC codes together with a fixed Clifford deformation. Assume that for each blocklength $n$ there exists a BLO $\mathcal{B}_Z(n)=\{L_1,\ldots,L_{n_B^{(n)}} \}$ and an assignment map $a$ with loads $\{o_i\}_{i\in \left[n_B^{(n)}\right]}$ such that:
\begin{itemize}
    \item the loads satisfy the covering condition of Definition~\cref{def:covering_condition_app};
    \item for every $i \in \left[n_B^{(n)}\right]$,
    \[
    |L_i| \geq K_1 n^{\alpha}
    \]
    for some constants $K_1>0$ and $0<\alpha \le 1$;
    \item for every $i \in \left[n_B^{(n)}\right]$,
    \[
    o_i \leq K_2 n^{\beta}
    \]
    for some constants $K_2\ge 0$ and $0\le \beta<\alpha$
\end{itemize}
Then under i.i.d.\ infinite-bias $Z$ noise with $p_Z<\tfrac12$,
\[
P_{\mathrm{fail}}(n)\le \exp\!\bigl(-\Omega(n^{\alpha})\bigr),
\]

In particular, the family has a $50\%$ infinite-bias threshold.
\end{theorem}

\begin{proof}
By Proposition~\cref{prop:reduction_to_BAD_app},
\begin{equation}
P_{\mathrm{fail}}
\le
\sum_{i=1}^{n_B^{(n)}}\Pr(\mathrm{BAD}(i)).
\label{eq:proof_pf_BAD_app}
\end{equation}
Fix $i\in \left[n_B^{(n)}\right]$, and let
\[
X_i\sim \Bin(|L_i|,p_Z)
\]
be the number of $Z$ errors on $\supp(L_i)$. By definition,
\[
\Pr(\mathrm{BAD}(i))=\Pr(X_i\ge t_i),
\qquad
t_i=\frac{|L_i|}{2}-o_i.
\]
Applying the Chernoff--Hoeffding bound gives
\begin{equation}
\Pr(X_i\ge t_i)
\le
\exp\!\left(-2\frac{(t_i-p_Z|L_i|)^2}{|L_i|}\right).
\label{eq:chernoff_main_app}
\end{equation}
Now
\[
t_i-p_Z|L_i|
=
|L_i|\left(\frac12-p_Z-\frac{o_i}{|L_i|}\right).
\]
Since
\[
|L_i| \geq K_1 n^{\alpha},
\qquad
o_i \leq K_2 n^{\beta},
\qquad
\beta<\alpha,
\]
we have
\[
\frac{o_i}{|L_i|}
\leq
\frac{K_2}{K_1}n^{\beta-\alpha}
\xrightarrow[n\to\infty]{}0.
\]
Hence for fixed $p_Z<1/2$,
\[
\frac12-p_Z-\frac{o_i}{|L_i|}=\Omega(1),
\]
and therefore
\begin{equation}
\Pr(\mathrm{BAD}(i))
\le
\exp\!\bigl(-\Omega(n^{\alpha})\bigr).
\label{eq:BAD_exp_bound_app}
\end{equation}
Substituting \eqref{eq:BAD_exp_bound_app} into \eqref{eq:proof_pf_BAD_app} gives
\[
P_{\mathrm{fail}}
\leq
\sum_{i=1}^{n_B^{(n)}} \exp\!\bigl(-\Omega(n^{\alpha})\bigr)
\le
n_B^{(n)}\,\exp\!\bigl(-\Omega(n^{\alpha}\bigr).
\]
Using $n_B^{(n)}\le n$ as shown in Eq.~\eqref{eq:BLO} of the main text, we obtain
\[
P_{\mathrm{fail}}
\le
n\,\exp\!\bigl(-\Omega(n^{\alpha}\bigr)
=
\exp\!\bigl(-\Omega(n^{\alpha})\bigr).
\]
Thus $P_{\mathrm{fail}}\to 0$ for every fixed $p_Z<1/2$, proving the claimed threshold.
\end{proof}

\textit{Remark.}
For constant-rate families, one expects that any feasible assignment forces $o_i$ to scale comparably to $|L_i|$ for some subset of BLO elements. In that regime, the bound above suggests a reduced threshold determined by the asymptotic ratio $o_i/|L_i|$.

% \subsubsection{Two examples of assigned overlap loads}

\subsection{Proof for biased noise threshold of the linear-deformed tile code using weight reduction}
\label{app:weight-reduction}

Consider a periodic $7\ell \times 7\ell$ lattice, with $\ell$ odd. Qubits are placed on edges and stabilizers correspond to the linear-deformed tile code described in \cref{sec: Clifford Deformed Tile code} and shown in \cref{fig:linear-deformation}, where a Hadamard has been applied to all vertical qubits. We prove here that there exists a family of such codes that has a $50\%$ threshold at infinite bias:
\begin{theorem}
    Let $\mathcal{L}$ the set of odd primes $\ell$ such that:
    \begin{itemize}
        \item \(\ell \equiv 2 \pmod{3}\),
        \item \(2\) is a primitive root modulo \(\ell\), that is
            \(\operatorname{ord}_\ell(2) = \ell - 1\).
    \end{itemize}
    Then the family of linear-deformed codes on a periodic lattice of size $7\ell \times 7\ell$ with $\ell \in \mathcal{L}$ has a $50\%$ threshold at infinite bias.
\end{theorem}
\noindent
The reason for choosing such a family will become clear in the proof. We note that the set $\mathcal{L}$ can be proven to be infinite under Artin's conjecture on primitive roots, which we therefore assume to be true here. Numerically, one can show that this set is large enough for the theorem to be relevant for any practical purpose.

% \AD{Is it possible to remove these proof numbers? i.e. the 2 in front of \textbf{Proof} below}
\begin{proof}
At infinite $Z$ bias, only the $X$ part of the stabilizers becomes relevant. We notice that the $X$ parts of the Clifford-deformed $X$ and $Z$ stabilizers have the same shape up to rotations and reflections of the lattice, with the first being supported on horizontal qubits and the second on vertical qubits (see \cref{fig:linear-deformation}). 
Without loss of generality, we can therefore focus our attention to the $X$ part of the Clifford-deformed $X$ stabilizers, which behaves like a classical parity check acting against $Z$ errors. Considering a periodic lattice of size $L \times L$ such that every horizontal edge is given integer coordinates (as shown in \cref{fig:materialized-symmetry-1-a}), there is one such weight-3 check for every coordinates $(x,y)$, supported on edges $(x,y)$, $(x+2,y+1)$, $(x+2, y+2)$. We call such check $\tilde{S}_0^{(x,y)}$. More precisely, writing $X_{(x,y)}$ for a Pauli $X$ operator supported on the qubit located at coordinates $(x,y)$, we have by defn:
\begin{align}
    \tilde{S}_0^{(x,y)} = X_{(x,y)} X_{(x+2,y+1)} X_{(x+2, y+2)}
\end{align}

To prove that the code has a threshold of $50\%$ at infinite bias, we use the \textit{weight-reduction method}~\cite{Huang_2023,tiurev2024domain}. It consists of first finding linear materialized symmetries, that is, sets of $\Theta(L)$ checks along a one-dimensional manifold whose product is the identity. Any such symmetry leads to a 1D repetition code problem, whose decoding allows us to find the value of new checks of reduced weight. We can then repeat this process until all the qubits are decoded. Since such a decoding method reduces to decoding repetition codes, it has a threshold of $50\%$. In our case, we start by constructing new weight-4 checks by multiplying the checks $\tilde{S}_0^{(x,y)}$ appropriately. We then find materialized symmetries involving those weight-4 checks, and show that those materialized symmetries involve $\Theta(L)$ qubits. Finally, we determine the weight-reduced decoding problem and show that the full decoding problem has a threshold of $50\%$.

\medskip
\textbf{Recursive construction of weight-4 checks. }
On an infinite lattice (that is, before imposing the modulo $7\ell$ periodicity), we can obtain enlarged versions of the initial checks by multiplying them in an appropriate way, as shown in \cref{fig:materialized-symmetry-1-b} and \cref{fig:materialized-symmetry-1-c}. More precisely, we can obtain a check $\tilde{S}_1^{(x,y)}$ supported on the edges $(x,y)$, $(x+4,y+2)$ and $(x+4,y+4)$ by multiplying $\tilde{S}_0^{(x,y)}$, $\tilde{S}_0^{(x+2,y+1)}$ and $\tilde{S}_0^{(x+2,y+2)}$:
\begin{align}
    \begin{split}
        \tilde{S}_1^{(x,y)} 
        &= \tilde{S}_0^{(x,y)} \tilde{S}_0^{(x+2,y+1)} \tilde{S}_0^{(x+2,y+2)} \\
        &= \left(X_{(x,y)} X_{(x+2,y+1)} X_{(x+2,y+2)}\right) \left(X_{(x+2,y+1)} X_{(x+4,y+2)} X_{(x+4,y+3)}\right) \left(X_{(x+2,y+2)} X_{(x+4,y+3)} X_{(x+4,y+4)}\right) \\
        & = X_{(x,y)} X_{(x+4,y+2)} X_{(x+4,y+4)}.
    \end{split}
\end{align}
We can repeat this process recursively to obtain a check $\tilde{S}_k^{(x,y)}$ supported on edges $(x,y)$, $(x+2^{k+1},y+2^k)$ and $(x+2^{k+1}, y+2^{k+1})$ by multiplying $\tilde{S}_{k-1}^{(x,y)}$, $\tilde{S}_{k-1}^{(x+2^k,y+2^{k-1})}$ and $\tilde{S}_{k-1}^{(x+2^k,y+2^k)}$. Indeed, assuming that this support is correct for $S_{k-1}^{(x,y)}$, we have:
\begin{align}
    \begin{split}
        \tilde{S}_k^{(x,y)}
        &= \tilde{S}_{k-1}^{(x,y)} \tilde{S}_{k-1}^{(x+2^k,y+2^{k-1})} \tilde{S}_{k-1}^{(x+2^k,y+2^k)} \\
        &= \Bigl( X_{(x,y)} X_{(x+2^{k},y+2^{k-1})} X_{(x+2^{k}, y+2^{k})} \Bigr) \\
        &\quad \cdot \Bigl( X_{(x+2^k,y+2^{k-1})}
                X_{(x+2^{k+1},y+2^k)}
                X_{(x+2^{k+1}, y+3 \cdot 2^{k-1})} \Bigr) \\
        &\quad \cdot \Bigl( X_{(x+2^k,y+2^k)}
                X_{(x+2^{k+1},y+3 \cdot 2^{k-1})}
                X_{(x+2^{k+1}, y+2^{k+1})} \Bigr). \\
        &= X_{(x,y)} X_{(x+2^{k+1},y+2^k)} X_{(x+2^{k+1},y+2^{k+1})}
    \end{split}
\end{align}

Next, we multiply two successive checks $\tilde{S}_k^{(x,y)}$ and $\tilde{S}_{k+1}^{(x,y)}$, to obtain a new weight-4 check that we call $S_k^{(x',y')}$, with $x'=x+2^{k+1}$ and $y'=y+2^k$, as illustrated in \cref{fig:materialized-symmetry-2-a} and \cref{fig:materialized-symmetry-2-b}. The check $S_k^{(x,y)}$ is supported on the edges $(x,y)$, $(x,y+2^k)$, $(x+2^{k+1},y+2^k)$ and $(x+2^{k+1},y+3 \cdot 2^k)$:
\begin{align}
    S_k^{(x,y)} &= \tilde{S}_k^{(x-2^{k+1}, y-2^k)} \tilde{S}_{k+1}^{(x-2^k, y-2^{k-1})} \\
    &= \left( X_{(x-2^{k+1},y-2^k)} X_{(x,y)} X_{(x,y+2^k)} \right)
    \left( X_{(x-2^{k+1},y-2^k)} X_{(x+2^{k+1},y+2^k)} X_{(x+2^{k+1},y+3 \cdot 2^k)} \right) \\
    &= X_{(x,y)} X_{(x,y+2^k)} X_{(x+2^{k+1},y+2^k)} X_{(x+2^{k+1},y+3 \cdot 2^k)}
\end{align}

\medskip
\textbf{Materialized symmetry on a periodic lattice. }
To obtain a materialized symmetry, we consider products of the form
\begin{align}
    \begin{split}
        M_k^{(x,y)} &= \prod_{j=0}^k S_j^{(x+2^{j+1}-2,y+2^j-1)} \\
        &= \prod_{j=0}^k \left( X_{(x+2^{j+1}-2,y+2^j-1)} X_{(x+2^{j+1}-2,y+2^{j+1}-1)} X_{(x+2^{j+2}-2,y+2^{j+1}-1)} X_{(x+2^{j+2}-2,y+2^{j+2}-1)} \right) \\
        &= X_{(x,y)} X_{(x,y+1)} X_{(x+2^{k+2}-2,y+2^{k+1}-1)} X_{(x+2^{k+2}-2,y+2^{k+2}-1)}
    \end{split}
\end{align}
We see that on an infinite lattice, such a product is supported on four edges: $(x,y)$, $(x,y+1)$, $(x+2^{k+2}-2,y+2^{k+1}-1)$, and $(x+2^{k+2}-2,y+2^{k+2}-1)$. Therefore, on a finite toric lattice of size $L \times L$, if there exists a $k > 0$ such that
\begin{align}
    (x,y) &\equiv (x+2^{k+2}-2,y+2^{k+1}-1) \pmod L \label{eq:condition-for-symmetry-a}
    \\
    (x,y+1) &\equiv (x+2^{k+2}-2,y+2^{k+2}-1) \pmod L \label{eq:condition-for-symmetry-b}
\end{align}
we get the materialized symmetry $M_k=I$. An example is shown in \cref{fig:materialized-symmetry-3}.

\begin{figure}
    \centering
    \begin{subfigure}[b]{0.28\linewidth}
        \includegraphics[width=\linewidth]{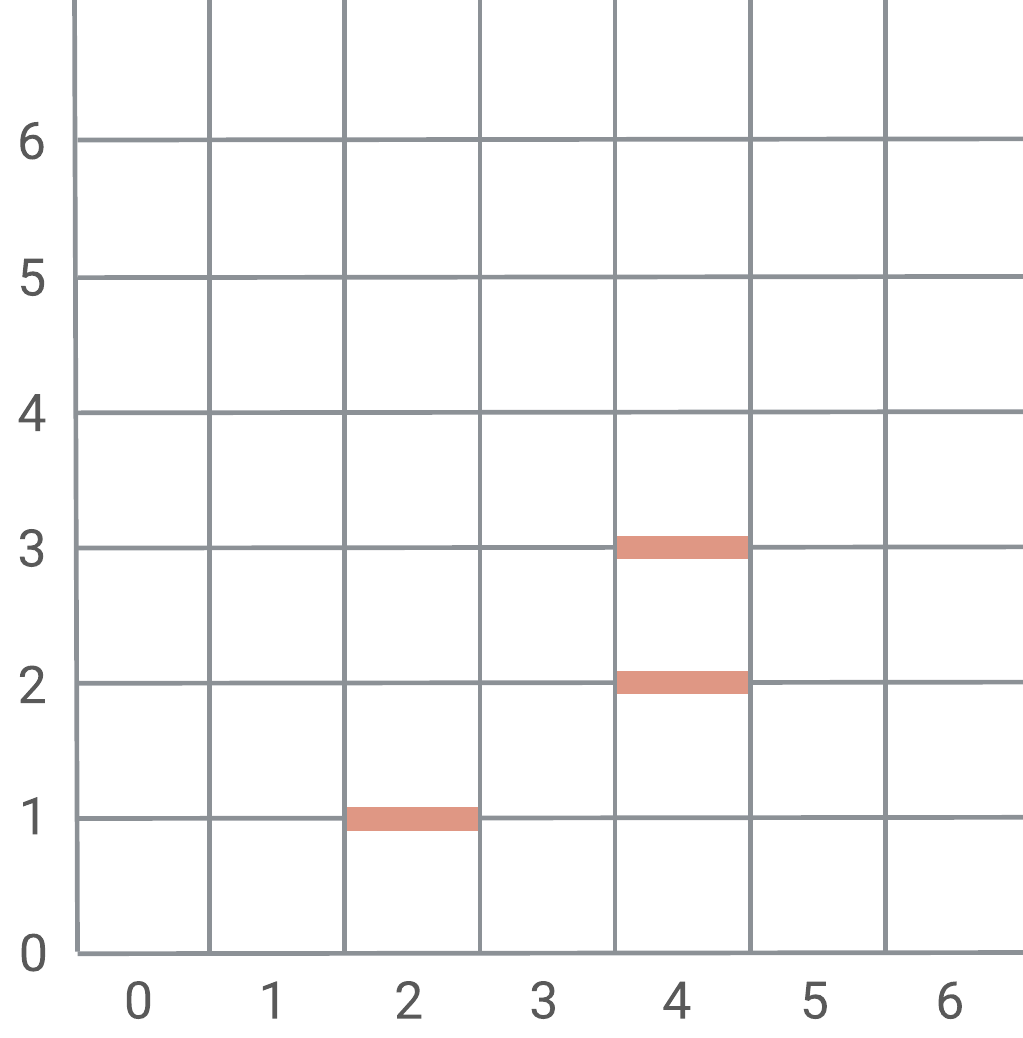}
        \caption{\label{fig:materialized-symmetry-1-a}}
    \end{subfigure}
    \begin{subfigure}[b]{0.28\linewidth}
        \includegraphics[width=\linewidth]{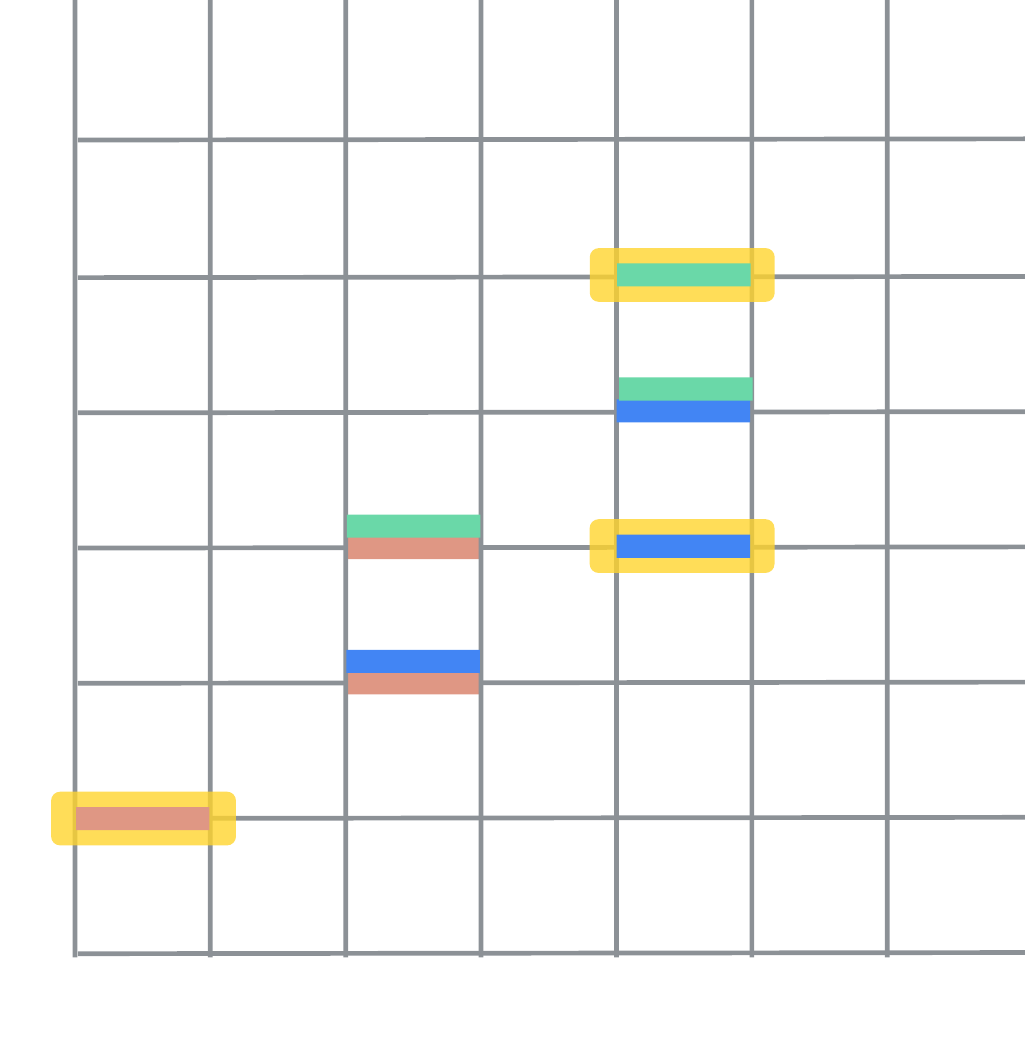}
        \caption{\label{fig:materialized-symmetry-1-b}}
    \end{subfigure}
    \begin{subfigure}[b]{0.38\linewidth}
        \includegraphics[width=\linewidth]{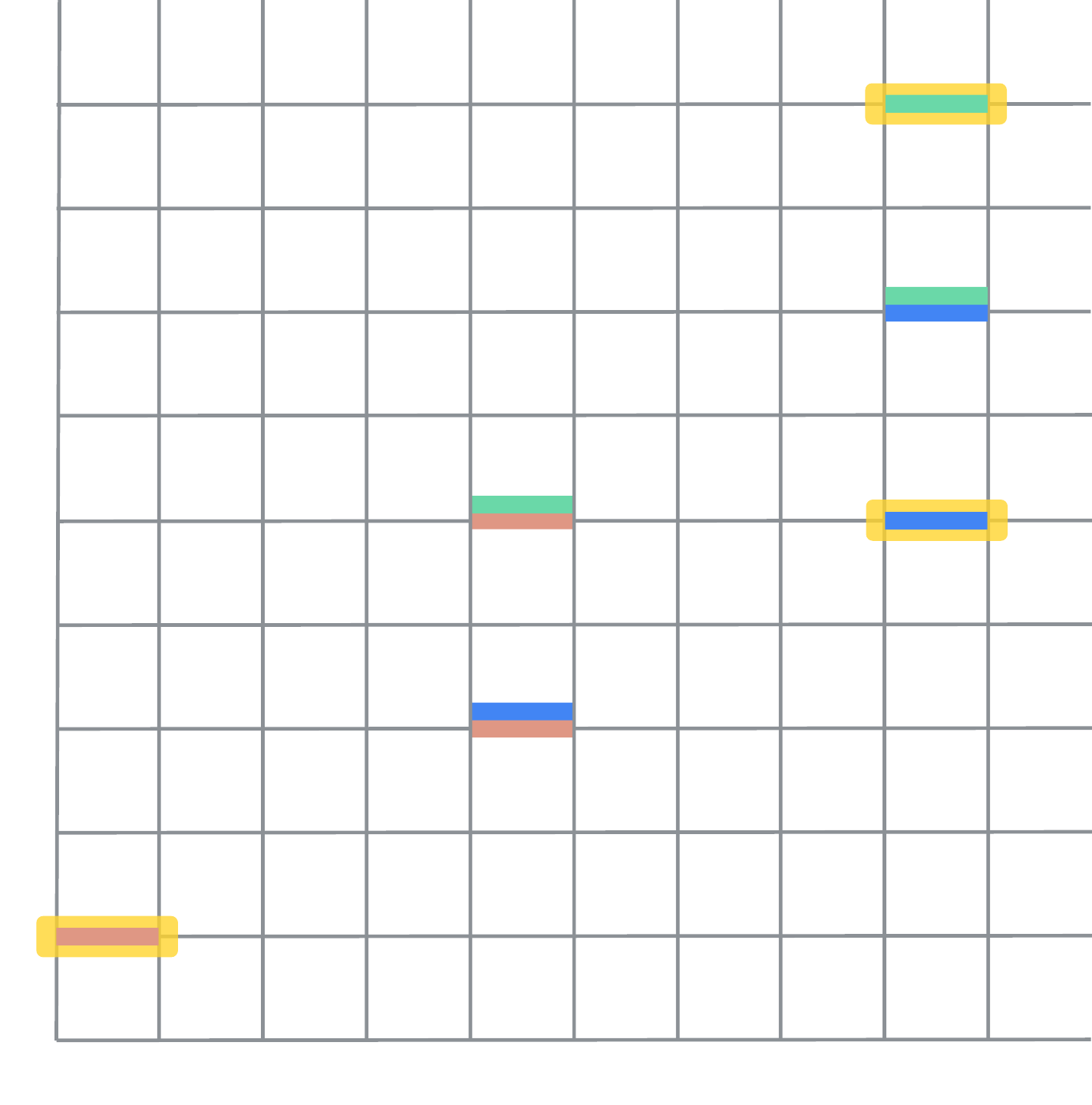}
        \caption{\label{fig:materialized-symmetry-1-c}}
    \end{subfigure}
    \captionsetup{justification=justified,singlelinecheck=false}
    \caption{
        Scaling up infinite-bias stabilizers of the periodic tile code with linear deformation. 
        \textbf{(a)} $X$ part of the linear-deformed $X$ stabilizer, as shown in \cref{fig:linear-deformation} and described in \cref{sec: Clifford Deformed Tile code}. This check detects $Z$ errors on horizontal edges of the lattice. Each horizontal edge has coordinates indicated on the bottom and left sides of the lattice. The check whose bottom qubit has coordinates $(x,y)$ is labeled $\tilde{S}_0^{(x,y)}$. For instance, the check represented here is $\tilde{S}_0^{(2,1)}$.
        \textbf{(b)} By multiplying three checks (red, blue, and green), we obtain a new weight-3 check (highlighted in yellow). This new check is a scaled-up version of the original check, which we call $\tilde{S}_1^{(x,y)}$. The check whose support is highlighted in yellow in the picture is, for instance, $\tilde{S}_1^{(0,1)}$.
        \textbf{(c)} We can repeat this process, multiplying three checks of the form $\tilde{S}_k^{(x,y)}$ to obtain new weight-3 checks $\tilde{S}_{k+1}^{(x,y)}$. We represent here the check $S_2^{(0,1)}$ on a lattice of size $10 \times 10$.
    }
    \label{fig:materialized-symmetry-1}
\end{figure}

\begin{figure}
    \centering
    \begin{subfigure}[b]{0.28\linewidth}
        \includegraphics[width=\linewidth]{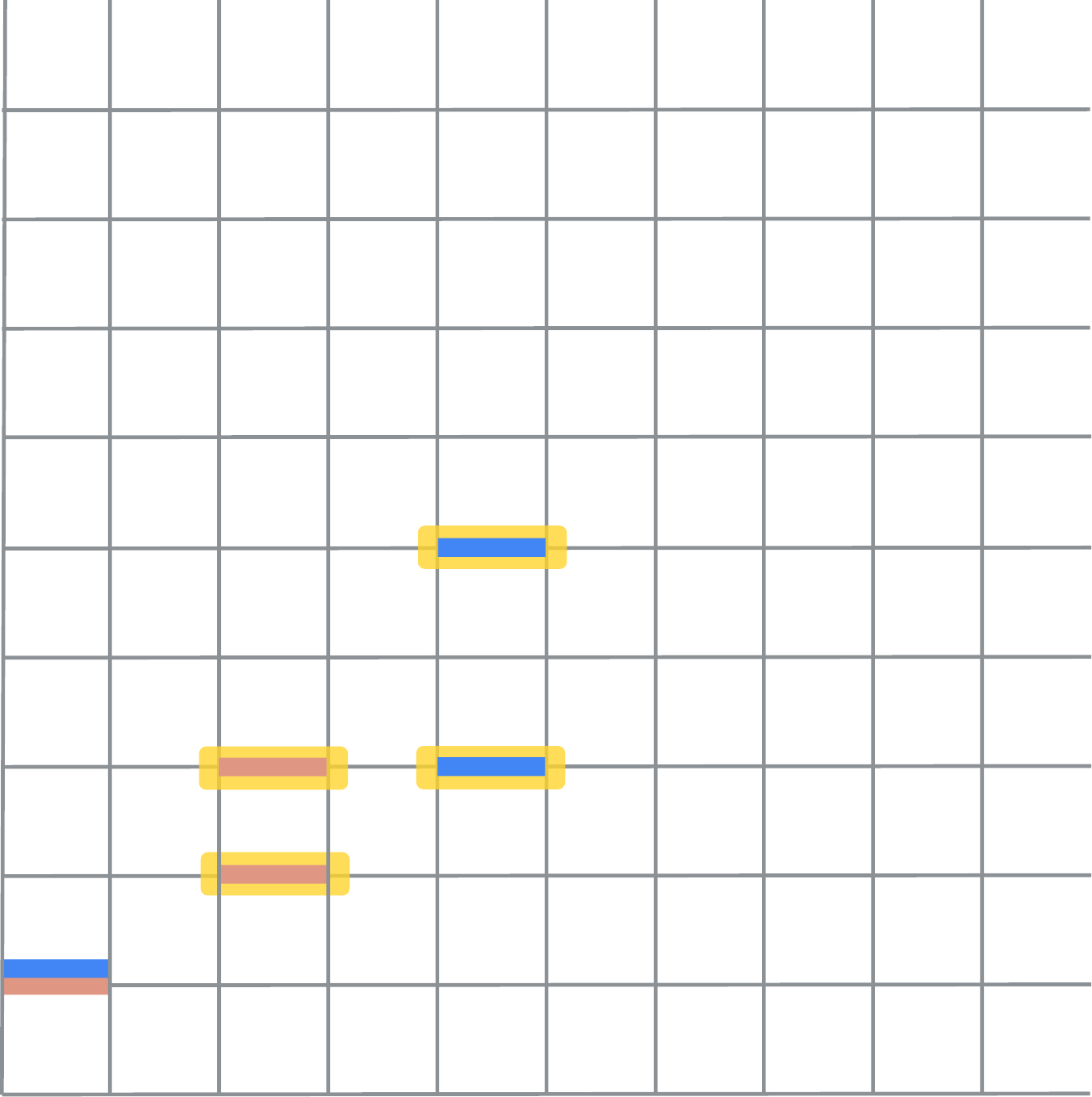}
        \caption{\label{fig:materialized-symmetry-2-a}}
    \end{subfigure}
    \hspace{2em}
    \begin{subfigure}[b]{0.28\linewidth}
        \includegraphics[width=\linewidth]{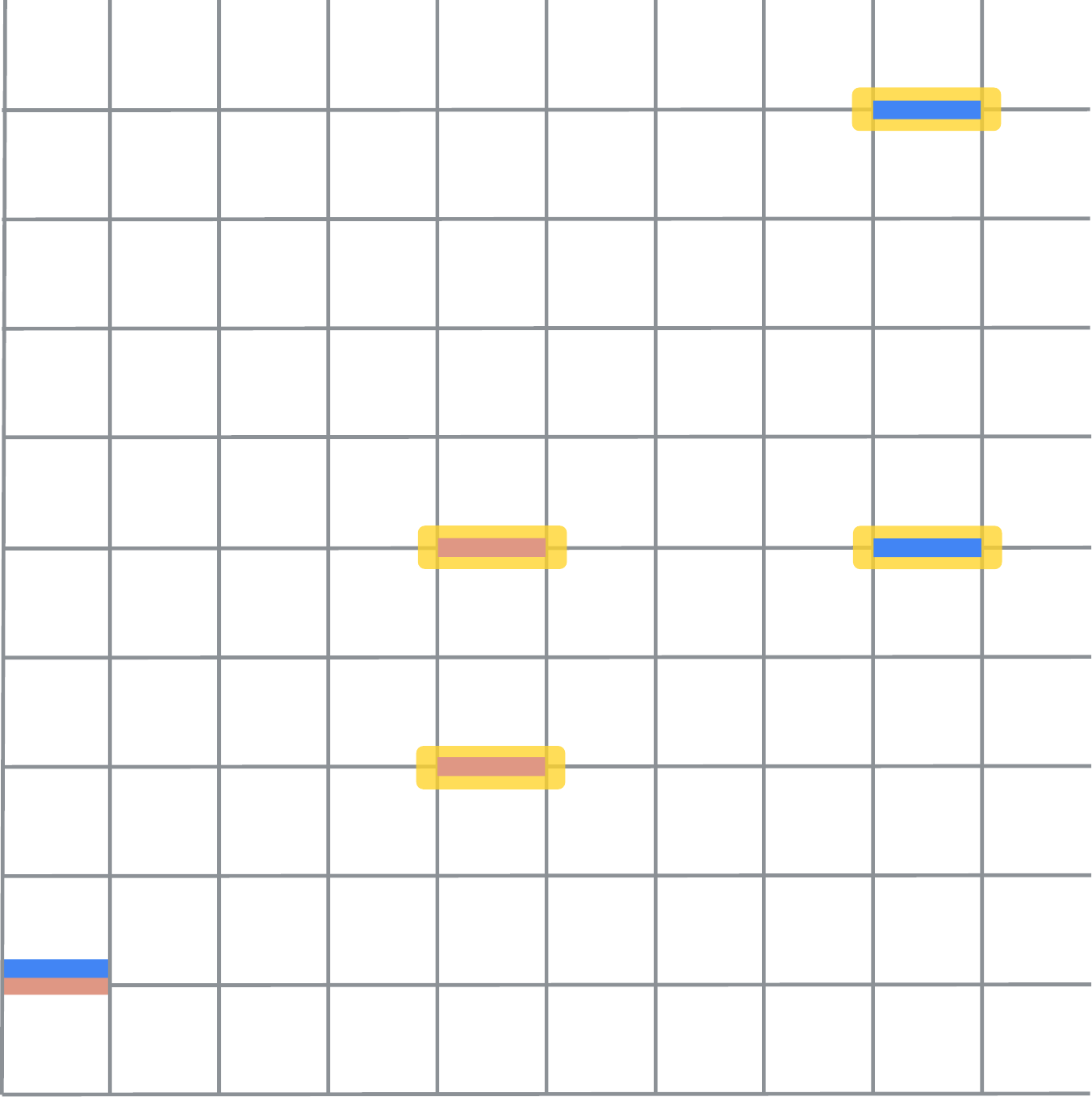}
        \caption{\label{fig:materialized-symmetry-2-b}}
    \end{subfigure}
    \hspace{2em}
    \begin{subfigure}[b]{0.28\linewidth}
        \includegraphics[width=\linewidth]{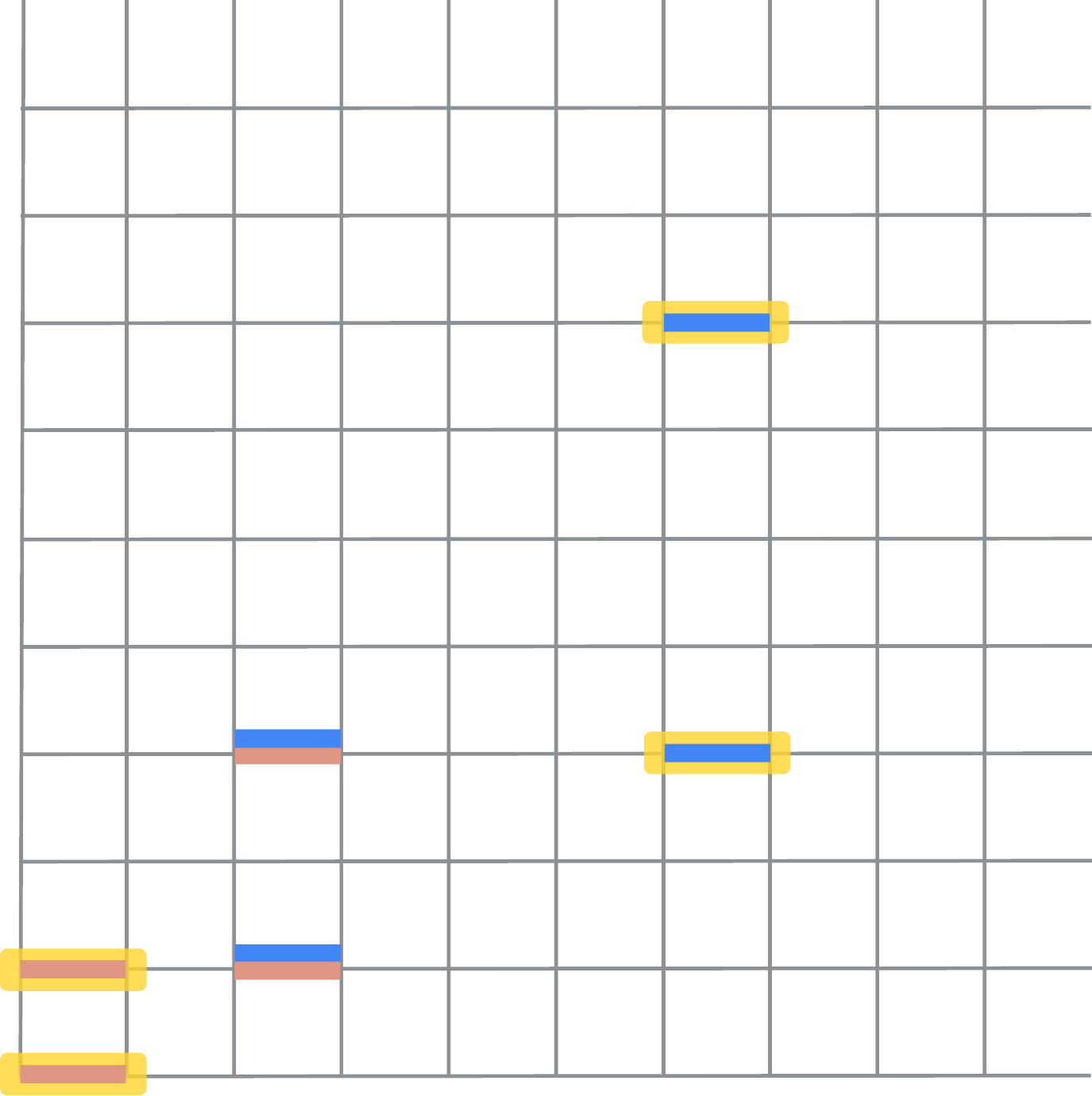}
        \caption{\label{fig:materialized-symmetry-2-c}}
    \end{subfigure}
\captionsetup{justification=justified,singlelinecheck=false}
    \caption{
        Weight-4 checks $S_k^{(x,y)}$ and repetition code on an infinite lattice
        \textbf{(a)} By multiplying $\tilde{S}_k^{(x,y)}$ with $\tilde{S}_{k+1}^{(x,y)}$, we obtain a weight-4 check that we call $S_k^{(x',y')}$, with $x'=x+2^k$ and $y'=y+2^{k-1}$. We represent here $S_0^{(2,2)}=\tilde{S}_0^{(0,1)}\tilde{S}_1^{(0,1)}$, with a support highlighted in yellow. 
        \textbf{(b)} Representation of $S_1^{(4,3)}=\tilde{S}_1^{(0,1)}\tilde{S}_2^{(0,1)}$. 
        \textbf{(c)} Product $M_k^{(x,y)}$ used to form the materialized symmetry. We show here $M_1^{(0,0)}=S_0^{(0,0)} S_1^{(2,1)}$.
        %\textbf{(c)} Materialized symmetry on a $7 \times 7$ periodic lattice: the product of $S_0^{(0,1)}$ (red), $S_1^{(2,2)}$ (blue) and $S_2^{(6,4)}$ (green) is the identity. 
    }
    \label{fig:materialized-symmetry-2}
\end{figure}

\begin{figure}
    \centering
    \includegraphics[width=0.24\linewidth]{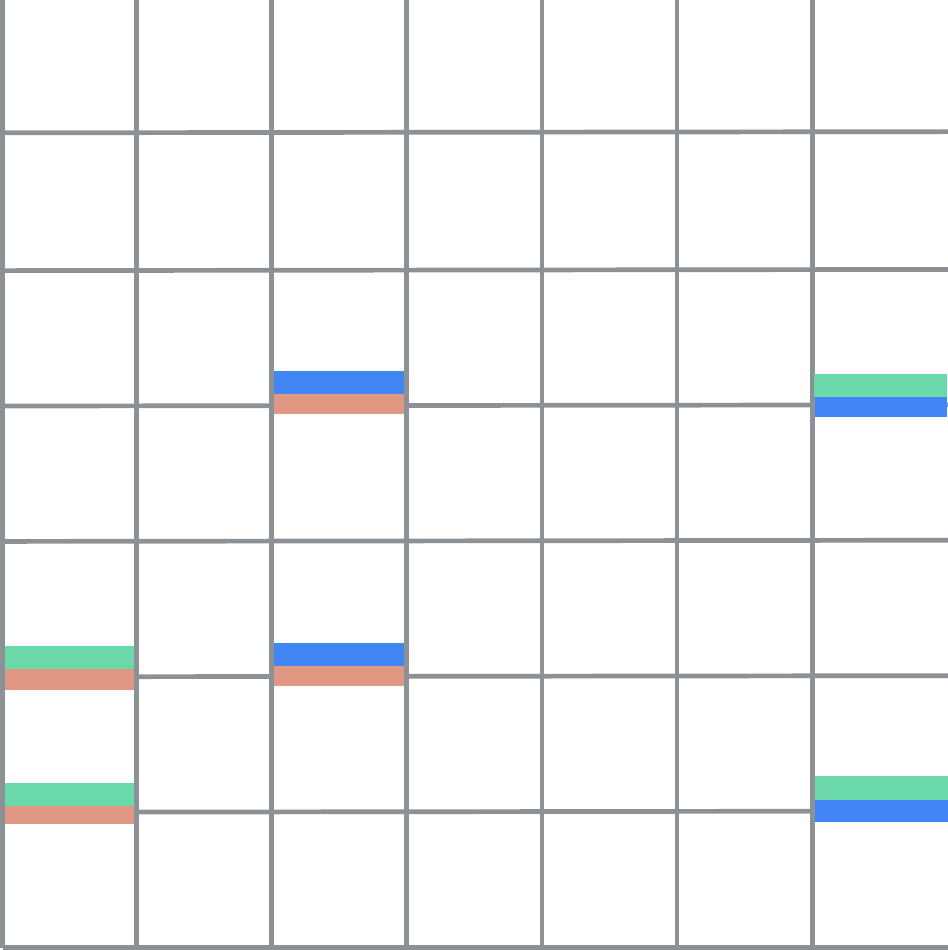}
    \captionsetup{justification=justified,singlelinecheck=false}
    \caption{
        Materialized symmetry on a $7 \times 7$ periodic lattice: the product of $S_0^{(0,1)}$ (red), $S_1^{(2,2)}$ (blue) and $S_2^{(6,4)}$ (green) is the identity. Those three checks form a repetition code, whose variables $X_i$ are the products of the two edges in the same column that are involved in the symmetry: $X_0=X_{(0,1)} X_{(0,2)}$, $X_1=X_{(2,2)} X_{(2,4)}$ and $X_2=X_{(6,4)} X_{(6,1)}$.
    }
    \label{fig:materialized-symmetry-3}
\end{figure}

\cref{eq:condition-for-symmetry-a} and \cref{eq:condition-for-symmetry-b} are equivalent to
\begin{align}
    2^{k+2} - 2 &\equiv 0 \pmod{7\ell}, \label{eq:x-cond}\\
    2^{k+1} - 1 &\equiv 0 \pmod{7\ell}.  \label{eq:y-cond}
\end{align}
Since $L=7\ell$ with $\ell$ odd, we have \(\gcd(2,7\ell)=1\), so we may divide
\eqref{eq:x-cond} by 2 and obtain
\begin{align}
    2^{k+1} \equiv 1 \pmod{7\ell},
\end{align}
which is exactly \eqref{eq:y-cond}.  
Since \(\gcd(2,7\ell)=1\), such a $k$ always exists, and the smallest one is given by
\[
k^\ast = \ord_{7\ell}(2) - 1
\]
where $\ord_{7\ell}(2)$ is the multiplicative order of \(2\) in \(\mathbb{Z}_{7\ell}^\times\).
It corresponds to the number of checks involved in the materialized symmetry, and therefore, to the size of the corresponding repetition code. In order to prove that the code has a threshold of $50\%$, we need to show that those repetition codes scale with the system size.

\textbf{Scaling of $k^\ast$. }
We now study how $\operatorname{ord}_{7\ell}(2)$ scales with system size, using our assumptions about $L=7\ell$:
\begin{itemize}
  \item \(\ell\) is prime and \(\ell \equiv 2 \pmod{3}\),
  \item \(2\) is a primitive root modulo \(\ell\), that is
        \(\operatorname{ord}_\ell(2) = \ell - 1\).
\end{itemize}

Since \(\gcd(7,\ell)=1\), we can apply the Chinese remainder theorem.  
For any integer \(x\) with \(\gcd(x,7\ell)=1\),
\[
\operatorname{ord}_{7\ell}(x) = \operatorname{lcm}
\bigl(\operatorname{ord}_7(x), \operatorname{ord}_\ell(x)\bigr).
\]
In particular,
\[
\operatorname{ord}_{7\ell}(2)
= \operatorname{lcm}\bigl(\operatorname{ord}_7(2), \operatorname{ord}_\ell(2)\bigr).
\]

Since $2^3 = 8 \equiv 1 \pmod{7}$, we have $\operatorname{ord}_7(2) = 3$.
Therefore
\[
\operatorname{ord}_{7\ell}(2) = \operatorname{lcm}\bigl(3,\operatorname{ord}_\ell(2)\bigr).
\]

By Lagrange's theorem, the order \(\operatorname{ord}_\ell(2)\) divides
\(\ell - 1\).  
Since \(\ell \equiv 2 \pmod{3}\), one has $\ell - 1 \equiv 1 \pmod{3}$
so \(3\) does not divide \(\ell - 1\), which implies
$3 \nmid \operatorname{ord}_\ell(2)$
and
\[
\ord_{7\ell}(2) = 3\,\operatorname{ord}_\ell(2).
\]

Under the additional hypothesis that \(2\) is a primitive root modulo \(\ell\),
we have \(\operatorname{ord}_\ell(2) = \ell - 1\), so finally
\[
k^\ast = 3(\ell - 1) - 1
\]
In particular, for these primes \(\ell\), the step at which a materialized symmetry appears in the scaling sequence grows linearly with \(\ell\).

\textbf{Decoding strategy. } Every pair of successive checks of a materialized symmetry $M^{(x,y)}_{k^\ast}$ overlaps on two qubits on the same column, with coordinates of the form $(x+2^{j+1}-2, y+2^j-1)$ and $(x+2^{j+1}-2,y+2^{j+1}-1)$ modulo $L$. Writing
\begin{align}
    X_j = X_{(x+2^{j+1}-2, y+2^j-1)} X_{(x+2^{j+1}-2,y+2^{j+1}-1)}
\end{align}
we have
\begin{align}
    M^{(x,y)}_{k^\ast}= \prod_{j=0}^{k^\ast} \left( X_{j} X_{j+1} \right)
\end{align}
That is, every check involved in the symmetry can be written as a weight-2 check over the new variables $X_j$, and the set of those checks forms a length-$k^\ast$ repetition code. For a small error and a successful decoding of this repetition code, we can therefore infer the value of the product
\begin{align}
    X_{(x+2^{j+1}-2, y+2^j-1)} X_{(x+2^{j+1}-2,y+2^{j+1}-1)}
\end{align}
for all $j \in \{0,\ldots,k^\ast\}$. In particular, picking $j=0$, we now have the value of
\begin{align}
    X_{(x, y)} X_{(x,y+1)}
\end{align}
Solving the repetition code corresponding to every materialized symmetry $M^{(x,y)}_{k^\ast}$, we therefore have a new repetition code for every column of the lattice, with checks of reduced weight $2$. Decoding it allows us to infer the value of every variable $X_{(x,y)}$. 

\textbf{Threshold. } Let us now show that this decoder has a threshold error rate of 50\%. The probability of success of our strategy is lower bounded by the probability that each repetition code decoder succeeds. The first step of our decoder consists of decoding $L^2$ repetition codes of size $\Theta(L)$, and the second step consists of decoding $L$ repetition codes of size $\Theta(L)$. As the success probability of a given repetition code is of the form $1-Ae^{-\alpha L}$ with $\alpha > 0$ and $0 < A < 1$ for any physical error rate below $50\%$, we have the lower bound 
\begin{equation}
    p_{\text{success}} \geq \left(1-e^{-\alpha L}\right)^{L^2} \left(1-e^{-\beta L}\right)^L > \left( 1-L^2e^{-\alpha L} \right) \left( 1-Le^{-\beta L} \right) \xrightarrow[L \rightarrow \infty]{} 1
\end{equation}
on the total success probability. 
Therefore, for any physical error rate below $50\%$, the error probability goes to zero as we increase the lattice size. This finishes the proof that the code has a threshold of $50\%$. 
\end{proof}

\subsection{Belief Propagation Decoding for finite bias}

Belief Propagation with Ordered Statistics Decoding (BPOSD) is a powerful two-stage decoder for LDPC codes. In the first stage, Belief Propagation (BP) operates on the Tanner graph to generate soft reliability estimates for each qubit, providing approximate marginal probabilities of Pauli errors consistent with the measured syndrome. Although BP is computationally efficient and exploits the sparsity of LDPC structures, it can stall or misidentify errors in the presence of short cycles or degeneracy. To address this, BPOSD adds a second stage, Ordered Statistics Decoding (OSD), which uses the BP-derived reliabilities to identify the least trustworthy qubits and performs a targeted search over a small set of low-weight flip patterns. For each candidate, OSD solves the parity-check equations exactly and selects the most likely correction compatible with the syndrome. The combination of fast probabilistic inference (BP) with a refined algebraic search (OSD) yields a decoder that is both efficient and remarkably accurate, approaching maximum-likelihood performance on many LDPC codes.

\textit{Logical failure estimation for CSS variant:}

Let's consider a CSS stabilizer code defined by a binary parity–check matrix $\mathcal{H}_{CSS }=\left( \begin{matrix}
\mathcal{H}_X & 0 \\
0 &  \mathcal{H}_Z
\end{matrix} \right),$ with $\mathcal{H}_X,$ $\mathcal{H}_Z$ $\in \mathbb{Z}^{m\times n}_2$ and logical operators $\mathcal{L}_X,$ $\mathcal{L}_Z$ $\in \mathbb{Z}^{k\times n}_2.$
Each qubit $q$ undergoes a biased Pauli channel,
\begin{align}
    \Pr(E_q = P) = p_q(P), \qquad P \in \{I, X, Y, Z\}.
\end{align}
For an \(n\)-qubit system, the physical error is the Pauli vector
\begin{align}
E = (E_1, \ldots, E_n), \qquad
\Pr(E) = \prod_{q=1}^n p_q(E_q).
\end{align}
In the binary symplectic representation, the same Pauli error is written as
\begin{align}
    e = [ e_X \mid e_Z] \in \mathbb{Z}_2^{2n},
\end{align}
where $
e_X,\, e_Z \in \mathbb{Z}_2^{n}
$
denote the \(X\)- and \(Z\)-supports of the Pauli operator.
Given some error $e,$ one can extract the syndrome $s$ from the relation 
\begin{align}
    \mathcal{H}_{CSS }\odot e=s 
\end{align}
and given that syndrome $s,$ the BPOSD decoder provides the most probable correction operator $c.$ The residual operator
\begin{align}
    r=e+c \hspace{0.1in}\text{(mod 2)} \in \text{ker}(\mathcal{H}_{CSS})
\end{align}
that implies;  $r$ can be either a stabilizer or a logical operator. 

Logical failure is determined by the residual Pauli operator after decoding. In binary symplectic form, write the residual as
\[
r=(r_X\mid r_Z)\in \mathbb{F}_2^{2n},
\]
where $r_X,r_Z\in\mathbb{F}_2^n$ denote the $X$- and $Z$-components of the residual operator.

For the undeformed CSS code, let
\[
L_X\in\mathbb{F}_2^{k\times n},
\qquad
L_Z\in\mathbb{F}_2^{k\times n},
\]
be binary matrices whose rows represent a choice of independent $X$-type and $Z$-type logical operators. In full symplectic form, these are embedded as
\[
(L_X\mid 0),
\qquad
(0\mid L_Z).
\]
The residual operator $r$ is logically trivial if and only if it commutes with all independent logical operators. Equivalently, decoding fails whenever $r$ anticommutes with at least one logical operator, namely whenever
\begin{align}
L_Z\,r_X^T \neq 0 \pmod 2
\qquad \text{or} \qquad
L_X\,r_Z^T \neq 0 \pmod 2.
\end{align}
If either condition holds, the decoder has introduced a non-trivial logical operator and decoding fails. If both expressions vanish, then the residual is logically trivial, i.e.\ it differs from the identity only by a stabilizer.

\textit{Decoding for deformed codes:}

Under a Clifford deformation, we keep the physical noise channel fixed and deform the code instead. For the original CSS code, the parity-check matrix in symplectic form is
\[
H=
\begin{pmatrix}
H_X & 0\\
0 & H_Z
\end{pmatrix},
\]
and the independent logical operators are written as
\[
L=
\begin{pmatrix}
L_X & 0\\
0 & L_Z
\end{pmatrix}.
\]

After deformation, each of these blocks is mapped to a full symplectic matrix. Concretely,
\[
(H_X\mid 0)\longmapsto H_X^{\mathrm{def}},
\qquad
(0\mid H_Z)\longmapsto H_Z^{\mathrm{def}},
\]
and similarly
\[
(L_X\mid 0)\longmapsto L_X^{\mathrm{def}},
\qquad
(0\mid L_Z)\longmapsto L_Z^{\mathrm{def}}.
\]
Thus the full deformed parity-check and logical matrices are obtained by vertical stacking:
\[
H_{\mathrm{def}}=
\begin{pmatrix}
H_X^{\mathrm{def}}\\
H_Z^{\mathrm{def}}
\end{pmatrix},
\qquad
L_{\mathrm{def}}=
\begin{pmatrix}
L_X^{\mathrm{def}}\\
L_Z^{\mathrm{def}}
\end{pmatrix}.
\]
Here \(H_X^{\mathrm{def}},H_Z^{\mathrm{def}}\in\mathbb{F}_2^{m_X\times 2n}\) and \(L_X^{\mathrm{def}},L_Z^{\mathrm{def}}\in\mathbb{F}_2^{k\times 2n}\) are full symplectic matrices. In particular, the superscript “def” indicates the deformed image of the original CSS blocks, not merely the \(X\)- or \(Z\)-part of a deformed matrix.

Using these deformed stabilizers and logical operators, decoding proceeds as above. The residual \(r=(r_X\mid r_Z)\) fails whenever it has a nonzero symplectic product with at least one of the \(2k\) deformed logical rows. Unlike the undeformed CSS case, after deformation, the logical operators are generally mixed \(X\)- and \(Z\)-type operators, so all rows of both \(L_X^{\mathrm{def}}\) and \(L_Z^{\mathrm{def}}\) must be checked.

% \paragraph{Heisenberg picture (deformed noise).}
% Here the parity-check matrices and logical operators remain fixed
% but the physical Pauli channel
% is transformed according to the Clifford deformation, yielding effective
% probabilities;
% \begin{align}
% \{p(I),\; p(X),\; p(Y),\; p(Z)\}
% \xrightarrow[\text{}]{\text{deformation}}
% \{p'(I),\; p'(X),\; p'(Y),\; p'(Z)\}.
% \end{align}

% After that errors are sampled from this deformed channel. However Heisenberg-picture decoding is more subtle than simply deforming the physical noise channel. Although the Clifford deformation can be pushed onto the Pauli channel at the operator level, BP decoding is not invariant under this transformation because its variable and check node update rules depend on the specific X/Z structure of the stabilizers. Clifford rotations mix Pauli components and can introduce multi-qubit correlations, while BP assumes factorized single-qubit priors and fixed parity constraints. Therefore, updating only the error channel is not sufficient: one must also transform the BP message-passing rules to reflect the rotated Pauli basis. Without this modification, the Heisenberg picture does not reproduce the decoding behavior of the Schrödinger picture.
\subsection{Threshold plots for CSS variant and the translational invariant tile codes}
Here we show the threshold plots for CSS tile codes under periodic and open boundary conditions under infinite bias ($\eta = \infty$) with corresponding thresholds $\approx 7.46\%$ and $\approx 10.04\%$ respectively (See \cref{fig:periodic} and \cref{fig:open}). Each data point for $p_{\mathrm{L}}$ is averaged over 40000 Monte Carlo runs. The thresholds are extracted from finite-size scaling. All translational invariant codes show a 50\% threshold under infinite bias ($\eta = \infty$), shown in \cref{fig:ti}.

\begin{figure}
    \centering
    \includegraphics[width=1\linewidth]{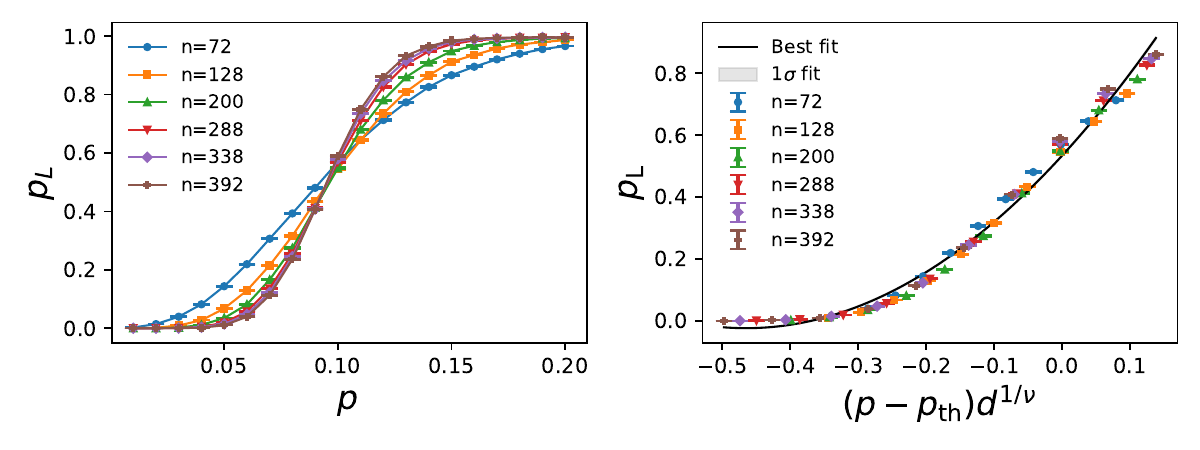}
    \captionsetup{justification=justified,singlelinecheck=false}
    \caption{\textbf{Threshold for CSS tile code under open boundary}: The performance using BP+OSD decoder for the tile codes under open boundary condition. From finite-size scaling, we get the threshold $p_\mathrm{th}$=0.1004}
    \label{fig:open}
\end{figure}

\begin{figure}
    \centering
    \includegraphics[width=1\linewidth]{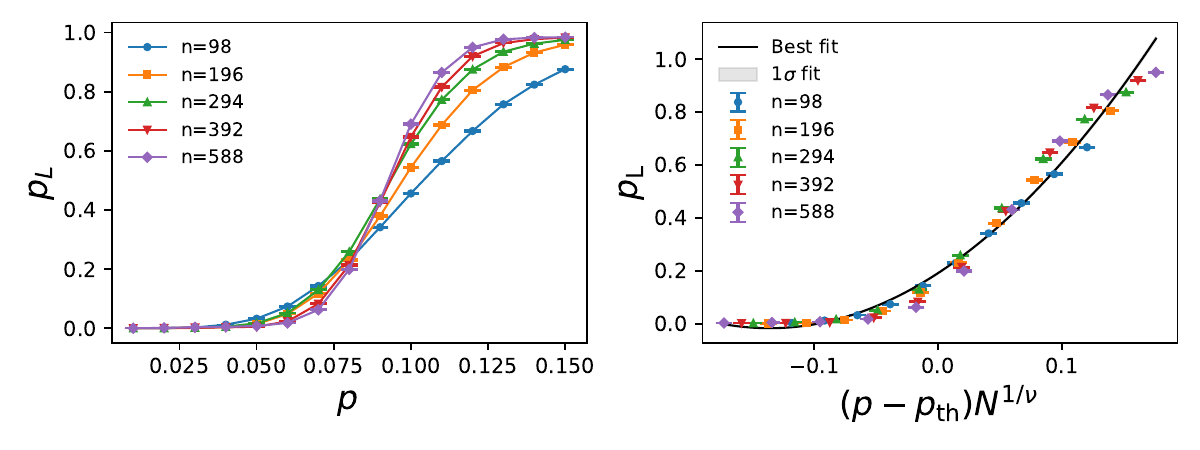}
    \captionsetup{justification=justified,singlelinecheck=false}
    \caption{\textbf{Threshold for CSS tile code under periodic boundary}: The performance using BP+OSD decoder for the tile codes under periodic boundary condition.  From finite-size scaling, we get the threshold $p_\mathrm{th}$=0.0746}
    \label{fig:periodic}
\end{figure}

\begin{figure}
\centering
\includegraphics[width=1
\columnwidth]{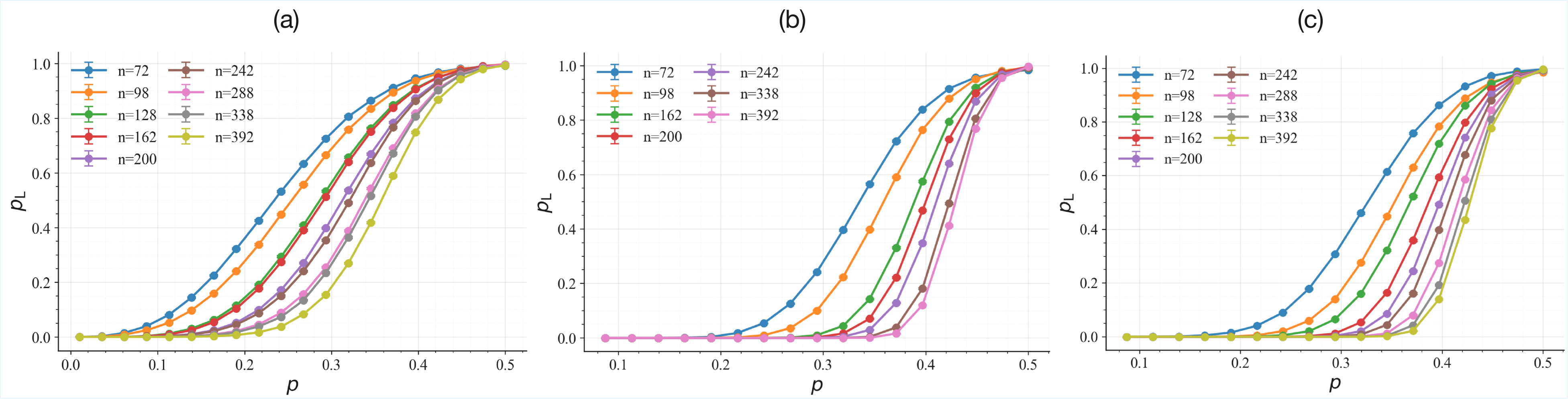}
\captionsetup{justification=justified,singlelinecheck=false}
\caption{\textbf{Threshold for three translational invariant codes}: Performance of CDTCs in the infinite-bias limit is shown for the (a) linear deformation, the  (b) XY deformation, and (c) the deformation (associated with the phase point (0.25,0.5)) whose unit cell is defined in ~\cref{fig:TI_middle_unit_cell}. All three exhibit a 50\% threshold.
\label{fig:ti}}
\end{figure}

% \begin{figure}
% \centering
% \includegraphics[width=1
% \columnwidth]{TI_025_05_without_MWBLO_final.pdf}

% \caption{ Fig. \cref{fig:TI_025_05_without_MWBLO_final} (a) shows the performance of CDTCs under TI (0.25,0.5) for infinite bias. The unit cell for this code is described in Fig. \cref{fig:TI_025_05_without_MWBLO_final} (b) where $YZ (XZ)$ deformations are performed on the red(green) qubits. The black qubits remain invariant. 
% \label{fig:TI_025_05_without_MWBLO_final}}
% \end{figure}
\vspace{0.2in}
For the below subsections, most of the associated plots are provided in Supplemental Material (SM) \cite{supplement}.
\subsection{Phase Diagram}
\label{app:phase-diagram}
As we discussed, we employ tile codes because they constitute a widely accessible family of LDPC codes with a constant number of logical qubits~\cite{steffan2025tile}. 
Throughout this work, we focus on the family encoding eight logical qubits, as described in  \cref{sec: Clifford Deformed Tile code}, using open boundary conditions. 
All translationally invariant (TI) Clifford-deformed tile codes exhibit clear threshold behavior in this setting. 
Consequently, open-boundary tile codes are used in all numerical studies except for the phase-diagram analysis, where periodic boundary conditions are adopted.

Under infinite bias with random Clifford deformations, open-boundary codes exhibit a pronounced finite-size crossover (pseudo-threshold) region in the physical error-rate range
\(0.2 \lesssim p \lesssim 0.3\) for most phase points. 
This behavior is attributed to a decoder-induced finite-size crossover, likely arising from belief-propagation instabilities caused by boundary-induced short cycles in the Tanner graph. 
The anomaly disappears under periodic boundary conditions, where the graph becomes more regular, and belief propagation remains stable.

For the same code family, imposing periodic boundary conditions reduces the number of logical qubits from eight to six and restricts the allowed lattice sizes to
\(L_x, L_y \in 7\mathbb{N}\), where \(\mathbb{N}\) denotes the set of positive integers.

Here we show the performance plots \footnote{, In most of the plots of the following subsections, the number of qubits is expressed by $N$,} for different phase points $(\Pi_{XZ},\Pi_{YZ})$ for random CDTCs under infinite bias for the phase diagram shown in  \cref{fig:phase_diagram}. The plots exhibit logical
error rate $p_L$ versus physical error rate $p$. The threshold is the maximum physical error rate $p_{\text{th}}$ such that, for all $p<p_{\text{th}}$, the logical failure probability decreases as the code distance (or system size) increases. Each data point is averaged over
50,000 Monte Carlo runs of the BPOSD decoder. 

% Fig. \cref{fig:50_percent_set1}, \cref{fig:50_percent_set2},\cref{fig:50_percent_set3} and \cref{fig:near_50_percent}
Figs. S1, S2, S3, and S4  (See SM \cite{supplement}) show the plots of the 50\% threshold corresponding to the blue regime in the phase diagram. For the phase points in Fig. S4,
%\cref{fig:near_50_percent},
The largest code considered exhibits deviations from the expected scaling behavior, which is possibly due to the decoder artifacts.  For a few phase points lying on the yellow boundary—such as $(0.5,0),(0,1)$ the threshold extracted from the BP-OSD decoder is ambiguous or not sharply defined. Although we expect a 50\% threshold for those points, if we study certain finer subfamilies of tile codes with this periodic boundary condition. Note that there is a 50\% threshold for translational invariant tile code under linear and XY deformation, where we used \textit{open} boundary condition. The thresholds for some points on the yellow boundary are illustrated in Fig. S5.
%\cref{fig:boundary_yellow}.
As we move away from the blue region, the threshold systematically decreases. Interestingly, however, the phase diagram contains several iso-threshold contours, where distinct parameter regions exhibit similar threshold behavior. The corresponding performance plots for these iso-threshold regimes are shown in Figs. S6-S10

\subsection{MWBLO Result}
\label{app:mwbo}
In this subsection, we provide a detailed characterization of the MWBLO (Minimum Weight basis Logical Operator) for different CDTCs under infinite $Z$-bias. Throughout the plots, $|\text{BLO}|= |\mathcal{B}_Z(n)|$ denotes the number of logical operators in the basis, while the $Z$-distance refers to the weight of the minimum-weight logical operator within that basis. The overlap per logical quantifies the total overlap associated with the minimum-weight logical.

We first present the MWBLO data for several phase points within the 50\%-threshold region (see Fig. 
S11 and S12 in SM \cite{supplement}),
%\cref{fig:MWBLO_data_with_50_threshold_set1} and %\cref{fig:MWBLO_data_with_50_threshold_set2}), 
followed by some representative points located along the yellow boundary in Fig. S13.
%\cref{fig:MWBLO_data_with_50_threshold_boundary}.
The results in Fig.  S11 and S12 
%\cref{fig:MWBLO_data_with_50_threshold_set1} and %\cref{fig:MWBLO_data_with_50_threshold_set2}
are consistent with the analytical proof establishing a $50\%$ threshold. 
In particular, each plot satisfies the conditions
\[
 0<\alpha\le 1 .
\]

As we move from the center of the blue regime to the yellow boundary (as in phase diagram \cref{fig:phase_diagram}), the $|\text{BLO}|$ starts increasing very slowly, maintaining the 50\% threshold criteria, and the corresponding results are summarized in Fig. S14.
%\cref{fig:middle_to_yellow_boundary}. 
Beyond the yellow boundary, the MWBLO data supporting \cref{fig:support_for_threshold_decrease_final} illustrate that as we move further away from the blue regime of the phase diagram, $|\text{BLO}|$ grows rapidly and the corresponding threshold drops below the 50\% criterion. This behavior, along with the average $Z$-distance and total overlap per logical are shown in Fig. S16.
%\cref{fig:logical_vs_threshold}.
Different iso-threshold contours outside the blue regime can be explained by the Figs. S18 and S19.
%\cref{fig:supporting_data_1}  and \cref{fig:supporting_data_2}.
As discussed in \cref{sec:mwblo}, when $|\text{BLO}|$ starts to increase significantly, the \cref{thm:subexp_logicals_threshold} becomes invalid. In this context, we study the slope $s$ in scaling $|\text{BLO}| = sn + c.$ We show how the average $|\text{BLO}|$ increases as the threshold drops, as well as each iso-threshold contour reflects the same scaling of $|\text{BLO}|$ shown in  \cref{fig:support_for_threshold_decrease_final}.

\begin{figure}[t]
\centering
\includegraphics[width=1\columnwidth]{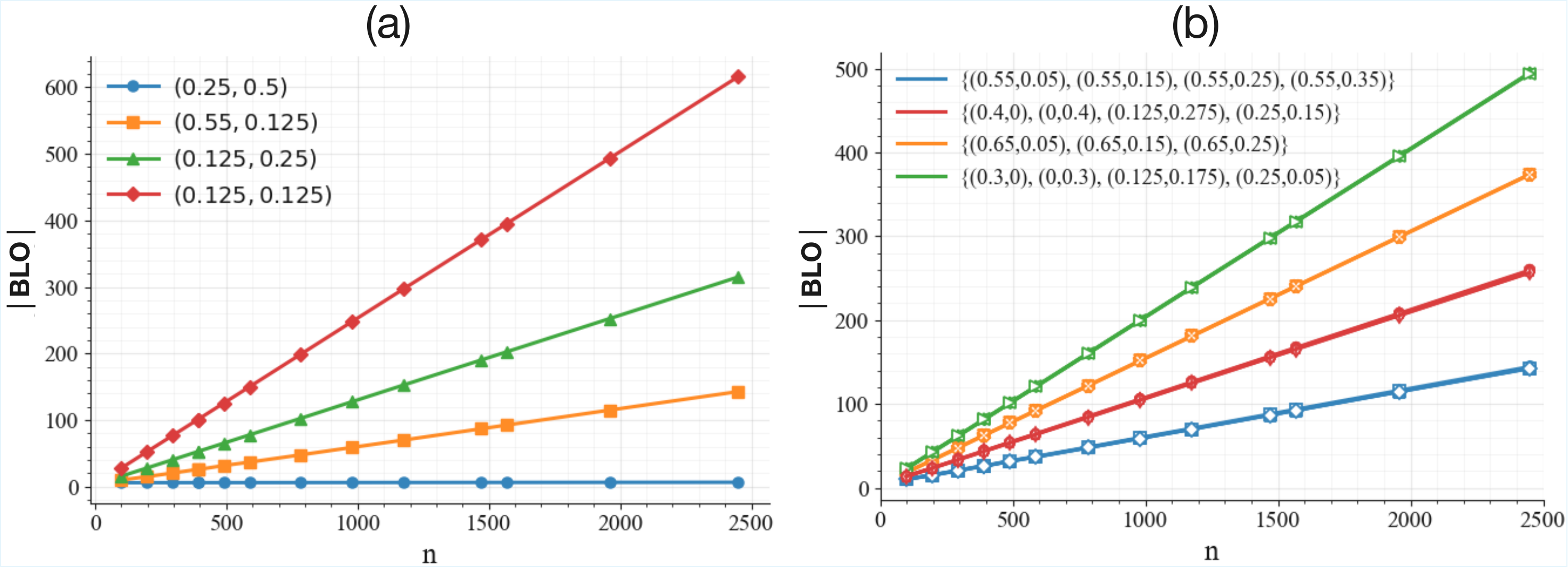}
\captionsetup{justification=justified,singlelinecheck=false}
\caption{
Monotonic growth of BLO size as the threshold decreases and scaling across iso-threshold contours. (a) As one moves away from the $50\%$ threshold regime in ~\cref{fig:phase_diagram}, the decoding threshold decreases, and the number of BLO elements grows sharply with system size. (b) Within each iso-threshold contour, the BLO size exhibits the same scaling with system size, indicating a shared slope $s$. Different contours correspond to different $s$, separating the phase diagram into families with distinct asymptotic BLO growth.}
\label{fig:support_for_threshold_decrease_final}
\end{figure}

Earlier, we described the procedure used to minimize the weight of each basis logical operator. Specifically, for a given basis logical operator, we collect all nearby $Z$ stabilizers and form all possible linear combinations of these stabilizers. Each such combination is multiplied by the logical operator under consideration, and the representative with the minimum weight is selected.
As one moves into lower-threshold regimes, the number of relevant $Z$ stabilizers grows rapidly, making the enumeration of all linear combinations computationally expensive. Consequently, MWBLO data are presented only for smaller system sizes at phase points corresponding to lower thresholds.

Apart from the phase diagram, there are three translational invariant tile codes for the open boundary. The MWBLO data clearly justifies the 50\% threshold of those three codes (See Fig. S17). 
%\cref{fig:TI_MWBLO_final}). 
On the other hand, for the code family under consideration,  the TI linear and $XY$ variants under periodic boundary do not exhibit clear thresholds from the decoding of the BP-OSD decoder, and the associated MWBLO plots also show fluctuating behaviors (shown in Fig. S15).
%\cref{fig:TI_periodic}). 
But it is possible for some finer sub-families, for example, in the case of the sub-family in the weight reduction subsection (see \cref{app:weight-reduction}), we get a 50\% threshold for the linear deformed tile code under periodic boundary.

% \subsection{Finite Bias Code Capacity Model}

\subsection{Stim-based circuit-level simulations and generalization to tile codes}
\label{app:stim}

To perform circuit-level threshold studies for the Clifford-deformed tile codes, we use the high-performance stabilizer-circuit simulator \texttt{Stim}. Stim enables extremely fast Monte Carlo simulation of large stabilizer circuits through two core design features: firstly 
an optimized \textit{stabilizer-tableau backend} that compiles the circuit into an inverse tableau, allowing deterministic measurements to be updated in linear---rather than quadratic---time; and secondly 
a \textit{Pauli-frame sampling engine} that reuses a compiled tableau to propagate batches of Pauli frames, enabling millions of independent shots to be generated at negligible marginal cost.

These optimizations, combined with cache-friendly memory layouts and explicit SIMD vectorization in Stim's C++ backend, allow simulation of surface-code-scale memory circuits with $10^6$--$10^7$ gates and hundreds of qubits. Stim is therefore ideally suited for the large-scale circuit-level threshold sweeps required in this work.

\subsubsection{Detector error models.}
Stim exposes the full space--time structure of a stabilizer circuit via the \texttt{DETECTOR}, 
\texttt{SHIFT\_COORDS}, and \texttt{OBSERVABLE\_INCLUDE} instructions.
A \textit{detector} specifies a parity constraint on a set of measurement outcomes and \textit{fires} whenever that parity flips.
From a detector-annotated circuit, Stim's compiler produces a \textit{detector error model} (DEM): a hypergraph-like representation that records which elementary errors flip which detectors and logical observables.
The DEM can be passed directly to the decoders, including BPOSD, union-find, and MWPM-based decoders, eliminating all hand-written syndrome logic.
Thus, the circuit itself fully specifies the syndrome-extraction process.

\subsubsection{Baseline generator.}
Our circuit generator follows the ``head--body--tail'' structure used in Stim's built-in surface-code templates:

\textit{Head:} initializes data qubits and ancillas, assigns planar coordinates via \texttt{QUBIT\_COORDS}, and defines the first layer of detectors;

\textit{Body:} repeats the stabilizer-measurement cycle consisting of (i) ordered CNOT/CZ sub-rounds and (ii) ancilla measurements. Consecutive rounds are linked using \texttt{SHIFT\_COORDS} and differential detectors;

\textit{Tail:} performs final data measurements and appends terminating detectors and \texttt{OBSERVABLE\_INCLUDE} instructions to define logical observables.

A key architectural feature is the clean separation between \textit{geometry} (qubit coordinates, ancilla placement, and interaction offsets) and \textit{circuit logic} (round structure, detectors, and measurement scheduling).
Once the geometry is specified, the same head--body--tail logic produces valid Stim circuits for any lattice configuration.

\subsubsection{Generalization to CSS tile codes.}
To extend the surface-code generator to the Clifford-deformed tile code, we retain the same head--body--tail architecture but replace the geometry layer.

\textit{Data-qubit graph.}
For a tile code of size $(l,m)$ with block parameter $B=3$, we enumerate horizontal and vertical edges of the underlying square lattice, keep only those participating in at least one stabilizer, and assign each retained edge a data-qubit index.
Each qubit is then assigned planar coordinates (scaled by an appropriate factor) so that all data and ancillas occupy distinct integer lattice sites.

\textit{Ancilla placement.}
Red (X-type) and blue (Z-type) stabilizers are generated from fixed local offset templates applied to bulk anchors and boundary anchors.
A  real-axis displacement ($\pm 1$) is added between the two ancilla sets purely to avoid coordinate collisions, ensuring that Stim can assign unique detector coordinates.

\subsubsection{Finite-bias noise and Clifford deformation.}
Stim supports biased Pauli noise through \texttt{PAULI\_CHANNEL\_1}.  
For a total error rate $p$ and Z-bias $\eta$, we use the decomposition
\begin{align}
    p_Z = p \frac{\eta}{\eta + 2}, \hspace{0.2in}
    p_X = p_Y = p \frac{1}{\eta + 2},
\end{align}
consistent with the biased-Pauli model used in the main text.
Noise can be inserted before measurement, after reset, after Clifford gates, or at the end of each stabilizer-measurement round.

We apply the deformations on top of the CSS variant (See  \cref{fig:circuit_css}) by using suitable Clifford gates on the qubits and/or between qubits and checks throughout the full circuit.  Because all stabilizer rounds (including the linear and CSS-like deformations) consist only of Clifford gates, the entire deformed tile-code circuit remains a stabilizer circuit and is exactly representable in Stim.
Stim therefore produces an exact DEM, and the BPOSD decoder can be applied directly to the resulting detection graph.
The same pipeline applies to all deformation families studied, with only the interaction schedules differing.

Overall, Stim’s detector-based circuit compiler combined with our geometry-parametric tile-code generator enables large-scale circuit-level threshold studies of Clifford-deformed tile codes under finite Z-bias, without requiring any manual syndrome bookkeeping.

\subsubsection{Results}
From the circuit-level simulations, we extract the thresholds from the $p_L$ vs $p$  curves for the tile codes under four different Clifford variants. The codes in descending order of their thresholds are the Linear, TI (0.25, 0.5), TI-XY, and CSS variants, with corresponding thresholds of 1.5\%, 1.04\%, 0.814\%, and 0.634\%, respectively, as shown in  \cref{fig:circuit_level_linear_css_10000_final} and \cref{fig:circuit_level_XY_middle_10000_final}.
We evaluated the performance of the four CDTCs under finite biases
$\eta = 10{,}000$, $1{,}000$, $100$, $50$, and $10$, and extracted the corresponding
thresholds using finite-size scaling (See Fig. S24 - S31 in SM \cite{supplement}).
% \cref{fig:circuit_level_linear_scaling_1}-\cref{fig:circuit_level_TI_middle_scaling_2}.)
The variation of subthreshold logical error rates $(p_L)$ against distance $(d)$ for different bias $\eta$ and physical error rate ($p$) is summarized in Fig.  S20 -S23 of SM \cite{supplement}.
% \cref{fig:subthreshold_006_final}-\cref{fig:subthreshold_003_final}. 
The best performing code is TI(0.25,0.5).

\begin{figure}
\centering
\includegraphics[width=0.85\columnwidth]{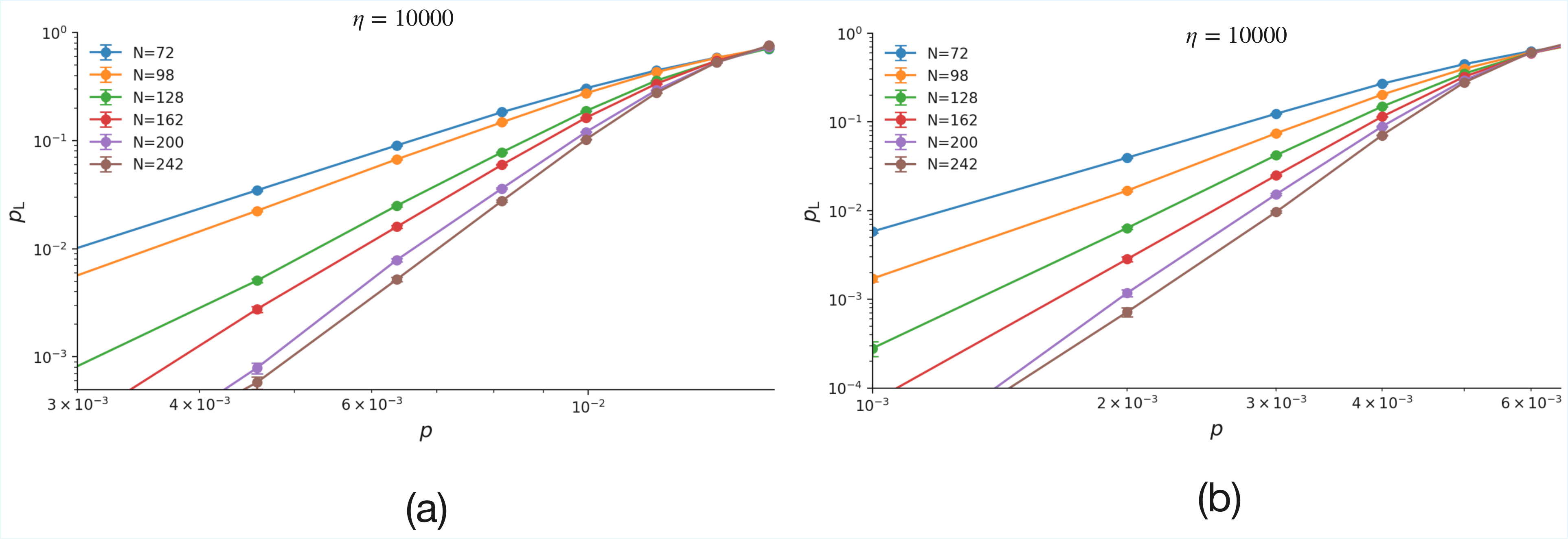}
\captionsetup{justification=justified,singlelinecheck=false}
\caption{ Circuit-level performance of open tile codes decoded with the StimBPOSD decoder using 100,000 trials and eight rounds. Fig (a) shows results for linear deformation, while Fig (b) presents results for CSS tile codes. The corresponding thresholds for bias $\eta=10000$ are nearly 1.5\% for the linear case and 0.634\% for the CSS case respectively. 
\label{fig:circuit_level_linear_css_10000_final}}
\end{figure}
\begin{figure}
\centering
\includegraphics[width=0.85\columnwidth]{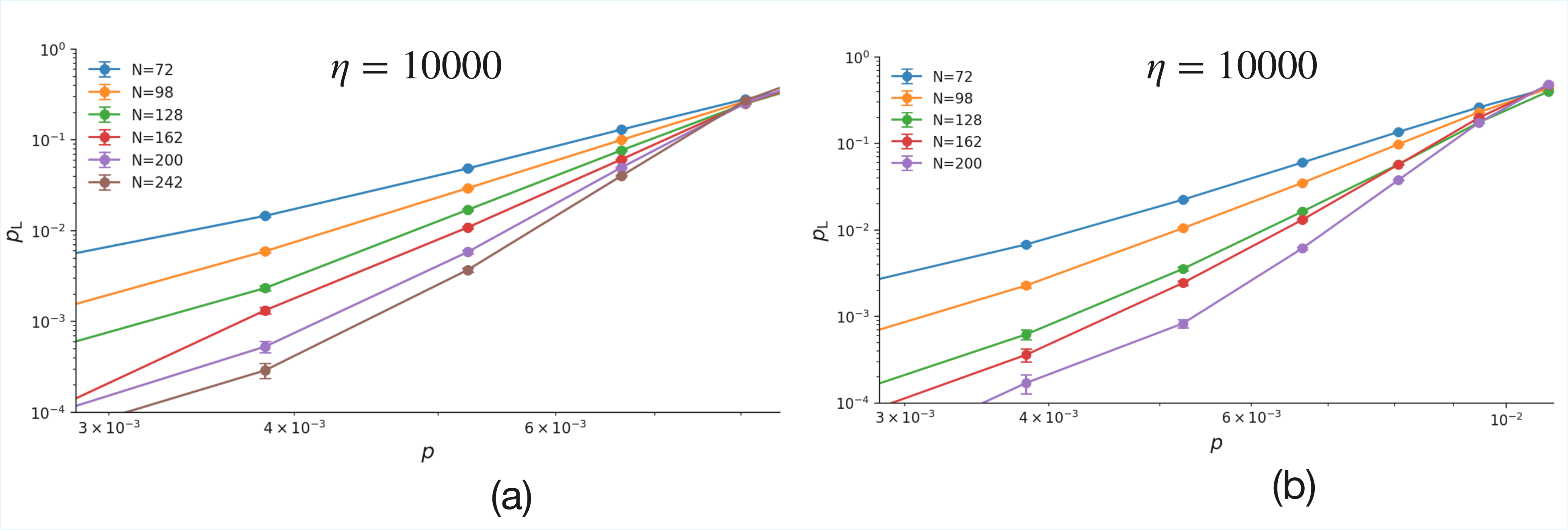}
\captionsetup{justification=justified,singlelinecheck=false}
\caption{ Circuit-level performance of open tile codes decoded with the StimBPOSD decoder using 100,000 trials and eight rounds. Fig (a) shows results for TI XY tile codes, while Fig (b) presents results for TI (0.25,0.5) codes. The corresponding thresholds for bias $\eta=10000$ are nearly 0.814\% for the TI XY case and 1.04\% for the TI (0.25,0.5) case, respectively. 
\label{fig:circuit_level_XY_middle_10000_final}}
\end{figure}

\begin{figure}
\centering
\includegraphics[width=0.69\columnwidth]{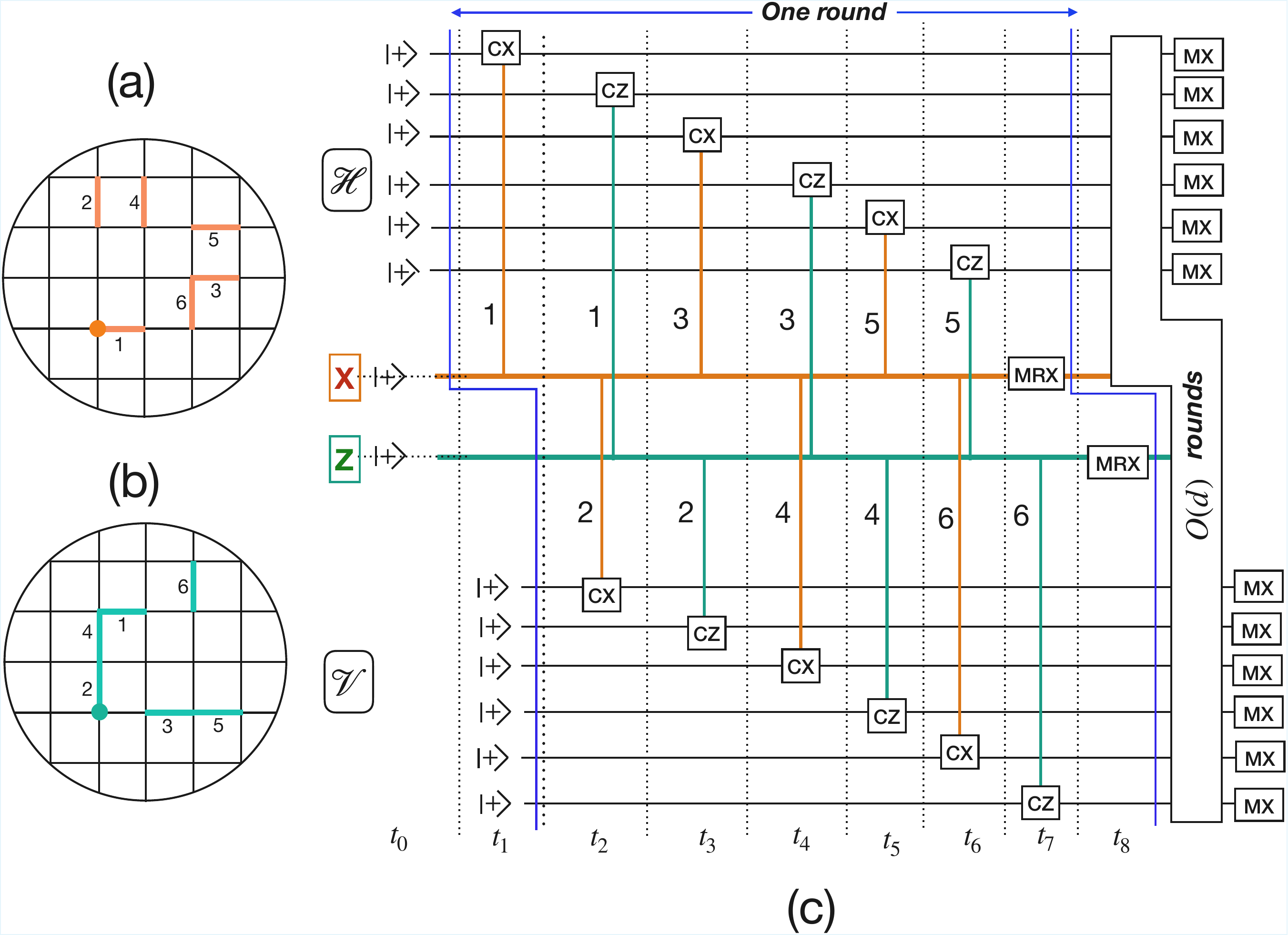}
\captionsetup{justification=justified,singlelinecheck=false}
\caption{Single-round syndrome-measurement circuit in (c) for the CSS tile code, involving $CX$ and $CZ$ gates. Each $X$ and $Z$ check is connected (via $CX$ and $CZ$ gates, respectively) to three vertical and three horizontal qubits. Here $\mathcal{H}$ and $\mathcal{V}$ represent the horizontal and vertical qubits, respectively, in the tile code. All data qubits and checks are initialized in the $\ket{+}$ state. MRX indicates measurement and reset in $\ket{+}$ basis, whereas  MX corresponds to only measurement in the same basis. The qubit ordering for $X$ and $Z$ stabilisers is shown in (a) and (b), respectively.}
\label{fig:circuit_css}
\end{figure}

\subsection{Optimization of the Syndrome Extraction Schedule}
\label{sec:appendix_optimization}

When implementing stabilizer‑measurement circuits, it is important to avoid space–time error mechanisms that can reduce the effective circuit distance. Such mechanisms arise from error propagation involving both data and ancilla qubits during syndrome extraction, and they depend sensitively on the ordering of two‑qubit gates.

Here, for the weight-6 checks of the open tile code, we can choose a disjoint measurement schedule in which the $X$ checks and $Z$ checks are measured sequentially for each round, which would require 12 entangling time steps for each round. Although this circuit is distance preserving, the increased depth of the syndrome extraction circuit makes the data qubits susceptible to significant idle errors. Under standard circuit-level noise models, this accumulation of idle errors can substantially reduce the fault-tolerant threshold.

Therefore, to mitigate idle errors and optimize the threshold, we try to compress the syndrome-extraction circuit by executing the entangling gates for the $X$ and $Z$ checks concurrently. However, if $X$ and $Z$ checks, which do not mutually commute on shared data qubits, are interleaved in an arbitrary order, the resulting circuit need not preserve stabilizer eigenvalues at intermediate times. In such schedules, a valid code state can be mapped outside the intended stabilizer subspace before readout, which in the detector‑error‑model description manifests as “non‑deterministic detectors,” i.e., detectors whose outcomes are not fixed by the encoded state even in the absence of physical errors.
% \AD{We first say "compress the syndrome extraction circuit into 7 time steps" but later says "we finally get 3 out of the 720 schedules that work" and "optimally scheduled, 6-step circuits." Clarify whether the circuit has 6 or 7 time steps. Just use the final schedule and we don't need to tell what was the research obstacles; describe the final tool and the optimal schedules}
% \SD{The circuit has 7 steps. I have corrected the error. Also corrected the language accordingly.}

To identify a fault-tolerant syndrome-extraction schedule, we held the measurement schedule for the $X$ checks fixed and then performed an exhaustive search over the $6! = 720$ distinct sequence permutations of the $Z$ checks.  
We use \texttt{Stim} \cite{gidney2021stim} to construct the complete detector error model (DEM) for each circuit and check whether they produce non-deterministic detectors. After filtering out the invalid schedules, we obtain 6 of the 720 schedules that work.
Among them we select one 7-step circuit shown in the \ref{fig:circuit_css} and exactly calculate the circuit distance and check that it is indeed distance preserving. We use this particular circuit schedule for all threshold and sub-threshold logical error-rate calculations presented in this work.
\begin{figure*}[t]
\centering
\begin{minipage}{0.95\textwidth}
{\captionsetup[algorithm]{aboveskip=2pt,belowskip=2pt}

\hrule height 0.4pt
\captionof{algorithm}{Syndrome-extraction schedule optimization}
\hrule height 0.4pt
}
\vspace{2pt}
\label{alg:schedule_optimization}
\begin{algorithmic}[1]
\Require Base spatial ordering of Z-stabilizers $S_Z$, base spatial ordering of X-stabilizers $S_X$
\Ensure List of valid schedules $\mathcal{V}$ with deterministic detectors
\State Initialize $\mathcal{V} \leftarrow []$
\ForAll{permutation $P$ of the six Z-check layers from $S_Z$}
    \State Initialize circuit $C \leftarrow \emptyset$
    \State Append first X-check entangling layer (CNOT) per $S_X$; insert tick
    \For{$k = 2, \ldots, 6$}
        \State Append $k$-th X-check layer (CNOT) per $S_X$
        \State Append $(k-1)$-th Z-check layer (CZ) per $P$; insert tick
    \EndFor
    \State Append final Z-check layer (CZ) per $P_6$; insert tick
    \State Append measurements and detector definitions
    \If{DEM builds successfully using Stim}
        \State Append $P$ and $C$ to $\mathcal{V}$
    \Else
        \Comment{Non-deterministic detectors}
        \State Discard schedule 
    \EndIf
\EndFor
\State \Return $\mathcal{V}$
\end{algorithmic}
\vspace{2pt}
\hrule height 0.4pt
\end{minipage}
\end{figure*}
\subsection{Effective Pauli Bias Estimation} \label{app:effective_bias}

In this segment of the appendix, we describe how the effective bias of the compiled circuits in ~\cref{sec:qubit-platforms} is obtained. We propagate Pauli error strings through the circuit via Monte Carlo sampling, collecting them at the end to compute effective error probabilities, as described in Algorithm~\cref{alg:2}. 

The circuit is represented as a gate sequence $\mathcal{G}$ of $m$ tuples $(G_k,\mathbf{q}_k)$, where $G_k$ is a Clifford gate acting on qubits $\mathbf{q}_k$. $\mathcal{Q}$ denotes the full qubit set. Before starting, we group gates that can be applied simultaneously into gate layers. Each $e_j$ denotes a single-qubit Pauli sampled from $\{I,X,Y,Z\}$ with specified probabilities.

Initialization and measurement errors are modeled in the $X$ basis. In this case, only errors that flip the prepared or measured state contribute, which correspond to $Y$ and $Z$ errors. This leads to the effective Pauli channel
\[
\mathbf{w}_{\mathrm{init/meas}}=(0,0,p_Y+p_Z,1-(p_Y+p_Z)).
\]

During circuit evolution, each gate layer is applied ideally via Clifford conjugation. In the absence of hardware-specific gate noise, a single-qubit Pauli channel with weights $\mathbf{w}_{\mathrm{circ}}=(p_X,p_Y,p_Z,p_I)$ is applied independently to all qubits.

When incorporating platform-specific noise, each gate $G_k$ is replaced by its noisy implementation $\widetilde{G}_k = \mathcal{N}_{G_k} \circ G_k$, where $\mathcal{N}_{G_k}$ is a Pauli channel. After ideal propagation through a gate layer, a Pauli error $E_k$ is sampled independently from $\mathcal{N}_{G_k}$ for each gate in the layer and applied to its support $\mathbf{q}_k$. Additionally, single-qubit Pauli errors are sampled and applied to qubits that are idle during the entire layer, i.e., qubits not acted on by any gate in that layer.

After collecting $N_{\mathrm{samp}}$ Pauli error strings propagated to the end of the circuit, the occurrences of each Pauli type are counted across all samples and qubits. This yields effective single-qubit marginal error probabilities $\bar{p}_\sigma$, defined as the probability that a randomly chosen qubit in a randomly sampled final Pauli string carries label $\sigma\in\{I,X,Y,Z\}$. The effective bias is then reported as $\eta_{\mathrm{eff}} = \frac{\bar{p}_Z}{\bar{p}_X + \bar{p}_Y}$.

All Pauli operations are performed using Stim~\cite{gidney2021stim}. This algorithm implements a circuit-level Pauli noise model and maps it to effective single-qubit marginal statistics. More realistic models could incorporate time-dependent error rates to account for idle periods and differing gate durations.
\begin{figure*}[t]
\centering
\begin{minipage}{0.95\textwidth}
{\captionsetup[algorithm]{aboveskip=2pt,belowskip=2pt}

\hrule height 0.4pt
\captionof{algorithm}{Monte-Carlo Pauli propagation}
\hrule height 0.4pt
}
\vspace{2pt}
\label{alg:2}
\begin{algorithmic}[1]
\Require Gate sequence $\mathcal{G}$ grouped into layers $\{\mathcal{L}_1,\ldots,\mathcal{L}_T\}$; qubits $\mathcal{Q}$; physical error rate $p$; bias $\eta$;
platform noise model $\mathcal{N}$; platform label $\mathrm{plat}$; number of samples $N_{\mathrm{samp}}$
\Ensure Effective Pauli probabilities $\{\bar p_I,\bar p_X,\bar p_Y,\bar p_Z\}$

\State $p_Z \gets p\frac{\eta}{\eta+2}$,\quad
       $p_X \gets \frac{p}{\eta+2}$,\quad
       $p_Y \gets \frac{p}{\eta+2}$,\quad
       $p_I \gets 1-(p_X+p_Y+p_Z)$
\State $\mathbf{w}_{\mathrm{circ}} \gets (p_X,p_Y,p_Z,p_I)$
\State $\mathbf{w}_{\mathrm{init/meas}} \gets (0,0,p_Y+p_Z,1-(p_Y+p_Z))$
\State Initialize counts $c_I,c_X,c_Y,c_Z \gets 0$

\For{$n \gets 1$ \textbf{to} $N_{\mathrm{samp}}$}
    \State $P \gets I^{\otimes |\mathcal{Q}|}$

    \Comment{Initialization noise}
    \ForAll{$j \in \mathcal{Q}$}
        \State sample $e_j \sim \mathbf{w}_{\mathrm{init/meas}}$
        \State $P \gets e_j\,P$
    \EndFor

    \For{$t \gets 1$ \textbf{to} $T$}
        \State Let $\mathcal{L}_t=\{(G_k,\mathbf{q}_k)\}$

        \Comment{Ideal Clifford propagation} %\SD{(Heisenberg)}}\PM{Good!}
        \State $P \gets \Big(\prod_{(G_k,\mathbf{q}_k)\in\mathcal{L}_t} G_k\Big)\,
                     P\,
                     \Big(\prod_{(G_k,\mathbf{q}_k)\in\mathcal{L}_t} G_k^\dagger\Big)$

        \If{$\mathrm{plat} \neq \mathrm{ideal}$}
            \Comment{Apply gate Pauli channels}
            \ForAll{$(G_k,\mathbf{q}_k)\in\mathcal{L}_t$}
                \State sample $E_k \sim \mathcal{N}_{G_k}$ on $\mathbf{q}_k$
                \State $P \gets E_k\,P$
            \EndFor

            \Comment{Idle noise on inactive qubits}
            \State $\mathcal{Q}_{\mathrm{active}} \gets \bigcup_k \mathbf{q}_k$
            \ForAll{$j \in \mathcal{Q}\setminus \mathcal{Q}_{\mathrm{active}}$}
                \State sample $e_j \sim \mathbf{w}_{\mathrm{circ}}$
                \State $P \gets e_j\,P$
            \EndFor
        \Else
            \Comment{Uniform circuit-level Pauli noise}
            \ForAll{$j \in \mathcal{Q}$}
                \State sample $e_j \sim \mathbf{w}_{\mathrm{circ}}$
                \State $P \gets e_j\,P$
            \EndFor
        \EndIf
    \EndFor

    \Comment{Measurement noise}
    \ForAll{$j \in \mathcal{Q}$}
        \State sample $e_j \sim \mathbf{w}_{\mathrm{init/meas}}$
        \State $P \gets e_j\,P$
    \EndFor

    \Comment{Accumulate single-qubit Pauli-label counts}
    \State Update $c_I,c_X,c_Y,c_Z$ from the letters of $P$
\EndFor

\State \Return $\bar p_\sigma = c_\sigma/(c_I+c_X+c_Y+c_Z)$ for $\sigma\in\{I,X,Y,Z\}$
\end{algorithmic}
\vspace{2pt}
\hrule height 0.4pt
\end{minipage}
\end{figure*}
\subsection{Native gates for the studied qubit platforms}\label{app:native_gates}

In this segment, we review the native entangling gates for the hardware platforms considered in the main text by specifying the effective Hamiltonian that is activated during the entangling operation. In general, a driven gate is described by the time-evolution operator
\begin{equation}
\hat{U}(t_g)
=
\mathcal{T}\exp\!\left(-\frac{i}{\hbar}\int_0^{t_g}\hat{H}(t)\,dt\right),
\end{equation}
where \(\mathcal{T}\) denotes time ordering, \(t_g\) is the gate duration, and \(\hat{H}(t)\) may include single-qubit drives and/or two-qubit interactions. In all the gate models considered here, the dominant time dependence enters through a pulse envelope (Rabi frequency) \(\Omega(t)\) multiplying a fixed generator \(\hat{G}\), i.e., \(\hat{H}(t)=\hbar\,\Omega(t)\,\hat{G}\). In this case \([\hat{H}(t),\hat{H}(t')]=0\), time ordering is unnecessary, and the evolution reduces to a rotation
\begin{equation}
\hat{U}(t_g)=\exp\!\left(-i\,\theta\,\hat{G}\right),
\qquad
\theta \equiv \int_0^{t_g}\Omega(t)\,dt,
\end{equation}
so that the gate angle is set by the pulse area \(\theta\).
In practice, native entangling gates are often combined with single-qubit rotations to construct Clifford gates. Since single-qubit Cliffords such as the Hadamard gate exchange \(X\) and \(Z\), they do not preserve a strong dephasing bias in general; consequently, even if the native entangling interaction is bias preserving, an arbitrary Clifford circuit compiled from it need not be.

We consider three qubit platforms: (i) trapped-ion qubits (with native entangling interactions locally equivalent to CX or CZ), (ii) superconducting qubits (native iSWAP, compiled to CX), and (iii) neutral-atom qubits (native CZ via Rydberg blockade).

\subsubsection{Trapped-ion qubits: native CX from a M\o lmer--S\o rensen interaction}
A commonly used entangling interaction in trapped ions is the M\o lmer--S\o rensen (MS) gate, which arises from driving excitations on two ions with two lasers of different frequency, but near the atomic energy gap \cite{Ion-traps-cnot-2000}. In an appropriate interaction picture and for suitable drive phases, the MS interaction can be viewed as an effective Ising coupling along the \(X\) axis, e.g.
\begin{equation}
\hat{H}_{\mathrm{MS}}
=
\frac{\hbar\Omega}{2}\,\hat{X}\!\otimes\!\hat{X}.
\label{eq:H_MS}
\end{equation}
Evolving under \(\hat{H}_{\mathrm{MS}}\) for an appropriate duration produces an entangling gate locally equivalent to \(R_{XX}(\pi/2)\), from which a CX gate can be obtained using single-qubit \(X\) and \(Y\) rotations. Because the interaction in Eq.~\eqref{eq:H_MS} contains \(X\), it does not commute with local \(Z\) operators. As a result, a pure dephasing-biased error model is not preserved by this entangling interaction.

\subsubsection{Trapped-ion qubits: native CZ from a Liebfried--S\o rensen interaction}
Trapped ions can also realize an entangling phase gate via the Liebfried--S\o rensen (LS) interaction \cite{Ion-traps-cz-2000}, which implements a \(ZZ\) phase up to local \(Z\) rotations. The LS gate can be represented by an effective Ising interaction of the form
\begin{equation}
\hat{H}_{ZZ}
=
\frac{\hbar\chi}{2}\,\hat{Z}\!\otimes\!\hat{Z},
\label{eq:H_ZZ_ions}
\end{equation}
where \(\chi\) depends on parameters such as the Lamb--Dicke parameter, the Rabi frequency, the number of ions, and the detuning of the applied fields \cite{Ion-traps-cz-2003}. Evolving under Eq.~\eqref{eq:H_ZZ_ions} for the appropriate duration yields a \(ZZ\)-phase gate (and hence a CZ gate up to single-qubit \(Z\) rotations). Importantly, \(\hat{Z}\!\otimes\!\hat{Z}\) commutes with single-qubit \(Z\) operators, so dephasing-biased noise is preserved at the level of the native entangling interaction.

\subsubsection{Superconducting qubits: native iSWAP (\(XY\) exchange interaction) and compilation to CX}
For capacitively coupled superconducting qubits, an effective exchange (``\(XY\)'') interaction is commonly engineered. A minimal model is
\begin{equation}
\hat{H}_{XY}
=
\frac{\hbar J}{2}\Big(\hat{X}\!\otimes\!\hat{X}+\hat{Y}\!\otimes\!\hat{Y}\Big),
\label{eq:H_XY}
\end{equation}
where \(J\) depends on the coupling capacitance. This model generates an iSWAP gate at a fixed interaction time, which can be compiled into CX using single-qubit \(Z\) and \(X\) rotations \cite{Superconducting-2010}. Because Eq.~\eqref{eq:H_XY} contains transverse terms \(X\) and \(Y\), it does not commute with local \(Z\), and therefore it does not preserve a dephasing bias under propagation through the native entangling interaction.

\subsubsection{Neutral-atom qubits: native CZ from Rydberg blockade}
Neutral-atom platforms can implement a CZ gate using the Rydberg blockade mechanism: excitation of the control atom to a Rydberg level shifts the doubly excited state \(\lvert rr\rangle\) by a large interaction energy, suppressing resonant excitation of the target atom and thereby producing a conditional phase \cite{Neutral-atoms-2000,Neutral-atoms-2001}.
A standard protocol uses a three-pulse sequence: \(\pi\) on the control, \(2\pi\) on the target, and \(\pi\) on the control \cite{Neutral-atoms-2010}. In the basis \(\{\lvert0\rangle,\lvert1\rangle,\lvert r\rangle\}\), a simple effective Hamiltonian model for the driven dynamics is
\begin{align}
\hat{H}_{\mathrm{c}} &= \frac{\hbar\Omega}{2}\Big(\lvert r\rangle\langle 1\rvert \otimes I + \mathrm{h.c.}\Big),\\
\hat{H}_{\mathrm{t}} &= \frac{\hbar\Omega}{2}\Big(I \otimes \lvert r\rangle\langle 1\rvert + \mathrm{h.c.}\Big)
+ \hbar V\,\lvert rr \rangle\langle rr \rvert,
\end{align}
where \(\Omega\) is the Rabi frequency and \(V\) is the blockade shift. After projecting back to the computational subspace, the net operation is a CZ gate (up to single-qubit phase corrections that can be implemented with \(Z\) rotations) \cite{Neutral-atoms-2010}. Since the entangling action is a controlled-\(Z\) operation, this native interaction is compatible with preserving a strong dephasing bias.

\subsection{Numerically simulated Pauli channels of the native gates (CX or CZ) across qubit platforms}
\label{app:noise_channels}

As described in the main text, we incorporate platform-specific physics by replacing the circuit-level Pauli channel acting on the qubits participating in each native entangling gate with numerically derived Pauli channels. In this appendix, we summarize how these Pauli channels are extracted from microscopic simulations for each qubit platform introduced in ~\cref{app:native_gates}, and we list the resulting channels used in the Pauli-propagation algorithm from ~\cref{app:effective_bias}. The explicit Pauli error distributions used for single- and two-qubit gates are summarized in ~\cref{tab:pauli_channels_all}.

To construct a Pauli channel associated with a native gate, we first simulate the ideal (noise-free) gate evolution using the von Neumann equation, obtained by omitting the dissipative terms from the Lindblad master equation. From this evolution we compute the Pauli transfer matrix (PTM) of the ideal native gate, denoted \(R_{\mathrm{ideal}}^{\mathrm{nat}}\).

To incorporate noise during gate evolution, we simulate the Lindblad master equation
\begin{equation}
\begin{aligned}
\dot{\hat{\rho}}(t)
&=
-\frac{i}{\hbar}\big[\hat{H}(t),\hat{\rho}(t)\big]
+\sum_n \lambda_n\,\mathcal{D}[\hat{L}_n]\!\big(\hat{\rho}(t)\big), \\
\mathcal{D}[\hat{L}]\!\big(\hat{\rho}(t)\big)
&=
\hat{L}\hat{\rho}(t)\hat{L}^\dagger
-\frac12\left\{\hat{L}^\dagger \hat{L},\,\hat{\rho}(t)\right\},
\end{aligned}
\label{eq:Lindblad}
\end{equation}
where \(\hat{L}_n\) are the coupling operators and \(\lambda_n\) are the corresponding dissipator coefficients. We take \(\hat{L}_n\) from the Pauli basis: \(\{X,Y,Z\}\) for single-qubit gates and \(\{I,X,Y,Z\}^{\otimes 2}\setminus\{I^{\otimes 2}\}\) for two-qubit gates, corresponding to a Markovian noise model.

Bias toward dephasing is imposed through the dissipator coefficients by setting \(\lambda_{IZ}+\lambda_{ZI}+\lambda_{ZZ}=\Lambda\), while all remaining two-qubit Pauli dissipators are assigned weight \(\Lambda/\eta_{\mathrm{sys}}\). The overall scale \(\Lambda\) is chosen such that the resulting noisy evolution yields the desired gate fidelity.

From this simulation we obtain the PTM of the noisy native gate, \(R_{\mathrm{noisy}}^{\mathrm{nat}}\). The PTM of the corresponding noise channel is then defined by factoring out the ideal evolution,
\begin{equation}
R_{\mathrm{noise}}^{\mathrm{nat}}
=
R_{\mathrm{noisy}}^{\mathrm{nat}}\,
\left(R_{\mathrm{ideal}}^{\mathrm{nat}}\right)^{-1},
\end{equation}
following Ref.~\cite{bias_preserving_CNOT_cat_qubits_2020}.

We represent quantum channels in the Pauli basis via their PTM \(R\), whose elements are given by
\begin{equation}
R_{ij}
=
\frac{1}{d}\,\mathrm{Tr}\!\left(P_i\,\mathcal{E}(P_j)\right),
\end{equation}
where \(d=2^n\) and \(\{P_i\}\) is the \(n\)-qubit Pauli basis. In practice, we evaluate \(\mathcal{E}(P_j)\) by evolving each Pauli operator under the same dynamics used for density matrices, implemented numerically using QuTiP~\cite{Qutip-2012}.

For Pauli channels, the diagonal elements of \(R_{\mathrm{noise}}^{\mathrm{nat}}\) determine the Pauli error probabilities \(\{p_E\}\), which are obtained via a symplectic Walsh--Hadamard transform,
\begin{equation}
p_E
=
\frac{1}{4^n}\sum_{P\in\mathcal{P}_n} H_{P,E}\,R_{PP},
\qquad
H_{P,E}\equiv (-1)^{\langle P,E\rangle},
\end{equation}
where \(\mathcal{P}_n\) is the \(n\)-qubit Pauli group modulo phases and \(\langle P,E\rangle\) is the symplectic commutation indicator.

These probabilities define the stochastic Pauli channel associated with a gate \(G\) on a given platform,
\begin{equation}
\mathcal{N}_{G,\mathrm{plat}}(\rho)
=
\sum_{E\in\mathcal{P}_n} p_E^{(G,\mathrm{plat})}\, E\,\rho\,E.
\end{equation}
Here, \(E\) runs over the Pauli basis acting on the support of the gate: \(E\in\{I,X,Y,Z\}\) for single-qubit gates and \(E\in\{I,X,Y,Z\}^{\otimes 2}\) for two-qubit gates. The probabilities \(p_E^{(G,\mathrm{plat})}\) therefore form a distribution over Pauli strings, indexed consistently with the ordering used in the PTM.

From these Pauli channels, we define a gate-level bias parameter by comparing the total weight of dephasing-type errors to that of all other non-identity Pauli errors. For two-qubit gates, we define
\begin{equation}
\eta_{\mathrm{gate}}
=
\frac{p_{IZ} + p_{ZI} + p_{ZZ}}{\sum_{E \neq I^{\otimes 2}} p_E - (p_{IZ} + p_{ZI} + p_{ZZ})}.
\end{equation}

Using this definition, we obtain the following gate-level bias values (for $\eta_{\mathrm{sys}}=100$):
CZ gates in trapped-ion and neutral-atom platforms preserve the imposed bias, yielding $\eta_{\mathrm{gate}}=100$, while CX implementations exhibit significantly reduced bias, with $\eta_{\mathrm{gate}}=1.13$ for trapped ions and $\eta_{\mathrm{gate}}=4.72$ for superconducting platforms.

Noise channels for single-qubit gates (Hadamard and phase) are computed using the same procedure to account for compilation into non-native gates. The explicit Pauli channels used in the simulations are listed in ~\cref{tab:pauli_channels_all}.

\begin{table}[t]
\centering
\footnotesize
\setlength{\tabcolsep}{4pt}
\renewcommand{\arraystretch}{1.05}

\begin{tabular}{@{} l r @{\hspace{1.2em}} l r @{\hspace{1.2em}} l r @{}}
\toprule
\multicolumn{2}{c}{\textbf{Single-qubit (all platforms)}} &
\multicolumn{2}{c}{\textbf{Native CX}} &
\multicolumn{2}{c}{\textbf{Native CZ}} \\
\cmidrule(lr){1-2}\cmidrule(lr){3-4}\cmidrule(lr){5-6}

\multicolumn{2}{l}{\textbf{Hadamard \(H\)}} &
\multicolumn{2}{l}{\textbf{Superconducting}} &
\multicolumn{2}{l}{\textbf{Trapped ions}} \\
\(I\) & 0.9997144182 & \(II\) & 0.9997120036 & \(II\) & 0.9996700778 \\
\(X\) & \(4.8340817\times 10^{-5}\) & \(IZ\) & \(7.1910099\times 10^{-5}\) & \(IZ\) & \(1.1097799\times 10^{-4}\) \\
\(Y\) & \(9.5198469\times 10^{-5}\) & \(ZI\) & \(7.1439037\times 10^{-5}\) & \(ZI\) & \(1.1097799\times 10^{-4}\) \\
\(Z\) & \(1.4204254\times 10^{-4}\) & \(ZZ\) & \(9.4325179\times 10^{-5}\) & \(ZZ\) & \(1.0469902\times 10^{-4}\) \\
      &                        & \(XY\) & \(2.3734056\times 10^{-5}\) & Other & \(\le 10^{-6}\) \\
      &                        & \(YX\) & \(2.3734056\times 10^{-5}\) &       & \\
      &                        & Other  & \(\le 10^{-6}\)              &       & \\
\cmidrule(lr){1-2}\cmidrule(lr){3-4}\cmidrule(lr){5-6}

\multicolumn{2}{l}{\textbf{Phase \(S\) (and \(S^\dagger\))}} &
\multicolumn{2}{l}{\textbf{Trapped ions}} &
\multicolumn{2}{l}{\textbf{Neutral atoms}} \\
\(I\) & 0.9996827966 & \(II\) & 0.9996826627 & \(II\) & 0.9997240150 \\
\(X\) & \(1.5707914\times 10^{-6}\) & \(ZI\) & \(1.0518675\times 10^{-4}\) & \(IZ\) & \(6.7489247\times 10^{-5}\) \\
\(Y\) & \(1.5707914\times 10^{-6}\) & \(XI\) & \(6.3515191\times 10^{-5}\) & \(ZI\) & \(9.3220474\times 10^{-5}\) \\
\(Z\) & \(3.1406184\times 10^{-4}\) & \(IZ\) & \(4.2060047\times 10^{-5}\) & \(ZZ\) & \(1.0251356\times 10^{-4}\) \\
      &                        & \(IY\) & \(3.1616011\times 10^{-5}\) & Other & \(\le 10^{-6}\) \\
      &                        & \(YI\) & \(3.2289131\times 10^{-5}\) &       & \\
      &                        & \(XY\) & \(1.0497059\times 10^{-5}\) &       & \\
      &                        & \(YX\) & \(1.0495067\times 10^{-5}\) &       & \\
      &                        & Other  & \(\le 10^{-6}\)              &       & \\
\bottomrule
\end{tabular}

\caption{Pauli error channels used in the simulations for \(p=3\times 10^{-4}\) and bias \(=100\). Left: single-qubit channels for \(H\) and \(S\) (platform independent). Middle: dominant two-qubit Pauli errors for native CX implementations. Right: dominant two-qubit Pauli errors for native CZ implementations. For all platforms, each omitted individual Pauli term satisfies \(p_E\le 10^{-6}\).}
\label{tab:pauli_channels_all}
\end{table}
\end{document}